\documentclass[10pt,journal]{IEEEtran}
\IEEEoverridecommandlockouts
\usepackage{cite}
\usepackage{amsmath,amssymb,amsfonts}
\usepackage{algorithmic}
\usepackage{graphicx}
\usepackage{textcomp}
\def\BibTeX{{\rm B\kern-.05em{\sc i\kern-.025em b}\kern-.08em
    T\kern-.1667em\lower.7ex\hbox{E}\kern-.125emX}}

\usepackage{url}
\usepackage{booktabs}
\usepackage{caption}
\usepackage{subfigure}
\usepackage{makecell}
\usepackage[ruled,linesnumbered]{algorithm2e}
\usepackage{setspace}
\usepackage{amsthm}
\newtheorem{theorem}{\textbf{Theorem}}
\newtheorem{assumption}{\textbf{Assumption}}
\newtheorem{corollary}{\textbf{Corollary}}
\newtheorem{remark}{\textbf{Remark}}
\newtheorem{lemma}{\textbf{Lemma}}
\usepackage{bm}

\makeatletter
\renewcommand\normalsize{%
\@setfontsize\normalsize\@xpt\@xiipt
\abovedisplayskip 3\p@ \@plus3\p@ \@minus5\p@
\abovedisplayshortskip \z@ \@plus3\p@
\belowdisplayshortskip 6\p@ \@plus3\p@ \@minus3\p@
\belowdisplayskip \abovedisplayskip
\let\@listi\@listI}
\makeatother

\begin{document}

\title{Energy-Efficient Channel Decoding for Wireless Federated Learning: Convergence Analysis and Adaptive Design}
\author{Linping Qu,~\IEEEmembership{Graduate Student Member,~IEEE,} Yuyi Mao,~\IEEEmembership{Senior Member,~IEEE,} \\Shenghui Song,~\IEEEmembership{Senior Member,~IEEE,} and Chi-Ying Tsui,~\IEEEmembership{Senior Member,~IEEE}

\thanks{Manuscript received 1 March 2024; revised 24 July 2024; accepted 3 August 2024. Part of this work was presented at the 2024 IEEE Wireless Communications and Networking Conference, Dubai, UAE \cite{qu2024robust}. This work was partially supported by AI Chip Center for Emerging Smart System (ACCESS) sponsored by InnoHK funding, Hong Kong, SAR, and a grant from the Research Grants Council of the Hong Kong Special Administrative Region, China (Project Reference Number: AoE/E-601/22-R), and a grant from the NSFC/RGC Joint Research Scheme sponsored by the Research Grants Council of the Hong Kong Special Administrative Region, China and National Natural Science Foundation of China (Project No. N\_HKUST656/22).

Linping Qu, Shenghui Song, and Chi-Ying Tsui are with the Department of Electronic and Computer Engineering, Hong Kong University of Science and Technology, Hong Kong (e-mail: lqu@connect.ust.hk; eeshsong@ust.hk; eetsui@ust.hk).
Yuyi Mao is with the Department of Electrical and Electronic Engineering, Hong Kong Polytechnic University, Hong Kong (e-mail: yuyi-eie.mao@polyu.edu.hk).}
}

\maketitle

\begin{abstract}
One of the most critical challenges for deploying distributed learning solutions, such as federated learning (FL), in wireless networks is the limited battery capacity of mobile clients. While it is a common belief that the major energy consumption of mobile clients comes from the uplink data transmission, this paper presents a novel finding, namely channel decoding also contributes significantly to the overall energy consumption of mobile clients in FL. Motivated by this new observation, we propose an energy-efficient adaptive channel decoding scheme that leverages the intrinsic robustness of FL to model errors. In particular, the robustness is exploited to reduce the energy consumption of channel decoders at mobile clients by adaptively adjusting the number of decoding iterations. We theoretically prove that wireless FL with communication errors can converge at the same rate as the case with error-free communication provided the bit error rate (BER) is properly constrained. An adaptive channel decoding scheme is then proposed to improve the energy efficiency of wireless FL systems. Experimental results demonstrate that the proposed method maintains the same learning accuracy while reducing the channel decoding energy consumption by $\sim$ 20\% when compared to an existing approach.
\end{abstract}


\section{Introduction}

\subsection{Motivation}
Deep learning has achieved remarkable successes in numerous realms, such as computer vision and natural language processing \cite{lecun2015deep}. However, the conventional centralized learning paradigm encounters the data isolation issue due to the rising concern about privacy \cite{yang2019federated}. As a result, uploading the distributed data, which may contain privacy-sensitive information, for centralized learning becomes prohibited. To address this issue, federated learning (FL) was proposed \cite{mcmahan2017communication} to enable the distributed training without sharing the data of users (e.g., IoT devices).
However, to fully unlock the potential of FL, we need to tackle several critical challenges, such as data and device heterogeneity \cite{chai2019towards,xie2023fedkl}, as well as the communication bottleneck\cite{chen2021communication}. Another issue of wireless FL is the energy consumption of mobile clients due to heavy computation and communication workloads for training deep neural networks (DNNs) and transferring model parameters\cite{mao2023green}.

Many innovative approaches have been proposed to reduce the energy consumption of mobile clients in wireless FL from different perspectives. For instance, reinforcement learning was utilized in \cite{kim2021autofl} to determine the near-optimal client selection and resource allocation scheme to reduce both the convergence time and energy consumption of mobile clients. In \cite{tran2019federated}, the trade-off between learning accuracy and total energy consumption of clients was optimized through computational and communication resource allocation. In \cite{yang2020energy}, the transmission and computation energy was minimized under a latency constraint by optimizing the time, bandwidth allocation, power control, computation frequencies, and learning accuracy. Also, authors of \cite{chen2020joint} minimized the transmission energy consumption of mobile clients while maintaining the learning performance by further designing the client selection strategies.

Besides system-level resource management, techniques such as model sparsification \cite{han2020adaptive} and quantization \cite{zheng2020design, qu2022feddq} have been proposed to minimize the communication overhead, thus reducing the energy consumption of clients' transmitters. Both quantization and periodic aggregation were adopted in \cite{reisizadeh2020fedpaq} to reduce the communication overhead and energy consumption, where a theoretical framework about the impact of parameter quantization on the convergence rate was also provided. Model sparsification and quantization techniques are also beneficial for minimizing the DNN processing energy consumption at mobile clients of wireless FL \cite{qiu2022zerofl, kim2023green}. 
However, most works on energy-efficient FL focused on the energy consumption of mobile processors and transmitters, while few of them paid attention to the energy consumption of channel decoders in mobile receivers. As will be shown in Section II-B, channel decoders may consume a substantial proportion of energy because of the iterative processes to decode massive DNN parameters in each communication round of wireless FL.

To address the research gap, in this work, we first quantitatively analyze the energy consumption of different components in a mobile unit and show that the channel decoder contributes significantly to the overall energy consumption. Building on this observation, an energy-efficient adaptive channel decoding scheme is proposed for wireless FL systems, by leveraging the fact that FL is robust to model errors due to its statistical learning nature\cite{tong2022nine}. 
Furthermore, by extending the convergence analysis in\cite{reisizadeh2020fedpaq} to scenarios with downlink communication erros, we formalize the relationship between the bit error rate (BER) and the convergence performance of wireless FL. A decreasing rule of BER with respect to communications rounds is determined, which can maintain the same convergence rate as the case with error-free communication. Finally, we propose a practical approach to adjust the number of channel decoding iterations in each communication round to achieve the desired BER, while reducing the energy consumption of mobile clients in wireless FL.

\subsection{Contributions}
This paper investigates the impact of channel decoding on the energy efficiency of wireless FL. Our major contributions are summarized as follows:
\begin{itemize}
    \item We theoretically analyze the relationship between BER and the convergence performance of wireless FL. Specially, an upper bound of the average expected gradient norm of the global training loss function with respect to BERs over communication rounds is derived by extending the analysis in \cite{reisizadeh2020fedpaq}. Based on the analysis, a decreasing rule of BERs is proposed to achieve the same convergence rate as FL with error-free communication.
    \item With the decreasing rule of BERs, an adaptive low-density parity checking (LDPC) decoding mechanism is proposed, which controls the maximum number of decoding iterations to achieve the desired BER requirement in every communication round. This provides a practical approach to maintain the target convergence rate, while reducing the channel decoding energy consumption of mobile clients in wireless FL.
    \item Experimental results on both IID and non-IID datasets show that with a given number of communication rounds, the proposed adaptive LDPC decoding scheme effectively reduces the energy consumption of channel decoders by 20\% when compared with a baseline scheme adopting a fixed maximum number of channel decoding iterations, while enjoying the same learning performance.
\end{itemize}

\subsection{Organization}
The remainder of this paper is organized as follows. Section \ref{System Model and Preliminaries} introduces the system model and presents a comprehensive energy consumption analysis for mobile clients. Section \ref{Energy Efficient LDPC Decoding Design} develops an energy-efficient adaptive LDPC decoding scheme for wireless FL systems by exploiting the relationship between the convergence rate and the achieved BER of LDPC decoding.  Experimental results are presented in Section \ref{Expirements}. Finally, conclusions are drawn in Section \ref{Conclusion}. 

\section{System Model and Energy Consumption Evaluation}
\label{System Model and Preliminaries}
In this section, we first introduce the wireless FL system. An overview regarding the energy consumption of different components at a mobile unit (i.e., client) is then presented, which highlights the significant contribution of channel decoding to the total energy consumption of wireless FL. For ease of reference, we list the key notations in Table~\ref{tab:notations}.

\begin{table}[t]
\small
\caption{Key Notations}
\centering
\begin{tabular}{|l|p{6.9cm}|}
\hline
{\textbf{Notation}}&{\textbf{Definition}}\\
\hline
{$\mathbf{w}_{r}$}&{Global model in the $r$-th communication round}\\
\hline
{$\Delta \mathbf{w}_{r}^{k}$}&{Local model update of client $k$ at the $r$-th communication round}\\
\hline
{$f\left(\cdot\right)$}&{The global loss function}\\
\hline
{$f^*$}&{Minimum global training loss}\\
\hline
{$R$}&{Number of communication rounds}\\
\hline
{$E$}&{Number of local training steps in each communication round}\\
\hline
{$\sigma_L^2$}&{Variance bound of local gradient estimator}\\
\hline
{$\sigma_G^2$}&{Dissimilarity bound between local and global gradients}\\
\hline
{$K$}&{Number of clients}\\
\hline
{$N$}&{Each global model parameter is digitalized as an $N$-bit sequence}\\
\hline
{$D$}&{Dimension of model $\mathbf{w}$}\\
\hline
{$T$}&{Total number of local training steps, $T=RE$}\\
\hline
{$L$}&{$L$-smoothness parameter}\\
\hline
{$b_r$}&{Bit error rate in the $r$-th communication round}\\
\hline
{$\eta$}&{Learning rate}\\
\hline
{$Q_r$}&{Maximum number of LDPC decoding iterations in the $r$-th communication round}\\
\hline
{$M$}&{Bound of $\max(\mathbf{w}_{r})-\min(\mathbf{w}_{r}), r = 0, \cdots, R-1$}\\
\hline
\end{tabular}
\label{tab:notations}
\end{table}

\begin{figure*}[t]
\small
\centering
\subfigure[Wireless FL system: Multiple mobile clients are coordinated by a remote server through noisy wireless channels, where model updates and model weights are communicated in the uplink and downlink, respectively.]{
\begin{minipage}[b]{0.4\linewidth}
\includegraphics[width=1\linewidth]{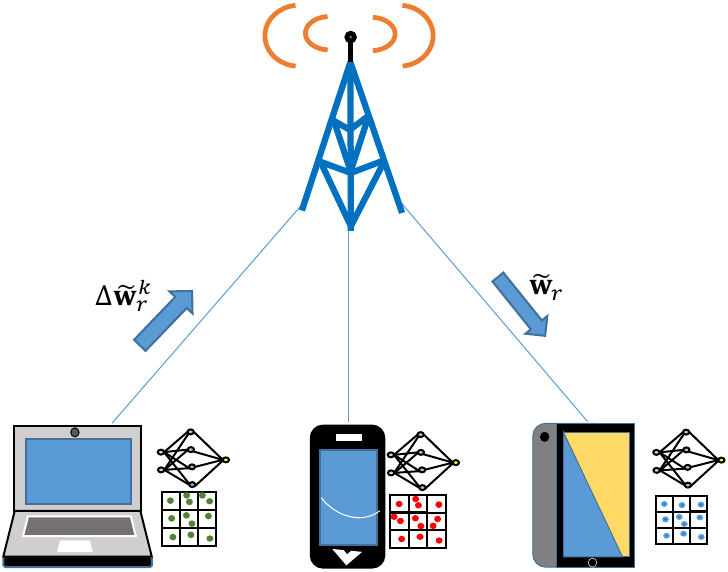}
\end{minipage}}
\qquad
\qquad
\subfigure[System operations: The server performs model aggregation, digitalization, LDPC encoding, and model broadcasting. Mobile clients perform LDPC decoding, local training, and upload the model updates. Procedures of uploading model updates are omitted.]{
\begin{minipage}[b]{0.4\linewidth}
\includegraphics[width=1\linewidth]{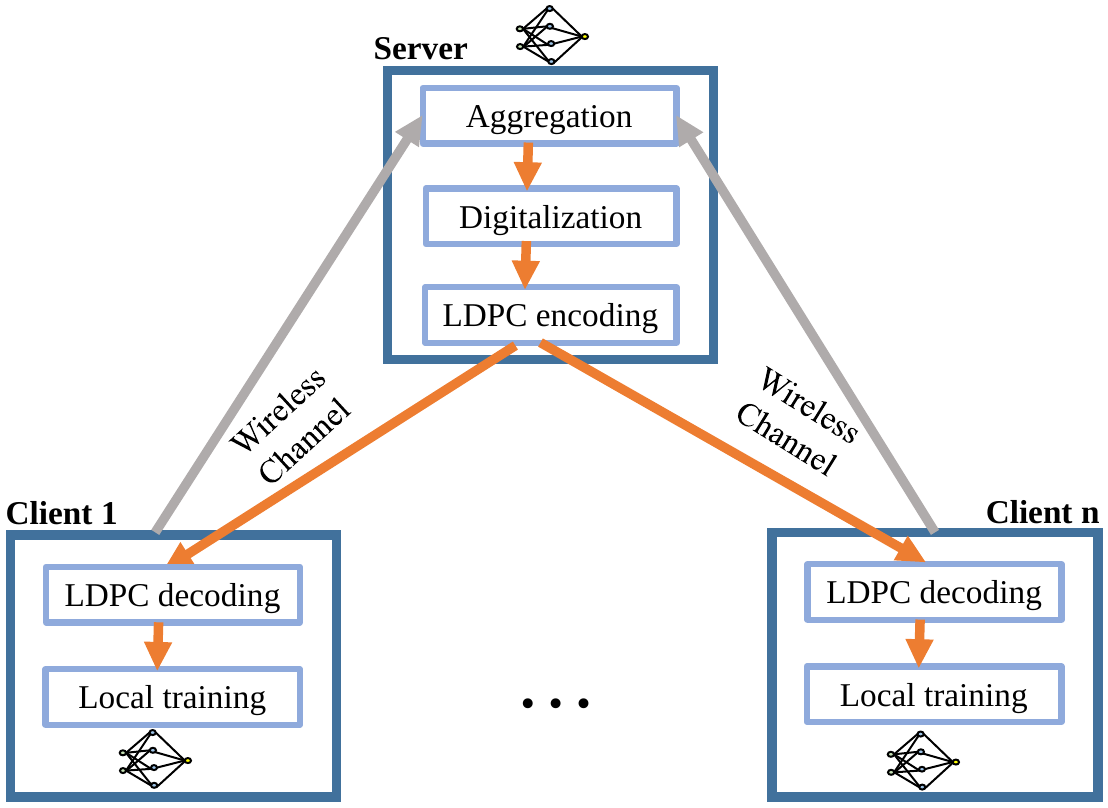}
\end{minipage}}
\caption{Wireless FL with digital communication and its operations.}
\label{Wireless FL system}
\end{figure*}

\subsection{Wireless FL}
Fig.~\ref{Wireless FL system}(a) depicts a wireless FL system, where model training of different clients is coordinated by a remote server. The FL problem can be formulated as follows \cite{mcmahan2017communication}:
\begin{equation} 
\min_{\mathbf{w}\in {\mathbb{R}^D}} f(\mathbf{w}) \triangleq \frac{1}{K}\sum_{k=1}^K f_k(\mathbf{w}),
\end{equation}where $\mathbf{w}\in {\mathbb{R}^D}$ denotes the DNN model with $D$ parameters, $K$ represents the number of mobile clients, $f_{k}\left(\cdot\right)$ is the local loss function of the $k$-th client, and $f\left(\cdot\right)$ is the global loss function.

\begin{figure}[t]
	\small
	\centering
	\includegraphics[width=2.5in]{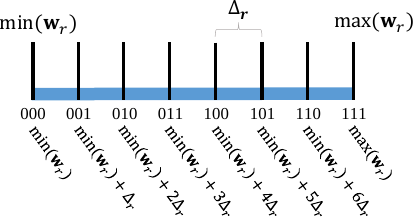}
	\caption{Digitalization for each parameter in $\tilde{\mathbf{w}}_{r}$ (an example with $N=3$), where a continuous model parameter is mapped to the binary sequence of the nearest boundary.}
	\label{fig:digit}
\end{figure}

To train a shared global model, the server communicates with the clients for multiple rounds through noisy wireless channels. As shown in Fig.~\ref{Wireless FL system}(b), in each communication round, the global model at the server (denoted as $\mathbf{w}_{r}$ for the $r$-th communication round) is broadcast to the clients. Each client then performs local training based on the received global model and transmits the model updates after local training to the server for aggregation.
Specifically, in the $r$-th communication round, the global model is digitalized by an $N$-bit sequence, encoded by a channel code (e.g., the LDPC codes), and broadcast to clients. The parameter digitalization scheme is illustrated in Fig.~\ref{fig:digit}, where the range of [$\min(\mathbf{w}_{r}), \max(\mathbf{w}_{r})$] is evenly divided into $2^N-1$ intervals and the width of each interval is $\Delta_r=\frac{\max(\mathbf{w}_{r})- \min(\mathbf{w}_{r})}{2^N-1}$. Each of the $2^{N}$ boundaries of these intervals are marked by a $N$-bit binary sequence. A parameter of $\mathbf{w}_{r}$ is mapped to the binary sequence of the nearest boundary. We assume $N$ is large enough so the continuous and digital representations of model parameters have negligible difference. The model parameters received/decoded by the $k$-th client at the $r$-th communication round is given as follows: 
\begin{align}
&\mathbf{w}^{k}_{r,0}=\mathbf{\tilde{w}}_r,\ k = 1,\cdots, K,\ r =0,\cdots,R-1,
\label{Bw}
\end{align}
where $\mathbf{\tilde{w}}_r$ denotes the distorted version of $\mathbf{w}_r$ due to channel decoding errors. To reconstruct $\tilde{\mathbf{w}}_{r}$ from the decoded bits at clients, the server also broadcasts the scalars $\min(\mathbf{w}_r)$ and $\max(\mathbf{w}_{r})$ through an error-free downlink channel. The impact of channel decoding errors can be quantified by considering the different error bit positions in the $N$-bit representations of model parameters, which will be analyzed in Section~\ref{Energy Efficient LDPC Decoding Design}-A.

Each client updates the received global model $\tilde{\mathbf{w}}_{r}$ by conducting $E$ steps of local mini-batch stochastic gradient descent (SGD) with a learning rate $\eta$ as follows:
\begin{equation}
\begin{split}
\mathbf{w}^{k}_{r,e+1}=\mathbf{w}^{k}_{r,e}-\eta&{\nabla} f_k(\mathbf{w}^k_{r,e};\xi_{r,e}^k),\\
&e = 0, \cdots, E-1,\ k=1,\cdots, K,
\end{split}
	\label{receive}
\end{equation}
where ${\nabla} f_k(\mathbf{w}^{k}_{r,e};\xi_{r,e}^{k})$ denotes the gradient computed from a mini-batch datasets $\xi_{r,e}^{k}$ at the $k$-th client.
By completing $E$ steps of local training, the local model update of the $k$-th client can be expressed as follows:
\begin{align}
\begin{split}
\Delta \mathbf{w}^{k}_r=\mathbf{w}^{k}_{r,E}-&\mathbf{w}^{k}_{r,0},\\
&r = 0, \cdots, R-1,\ k=1,\cdots, K,   
\end{split}
\end{align}
which is also digitalized, encoded, and uploaded to the server. For simplicity, we assume an orthogonal uplink channel is available for each client. After receiving these local updates, denoted as $\Delta \mathbf{\tilde{w}}^k_r$, the server performs model aggregation to update the global model as follows:
{\begin{equation}
		\mathbf{w}_{r+1}=\mathbf{w}_{r}+\frac{1}{K}\sum_{k=1}^K\Delta \mathbf{\tilde{w}}^{k}_r,\ r= 0, \cdots, R-1.
\label{server_update}
\end{equation}}

In this work, we focus on the downlink communication rather than the uplink since our objective is to design novel channel decoding schemes to save the energy consumption of mobile clients. For the uplink, the channel decoding happens at the server, where a powerful channel decoder as well as a sufficient number of decoding iterations can always be applied to achieve satisfactory BER performance. Hence, we consider $\Delta \mathbf{\tilde{w}}_{r}^{k} = \Delta \mathbf{w}_r^k$ in the following. Define $\bar{\mathbf{w}}_{r,e} \triangleq \frac{1}{K} \sum_{k=1}^{K} \mathbf{w}_{r,e}^{k}$ as the average of client models after $e$ local training steps in the $r$-th communication round. After $R$ communication rounds, the wireless FL system obtains $\bar{\mathbf{w}}_{R,E}$ as the learned model.

\subsection{Energy Consumption Evaluation}
\label{Energy Consumption Evaluation}
There are two main energy-consuming operations for clients participating in wireless FL, namely model training and model communication. For the latter, it is believed that the energy is mostly consumed by the data transmission. In this subsection, we evaluate the energy consumption of different components, taking the training of AlexNet for image classification tasks as an example \cite{krizhevsky2012imagenet}. The analysis is based on actual measurements obtained from fabricated chips published in recent years~\cite{aichip,trx,ldpc}. For fair comparisons, all data are collected for chips fabricated in the 65~nm complementary metal-oxide semiconductor (CMOS) process.

In\cite{aichip}, a fabricated AI hardware accelerator was reported to achieve the energy efficiency of 13.7~mJ/epoch for 8-bit quantized AlexNet. With five epochs of local training, the processing energy consumption in each communication round is 68.5~mJ. In \cite{trx}, a fifth-generation (5G) New Radio wireless transmitter and receiver were reported to achieve the energy efficiencies of 47.9~pJ/bit and 33.3~pJ/bit for transferring uncoded data, respectively, under a communication distance of around 200 meters. Therefore, the energy required to transmit and receive the 8-bit AlexNet with 60 million parameters in a single communication round is about 39~mJ. Assuming a code rate of 1/2, the communication energy consumption for coded transmission of the AlexNet model parameters in each round is 78~mJ. Besides, the LDPC decoder in \cite{ldpc} achieves the energy efficiency of 20.1~pJ/bit/iteration, which was measured with 10 decoding iterations on average. Accordingly, the energy consumption of the LDPC decoder for decoding the AlexNet model parameters is 96.5~mJ. The detailed comparison is presented in Table \ref{tab:energy}, which indicates that the energy consumption of LDPC decoding accounts for a large proportion of the overall energy consumption at clients in wireless FL.

\begin{table}[t]
	\small
	\caption{Energy Consumption of Main Components at a Mobile Unit in Wireless FL for Training an 8-bit AlexNet}
	\centering
	\begin{tabular}{| c | c | c |}
		\hline
		\textbf{Component} & \textbf{Energy efficiency} & \textbf{Energy consumption}\\
		\hline
		AI Chip& 13.7~mJ/epoch\cite{aichip} & 68.5~mJ\\
		\hline
		RF Transceiver& 81.2~pJ/bit\cite{trx}& 78~mJ\\
		\hline
		LDPC Decoder& 20.1~pJ/bit/iter.\cite{ldpc} & 96.5~mJ\\
		\hline 
	\end{tabular}
	\label{tab:energy}
\end{table}

\subsection{LDPC Decoding}
\label{LDPC}

In this paper, we take the LDPC codes as an example to illustrate how adaptive channel decoding can help reduce the energy consumption of wireless FL. In LDPC decoders, the iterative decoding process continues until $Q$ decoding iterations, which is termed as the maximum number of decoding iterations, are completed~\cite{koike2015iteration}. Such iterative LDPC decoding processing leads to high energy consumption. To save the decoding energy, early termination techniques are usually adopted. Specifically, if the decoded codewords can fully pass the parity check, the decoding process will stop even if not all the $Q$ decoding iterations have been completed\cite{nam2021early}. Therefore, early termination can save some decoding iterations without compromising the BER performance. The energy consumption of LDPC decoding depends on the target BER and thus the pre-defined maximum number of decoding iterations. As a result, to minimize the decoding energy consumption, it is essential to determine suitable BERs over communication rounds that can be tolerated by wireless FL systems, considering its unique statistical learning nature.

\section{Energy-Efficient LDPC Decoding Design}
\label{Energy Efficient LDPC Decoding Design}
In this section, we propose an energy-efficient LDPC decoding scheme to minimize the decoding energy consumption while guaranteeing the same convergence rate for wireless FL as the case with error-free communication. For that purpose, we explore the relationship between BER and convergence rate, and that between the maximum number of decoding iterations and BER. By connecting these two relationships, we demonstrate how to regulate the LDPC decoding scheme.

\subsection{Effect of BER on the Convergence Rate of Wireless FL}

With error-free communication, the convergence rate of FL with non-convex loss functions was proved as $\mathcal{O}(1/\sqrt{T})$ \cite{reisizadeh2020fedpaq}, where $T=RE$ is the total number of local training steps. In this subsection, we investigate how BER affects the convergence rate of FL and determine the admissible BER for maintaining the convergence rate of $\mathcal{O}(1/\sqrt{T})$.
Unlike FL with analog communication where model distortion is considered as additive noise\cite{ang2020robust, amiri2021convergence, amiri2020federated, wei2022federated, sery2021over,sun2023channel}, quantifying the impact of BER on FL with digital communication is not straightforward. In particular, we need to consider different cases of bit errors and analyze the impact of different bit errors, e.g., most significant bits and least significant bits. 

To conduct the theoretical analysis, we take the following assumptions that are widely adopted in literature.

\begin{assumption}
($L$-Smoothness~\cite{zheng2020design,reisizadeh2020fedpaq}) The loss function $f\left(\cdot\right)$ is $L$-smooth, i.e., $\forall \mathbf{w}, \mathbf{v}$, we have $||\nabla f(\mathbf{w})-\nabla f(\mathbf{v})||\leq L||\mathbf{w}-\mathbf{v}||$.
\label{assumption_smoothness}
\end{assumption}

\begin{assumption}
(Unbiased and Variance-bounded Local Gradient Estimator~\cite{yang2021achieving}) Let $\xi_{k}$ denote the batch of randomly sampled data at client $k$. The gradient estimator $\nabla f_{k}(\mathbf{w};\xi_{k})$ is unbiased, i.e., $\mathbb{E}_{\xi_{k}}[\nabla f_{k}(\mathbf{w};\xi_{k})]=\nabla f_{k}(\mathbf{w})$, and has bounded variance, i.e., $\mathbb{E}_{\xi_{k}}[||\nabla f_k(\mathbf{w};\xi_{k})-\nabla f_k(\mathbf{w})||^2] \leq \sigma_L^2$.
\label{assumption_unbiased}
\end{assumption}

\begin{assumption}
(Bounded Local and Global Gradient Dissimilarity~\cite{yang2021achieving}) The dissimilarity between the local and global gradients is bounded, i.e., $||\nabla f_k(\mathbf{w})-\nabla f(\mathbf{w})||^2  \leq \sigma_G^2$.
\label{assumption_dissimilarity}
\end{assumption}

\begin{assumption}
(Independent Bit Errors \cite{morrow1989bit}) \textit{Each bit in the bitstream of the model parameters encounters error (i.e., bit flipping) independently.}
\label{assumption_ber}
\end{assumption}

With the $N$-bit representation, each model parameter has $2^{N}-1$ possible distorted patterns. However, our empirical results in Fig.~\ref{fig:validation} show that it is sufficient to focus on cases with at most one bit error of each model parameter. In particular, Fig.~\ref{fig:validation} compares the actual model error with the analytical model error expression to be presented in (\ref{model_error}) assuming at most 1-bit error, which match closely under various BER conditions. Therefore, we also impose the following assumption to simplify our analysis.
	
\begin{assumption}
\textit{There is at most 1-bit error in the $N$-bit representation of each model parameter.}
\label{onebit_error}
\end{assumption}

Based on Assumptions \ref{assumption_ber} and \ref{onebit_error}, we are able to establish the relationship between the decoded model at a client and the broadcast global model in the following lemma.

\begin{lemma}
Denote $\tilde{\mathbf{w}}$ as the decoded model of $\mathbf{w}$. Based on Assumptions \ref{assumption_ber} and \ref{onebit_error}, the $d$-th dimension of $\tilde{\mathbf{w}}$ is given as
\begin{equation}
\tilde{w}_{d} = w_{d} + \delta \left(w_{d}\right),
\end{equation}
where $w_{d}$ is the $d$-th dimension of $\mathbf{w}$ and $\delta\left(w_{d}\right)$ is the distortion term.
The mean square model error caused by bit errors is given as
\begin{equation}
\mathbb{E}[||\mathbf{\tilde{w}}-\mathbf{w}||^2|\mathbf{w}]= \frac{D\left(4^N-1\right)}{3\left(2^N-1\right)^2}\cdot b(1-b)^{N-1} M_{\mathbf{w}}^2,
\label{model_error}
\end{equation}
where $b$ is the BER and $M_{\mathbf{w}} \triangleq \max(\mathbf{w})-\min(\mathbf{w})$.
\label{lemma_ber_distortion}
\end{lemma}
\begin{proof}
Please refer to Appendix A.
\end{proof}

\begin{figure}[t]
\small
\centering
\includegraphics[width=2.5in]{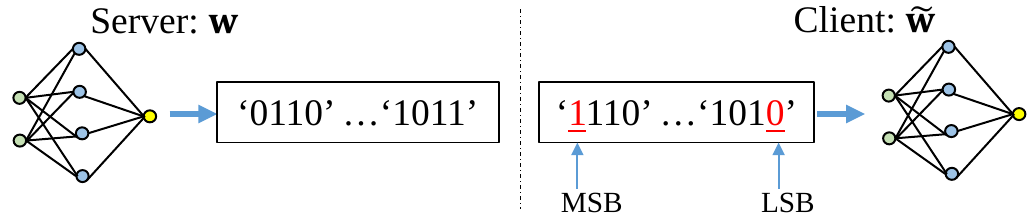}
\caption{Distortion between the decoded and broadcast global model. Different error bits, e.g., the most significant bit (MSB) and least significant bit (LSB), lead to different model distortion.}
\label{fig:communication errors}
\end{figure}

\begin{figure*}[t]
\small
\centering
\subfigure[$\text{BER} = 10^{-2}$.]{
\begin{minipage}[b]{0.3\linewidth}
\includegraphics[width=1\linewidth]{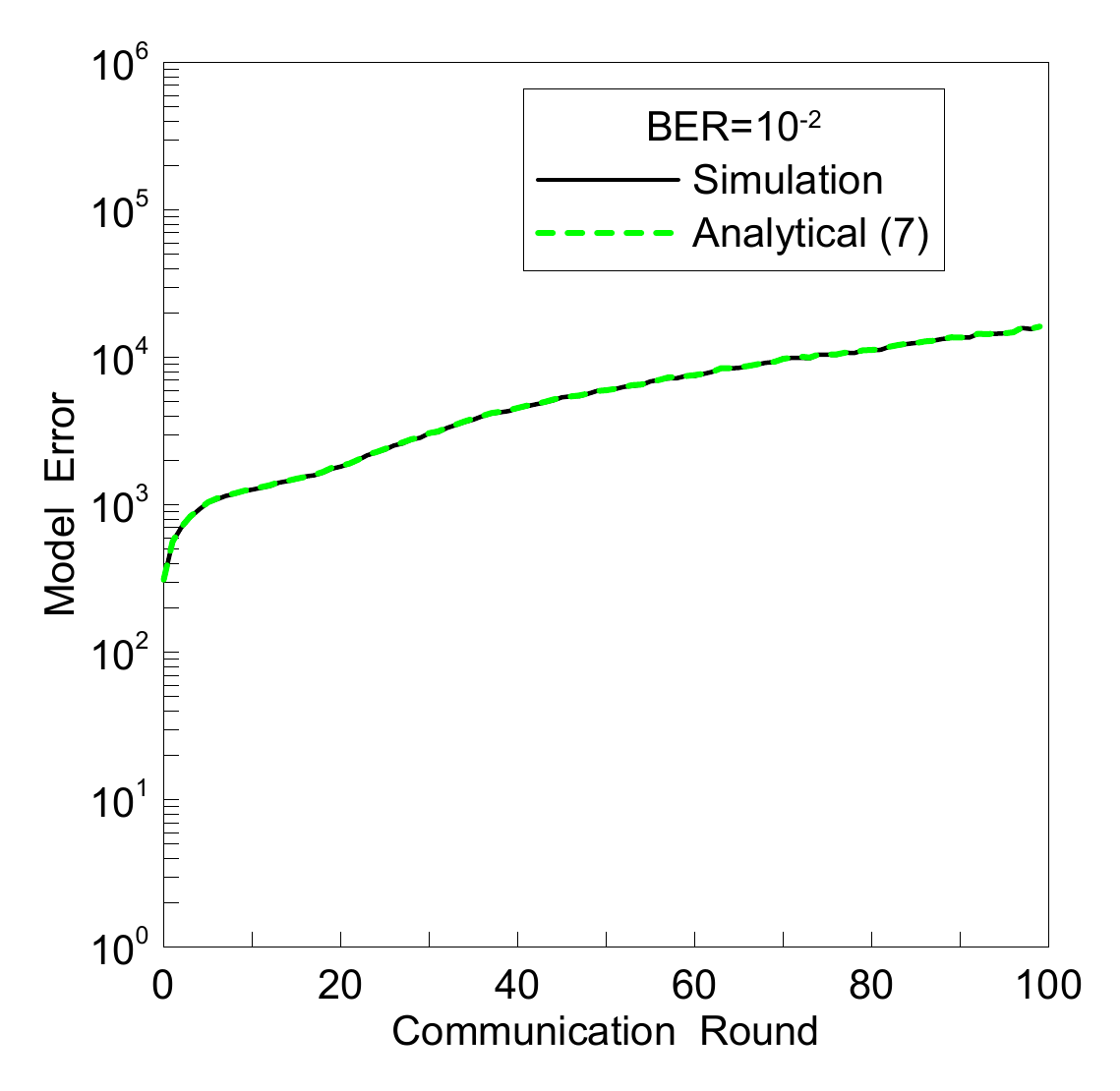}
\end{minipage}}
\quad
\subfigure[$\text{BER} = 10^{-3}$.]{
\begin{minipage}[b]{0.3\linewidth}
\includegraphics[width=1\linewidth]{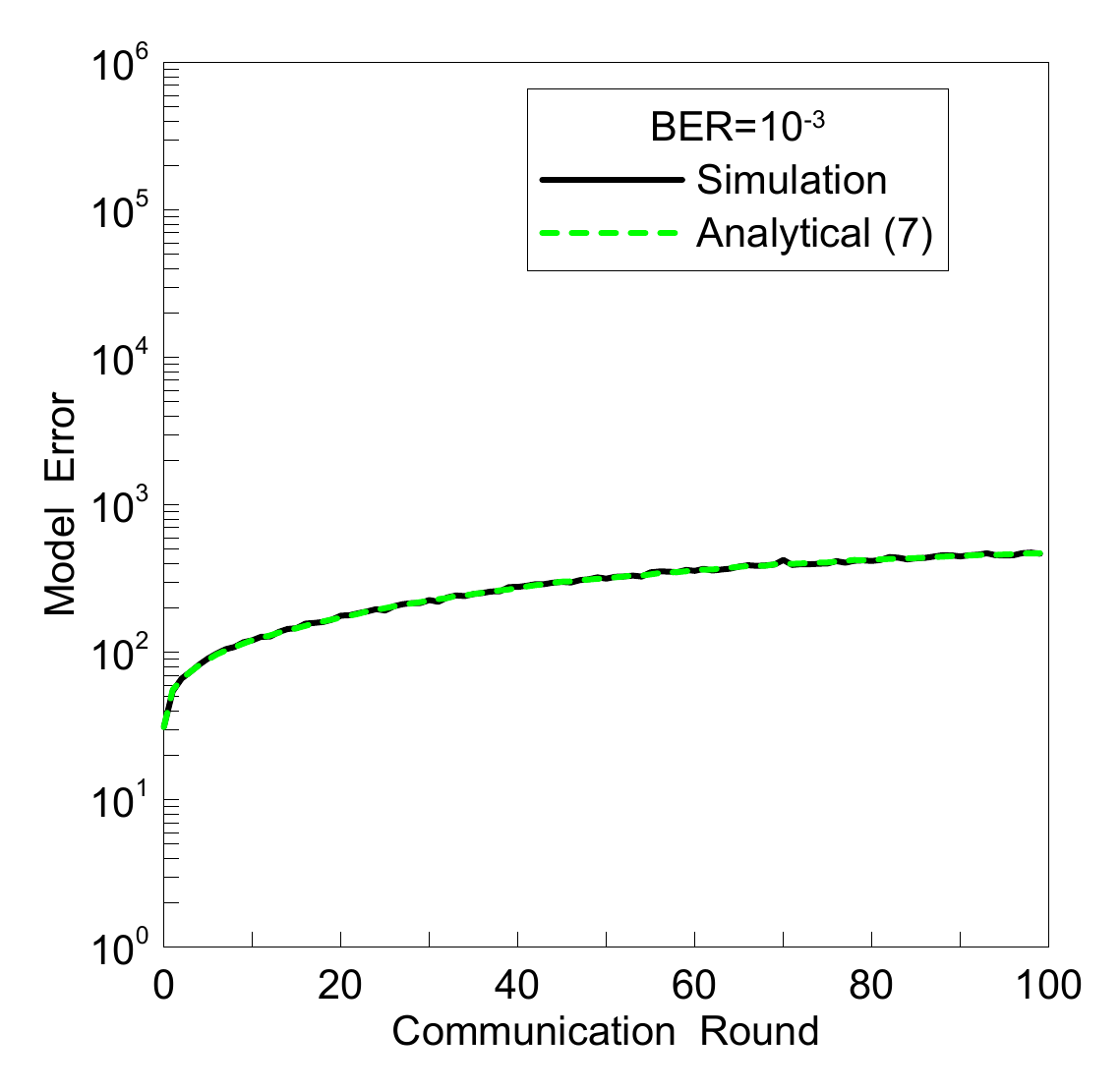}
\end{minipage}}
\quad
\subfigure[$\text{BER} = 10^{-4}$.]{
\begin{minipage}[b]{0.3\linewidth}
\includegraphics[width=1\linewidth]{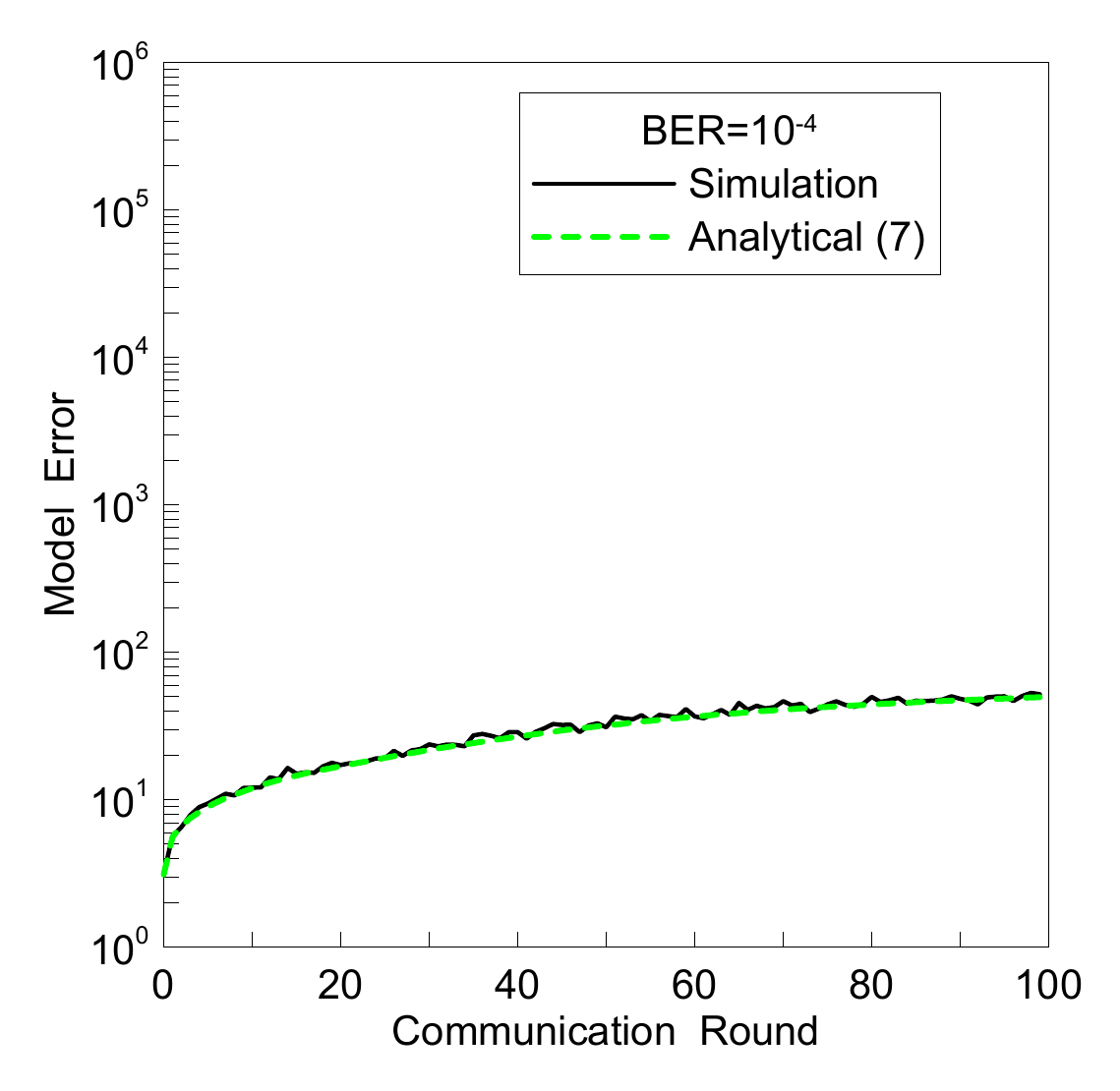}
\end{minipage}}
\caption{Validation of the model error analysis in (\ref{model_error}) by experiments on the Fashion-MNIST dataset. The experimental setting is detailed in Section~\ref{ex_setting}. In each communication round, we compute the actual model error between the decoded and broadcast global model on a selected client, and compare with the value calculated by the RHS in (\ref{model_error}) with $\mathbf{w}$ setting as $\mathbf{w}_{r}$. It is observed that the analytical and simulation results match closely for different BER conditions. The model errors also show a trend of boundness, which demonstrates the rationality of Assumption~\ref{range_bound}.}
\label{fig:validation}
\end{figure*}

\noindent Note that model distortion caused by bit errors $\mathbb{E}\left[\delta(w_d)|\mathbf{w}\right]$ does not necessarily equal zero. For example, any bit error of an all-zero sequence always leads to positive model distortion. Such non-zero bias, which was overlooked in a related analysis \cite{you2023broadband}, is carefully tackled in the convergence analysis in this paper.

After understanding the effect of BER, we then analyze the convergence property of wireless FL with communication error, for which, the following boundness assumption on the global model parameters \cite{zheng2020design} is needed. The rationality of this assumption can be evident in Fig.~\ref{fig:validation}. From this figure, it can be observed that the model error is bounded under various BER conditions, and thus $M_{\mathbf{w}}$ is also bounded according to (\ref{model_error}).
\begin{assumption}
(Boundness of Model Parameters \cite{zheng2020design}) For every communication round $r=0,1,\cdots, R-1$, $M_{r} \triangleq \max(\mathbf{w}_{r}) - \min(\mathbf{w}_{r}) \leq M$.
\label{range_bound}
\end{assumption}

\begin{lemma}
With Assumptions \ref{assumption_smoothness}-\ref{range_bound} and if the learning rate $\eta$ satisfies
\begin{equation}
1-\frac{K+L}{K}L\eta-2L^2\eta^2E(E-1) \geq 0,
\label{smalleta}
\end{equation}
the average expected gradient norm of $\{\nabla f(\mathbf{\bar{w}}_{r,e})\}$ is bounded as follows:
\begin{equation}
\begin{split}
&\frac{1}{RE}\sum_{r=0}^{R-1}\sum_{e=0}^{E-1}\mathbb{E}\Vert \nabla f(\mathbf{\bar{w}}_{r,e})\Vert^2 \leq \frac{2(\mathbb{E} f(\mathbf{w}_0)-f^*)}{\eta RE} \\
&+\frac{(L+1)^2DM^2(4^N-1)}{3(2^N-1)^2\eta RE}\sum_{r=0}^{R-1}b_r(1-b_r)^{N-1} \\
&+\frac{L(L+1)\eta(\sigma_L^2+\sigma_G^2)}{K}+\frac{L^2\eta^2\sigma_L^2(K+1)(E-1)}{K},
\end{split}
\label{convergence_ber_bound}
\end{equation}
where $f^{*}$ is the minimum global training loss, $M$ is defined in Assumption~\ref{range_bound}, $b_{r}$ is the BER in the $r$-th communication round, and $\mathbf{w}_{0}$ is the initial global model.
\label{thm_converge_bound}
\end{lemma}
\begin{proof}
Please refer to Appendix B.
\end{proof}
\begin{remark}
It is worthy noting that the first term on the RHS of (\ref{convergence_ber_bound}) captures the distance between the initial and minimum global training loss. The second term indicates how bit errors in each communication round affect convergence. The last two terms stem from non-IID data among clients and the random mini-batch sampling.\\
\end{remark}
We next simplify (\ref{convergence_ber_bound}) by applying $\eta =\frac{1}{L\sqrt{RE}}=\frac{1}{L\sqrt{T}}$, which can satisfy (\ref{smalleta}) with a sufficiently large value of $R$.

\begin{theorem}
With Assumptions \ref{assumption_smoothness}-\ref{range_bound} and a sufficiently large value of $R$ such that $\eta =\frac{1}{L\sqrt{RE}}$ satisfies (\ref{smalleta}), (\ref{convergence_ber_bound}) can be simplified as follows:
\begin{equation}
\begin{split}
\frac{1}{T}&\sum_{r=0}^{R-1}\sum_{e=0}^{E-1}\mathbb{E}\Vert \nabla f(\mathbf{\bar{w}}_{r,e})\Vert^2 \\ \leq & \frac{C_0}{\sqrt{T}} + \frac{C_1\sum_{r=0}^{R-1} b_r(1-b_r)^{N-1}}{\sqrt{T}} +\frac{C_2}{\sqrt{T}} +\frac{C_3}{T},
\end{split}
\label{our_T}
\end{equation}
where $C_0 \triangleq 2L(\mathbb{E} f(\mathbf{w}_0)-f^*)$, $C_1\triangleq \frac{L(L+1)^2DM^2(4^N-1)}{3(2^N-1)^2}$, $C_2 \triangleq \frac{(L+1)(\sigma_L^2+\sigma_G^2)}{K}$, $C_3 \triangleq \frac{L^{2}\sigma_L^2(K+1)(E-1)}{K}$.
\label{thm_converge}
\end{theorem}
According to \textbf{Theorem~\ref{thm_converge}}, in order to guarantee the convergence rate of $\mathcal{O}(1/\sqrt{T})$, it is crucial to control the BER in each communication round such that the RHS of (\ref{our_T}) reduces to zero at a rate of $\mathcal{O}(1/\sqrt{T})$. Since the explicit relationship between $b_r$ and $r$ is difficult to obtain, we propose a simple yet practical BER regulation method in \textbf{Corollary~\ref{coro_ber_regulation}}. As such, wireless FL systems can achieve the same convergence rate as those with error-free communication.
\begin{corollary}
If the BER in each communication round is regulated as
\begin{align}
b_r \propto \frac{1}{(r+1)^2}, r = 0, \cdots, R-1,
\label{BERr}
\end{align}
the RHS of (\ref{our_T}) diminishes to zero at a rate of $\mathcal{O}(1/\sqrt{T})$, i.e., wireless FL converges to a local optimal solution of the global training loss function at a rate of $\mathcal{O}(1/\sqrt{T})$.
\label{coro_ber_regulation}
\end{corollary}

\begin{proof}
Denote $b_r \propto \frac{1}{(r+1)^2} \triangleq \frac{\vartheta}{(r+1)^2}$. To prove the RHS of (\ref{our_T}) diminishes to zero at a rate of $\mathcal{O}(1/\sqrt{T})$, it is sufficient to verify
\begin{equation}
\lim_{T\rightarrow \infty} \frac{\sum_{r=0}^{R-1}b_r(1-b_r)^{N-1}}{\sqrt{T}} = 0.
\label{colla_limit_1}
\end{equation}
Specifically, since
\begin{equation}
\begin{split}
0 & \leq \frac{\sum_{r=0}^{R-1}b_r(1-b_r)^{N-1}}{\sqrt{T}} \leq \frac{\sum_{r=0}^{R-1}b_r}{\sqrt{T}} = \vartheta \frac{\sum_{r=0}^{R-1}1 \slash (r+1)^2}{\sqrt{T}}\\
& = \frac{\vartheta}{\sqrt{T}}\left(1+\sum_{r=1}^{R-1} \frac{1}{(r+1)^2}\right) \\
& < \frac{\vartheta}{\sqrt{T}} \left(1+\sum_{r=1}^{R-1} \frac{1}{r(r+1)}\right) = \frac{\vartheta}{\sqrt{T}}  \left(2 - \frac{1}{R} \right)\\
& = \frac{\vartheta}{\sqrt{T}}  \left(2 - \frac{E}{T} \right),
\end{split}
\label{colla_limit_2}
\end{equation}
we conclude the proof by letting $T\rightarrow \infty$ for both sides of (\ref{colla_limit_2}) and applying the Sandwich Theorem.
\end{proof}

\textbf{Corollary \ref{coro_ber_regulation}} indicates that gradually decreasing the BER at a rate of $\mathcal{O}(1/(r+1)^2)$ can guarantee wireless FL converges at a rate of $\mathcal{O}(1/\sqrt{T})$.
In other words, the proposed BER regulation rule achieves an improving communication quality over communication rounds. This is consistent with the solution developed in \cite{wei2022federated} for wireless FL on strongly convex loss functions, which improves the communication quality by increasing the signal-to-noise ratio (SNR) with communication rounds.

\subsection{BER vs. Number of LDPC Decoding Iterations}

The analysis in Section III-A indicates that the convergence rate of conventional FL with error-free communication can still be achieved in wireless FL if the BER is properly controlled. In this subsection, we analyze the relationship between BER and the maximum number of LDPC decoding iterations, based on which, we develop an adaptive decoding scheme to maintain the desired convergence rate.
\begin{algorithm}[t]
\label{algorithm}
	\setstretch{1.0}
	{\caption{Wireless FL with Adaptive LDPC Decoding}
	\textbf{Input: }{The initial global model $\mathbf{w}_0$, $R>1$, $b_0$, and $b_{R-1}$.}\\
	\For{ Communication round $r=0, \cdots, R-1$}
	{
            The server broadcasts the $N$-bit representation of the global model $\mathbf{w}_r$ and scalars $\min(\mathbf{w}_{r}),\max(\mathbf{w}_{r})$.\\
		\For{Client $k=1, \cdots, K$}
		{
			Compute $b_r=\frac{(b_0-b_{R-1})R^2}{(R^2-1)(r+1)^2}+\frac{b_{R-1}R^2-b_0}{R^2-1}$.\\
			Find $Q_r$ by its mapping with $b_r$.\\
			Initialize $decoding.iter=0$.\\
                \For{$decoding.iter < Q_r$}
			{Execute one iteration of the LDPC decoding algorithm for each frame that contains the global model parameters.\\
				\If{Parity check passes}
				{Early termination and break.}
				$decoding.iter+=1$
			}
			Reconstruct $\mathbf{\tilde{w}}_r^k$ according to the LDPC decoding results and $\max(\mathbf{w}_{r}),\min(\mathbf{w}_{r})$. \\
   Set
                $\mathbf{w}^k_{r,0}=\mathbf{\tilde{w}}_r^k$.\\
			\For{local training step $e=0,...,E-1$}
			{
				$\mathbf{w}^k_{r,e+1}=\mathbf{w}^k_{r,e}-\eta\nabla f_k(\mathbf{w}^k_{r,e};\xi^k_{r,e})$.\
			}
			Upload $\Delta \mathbf{w}_r^k=\mathbf{w}^k_{r,E}-\mathbf{w}_{r,0}^k$ to the server.\
		}
		Perform model aggregation at the server as $\mathbf{w}_{r+1}=\mathbf{w}_r+\frac{1}{K}\sum^K_{k=1}\Delta \mathbf{w}_r^k$.
 }
 \textbf{Output: }{$\mathbf{\bar{w}}_{R,E}=\frac{1}{K}\sum_{k=1}^K\mathbf{w}_{R,E}^k$ as the learned model.}}
\end{algorithm}
It has been shown in \cite{khandekar2001complexity} that for message-passing-based channel decoders, to achieve a fraction of 1-$\epsilon$ of the channel capacity over a binary erasure channel, the relationship between the decoding complexity (denoted as $\chi_D$) and BER (denoted as $b$) can be modeled as follows:
\begin{equation}
\chi_D = \mathcal{O}\left(\log\frac{1}{b}\right)+\mathcal{O}\left(\frac{1}{\epsilon}\log\frac{1}{\epsilon}\right).
\label{iter_ber_capacity}
\end{equation}
According to \cite{mackay1999good}, the relationship between the decoding complexity and the maximum number of decoding iterations ($Q$) can be modeled as follows:
\begin{equation}
\chi_D = Q\cdot\chi_0,
\label{complexity_iter}
\end{equation}
where $\chi_0$ denotes the decoding complexity per iteration.
By combining (\ref{iter_ber_capacity}) and (\ref{complexity_iter}), it leads to the following relationship between the maximum number of decoding iterations and BER:
\begin{equation}
Q = \frac{1}{\chi_0}\mathcal{O}\left(\log\frac{1}{b}\right)+\frac{1}{\chi_0}\mathcal{O}\left(\frac{1}{\epsilon}\log\frac{1}{\epsilon}\right).
\label{iter_ber}
\end{equation}
In the next subsection, we develop an adaptive scheme to control the maximum number of LDPC decoding iterations for clients in a wireless FL system based on \textbf{Corollary 1} and (\ref{iter_ber}).

\subsection{LDPC Decoding Iteration Control}
\label{Decoding Iterations Control}
In this subsection, we show how to adaptively control the maximum number of LDPC decoding iterations over communication rounds to achieve the desired BER.

In principle, by combining (\ref{BERr}) and (\ref{iter_ber}), we obtain the following relationship between the maximum number of LDPC decoding iterations and communication rounds:
\begin{equation}
Q_r = \frac{1}{\chi_0}\mathcal{O}\left(2\log\ (r+1)\right)+\frac{1}{\chi_0}\mathcal{O}\left(\frac{1}{\epsilon}\log\frac{1}{\epsilon}\right),
\label{iter_round}
\end{equation}
which provides a theoretical guideline to determine the maximum number of LDPC decoding iterations in the $r$-th communication round. Specifically, $Q_{r}$ can be set as $\alpha \log (r+1) + \beta$, where $\alpha$ and $\beta$ are constants. However, (\ref{iter_round}) cannot be directly utilized because $\alpha$ and $\beta$ are code-specific and depend on the receive SNR as well as channel decoding algorithms. Thus, we propose a practical approach via the mapping between BER and the maximum number of LDPC decoding iterations at different levels of receive SNR. Entries in the mapping can be obtained through empirical measurements specific to the LDPC code used in the wireless FL system. Fig.~\ref{fig:BER performance} shows the sample simulation results for the rate 1/2 LDPC code with a block length of 1008 bits and binary phase shift keying (BPSK) modulation over an additive white Gaussian noise (AWGN) channel with $\text{SNR}$ = 1.5 and 2.5 dB, which can be stored in memory of clients~\cite{lu2004performance}. From the mapping, we are able to obtain the corresponding maximum number of decoding iteration for a target BER for clients with different receive SNR in each communication round.

\begin{figure}[t]
\small
\centering
\includegraphics[width=2.5in]{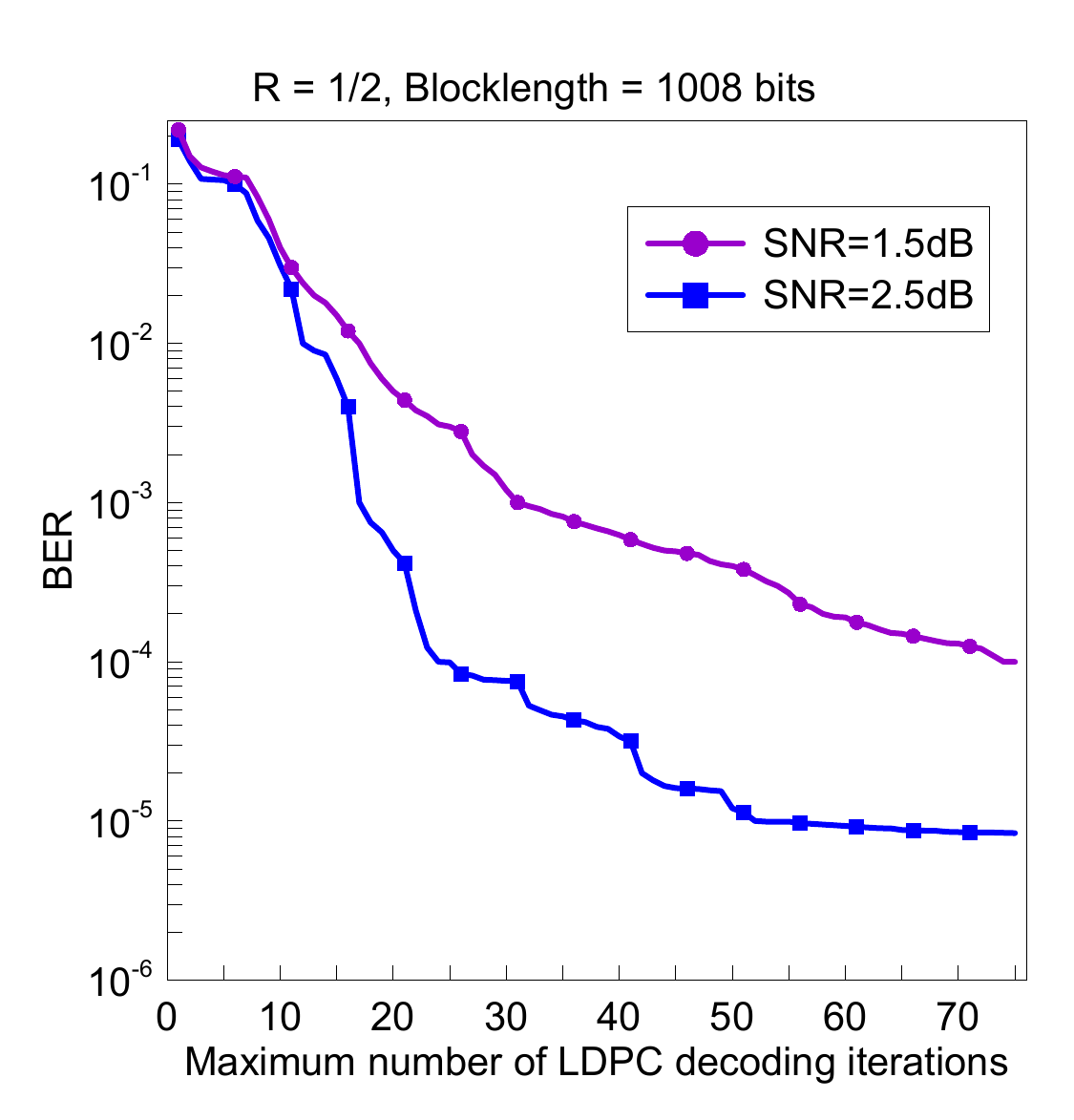}
\caption{Mappings between BER and the maximum number of LDPC decoding iterations at different receive SNRs.}
\label{fig:BER performance}
\end{figure}

The wireless FL procedures with adaptive LDPC decoding are summarized in \textbf{Algorithm~1}. The algorithm starts with an initial BER, denoted as $b_0$, and gradually decreases the BER with communication rounds at the rate of $\frac{1}{(r+1)^2}$. In practice, since the BER of LDPC decoding cannot be arbitrarily small (see Fig. \ref{fig:BER performance}), we set a target BER at the last communication round, denoted as $b_{R-1}$. With $b_0$, $b_{R-1}$ ($b_{0}>b_{R-1}$), and the desired BER should scale at the rate of $\frac{1}{(r+1)^2}$, $b_{r}$ can be designed as 
\begin{equation}
b_r=\frac{(b_0-b_{R-1})R^2}{(R^2-1)(r+1)^2}+\frac{b_{R-1}R^2-b_0}{R^2-1}
\label{concrete_br}
\end{equation}
for $R>1$.
The corresponding maximum number of LDPC decoding iterations in that round is obtained from the measured mapping relationship.  

\begin{remark}
Note that the number of entries in the mappings depend on the granularity of discrete BER and SNR values. Since this number is not large, the overhead for storing them is expected to be small. Additionally, reading these entries from memory takes constant time and does not introduce additional LDPC decoding latency. 
\end{remark}

\section{Experimental Results}
\label{Expirements}
In this section, we validate the effectiveness of the proposed adaptive LDPC decoding algorithm via numerical experiments. First, we introduce the experimental settings. Then, performance of the proposed algorithm, including test accuracy and decoding energy consumption, is compared with a baseline by fixing the maximum number of LDPC decoding iterations in each communication round (hereinafter referred to as the ``Fixed-$Q$'' scheme).

\begin{figure}[t]
\small
\centering
\includegraphics[width=2.5in]{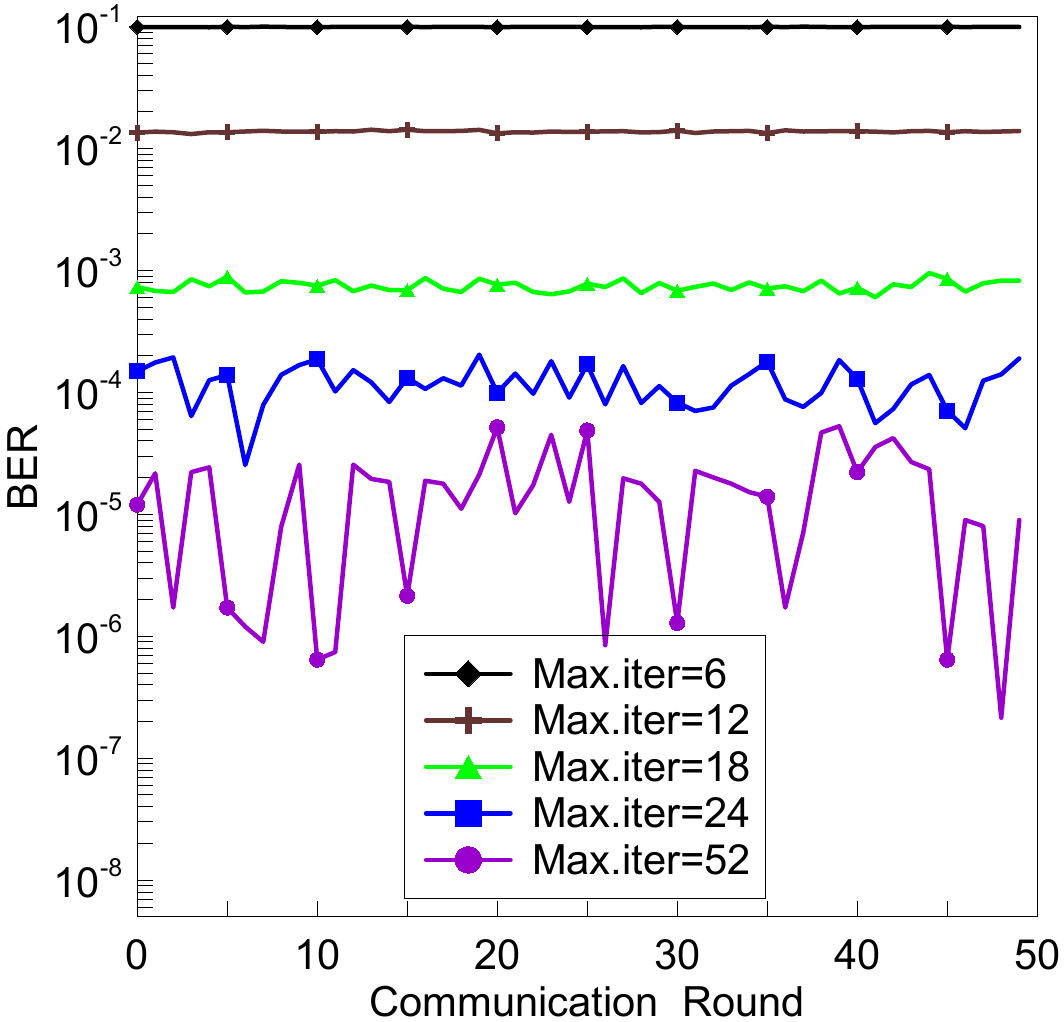}
\caption{BER of a sample client achieved by the baseline scheme (``Fixed-$Q$'') with a fixed maximum number of LDPC decoding iterations: By setting the maximum number of LDPC decoding iterations $Q$ to 6, 12, 18, 24, and 52, clients achieve the BER at around $10^{-1}$, $10^{-2}$, $10^{-3}$, $10^{-4}$, and $10^{-5}$, respectively.}
\label{fig:Achieved BER}
\end{figure}

\subsection{Experimental Settings}
\label{ex_setting}
\label{Experiment settings}
\begin{figure}[t]
\small
\begin{minipage}[b]{.48\linewidth}
  \centering
  \centerline{\includegraphics[width=4.1cm]{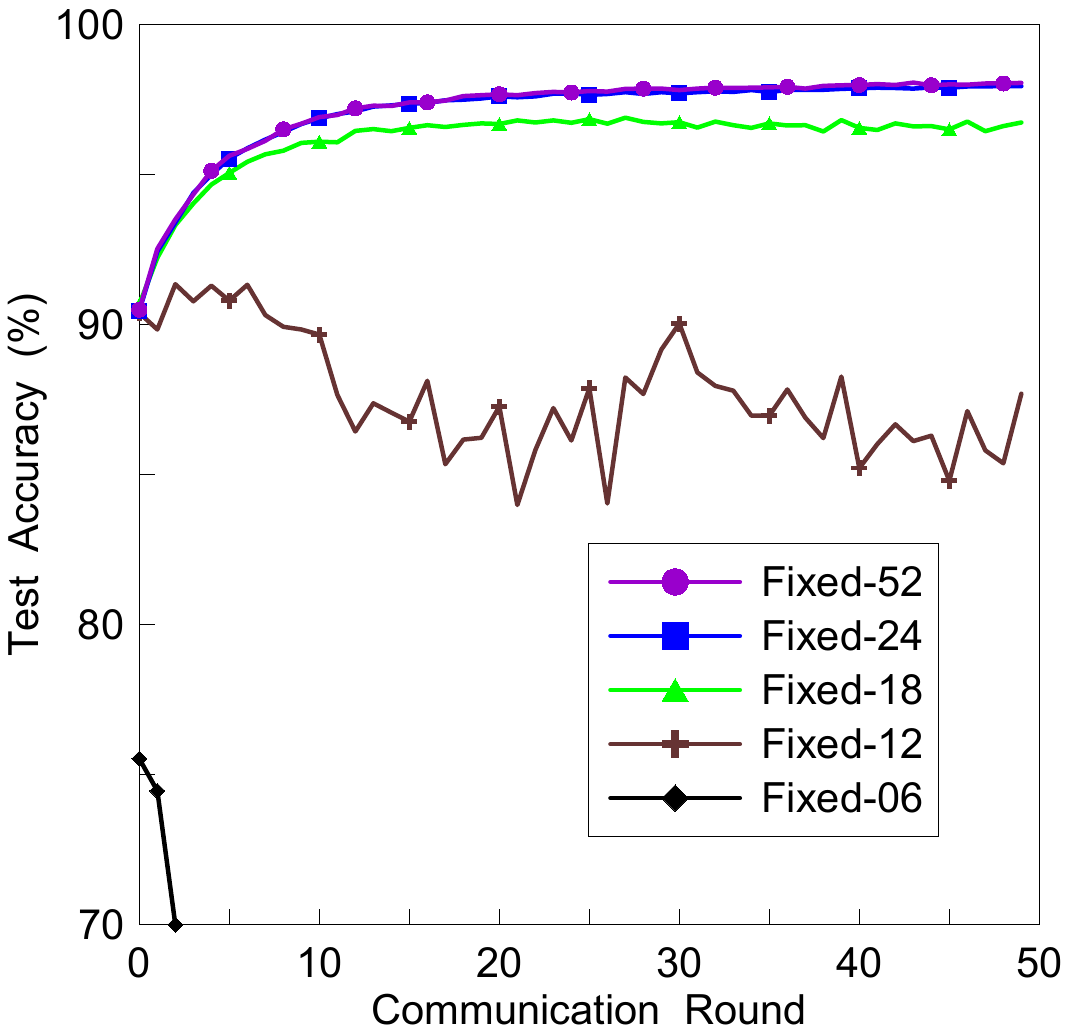}}
  \begin{center}
  (a) LeNet-300-100 on MNIST (IID).
  \end{center}
\end{minipage}
\hfill
\begin{minipage}[b]{0.48\linewidth}
  \centering
  \centerline{\includegraphics[width=4.0cm]{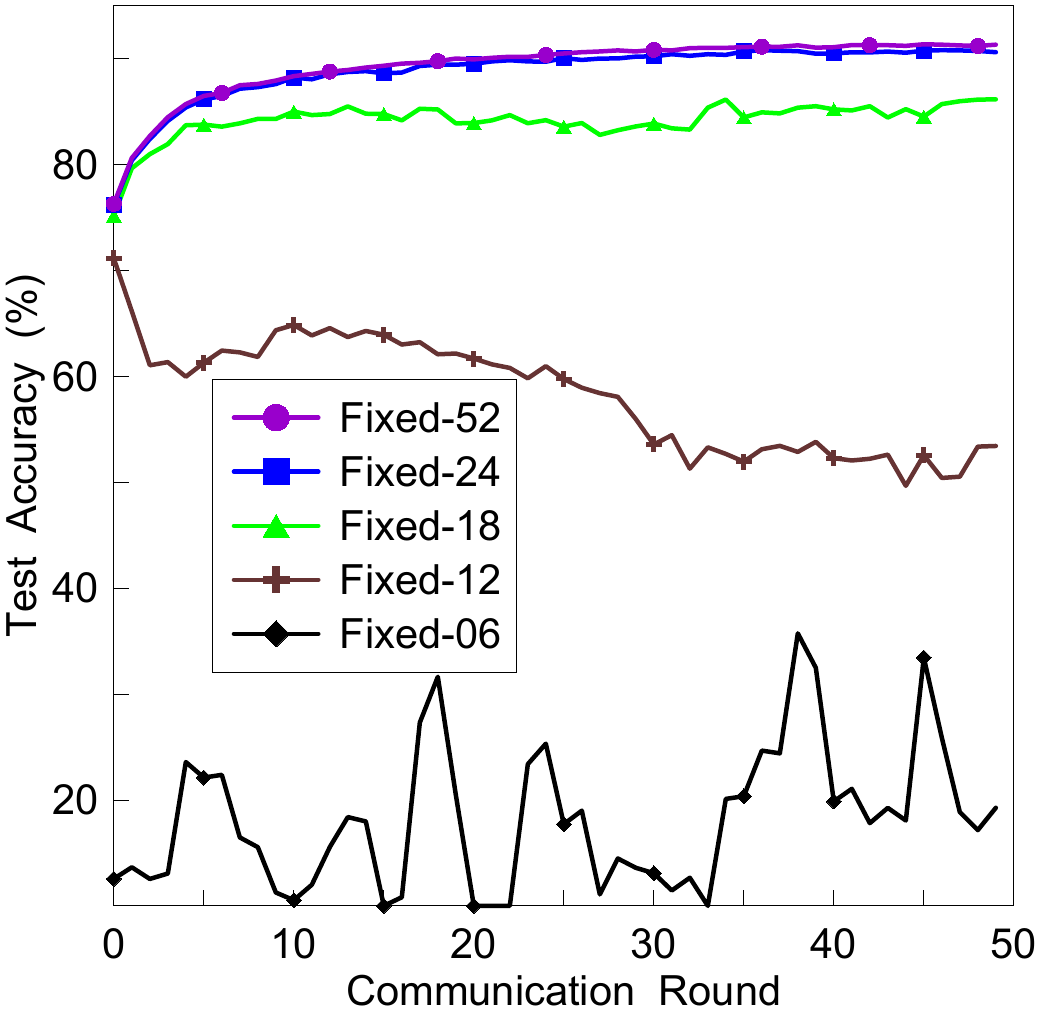}}
  \begin{center}
  (b) Vanilla-CNN on Fashion-MNIST (IID).
  \end{center}
\end{minipage}
\hfill
\begin{minipage}[b]{0.48\linewidth}
  \centering
  \centerline{\includegraphics[width=4.0cm]{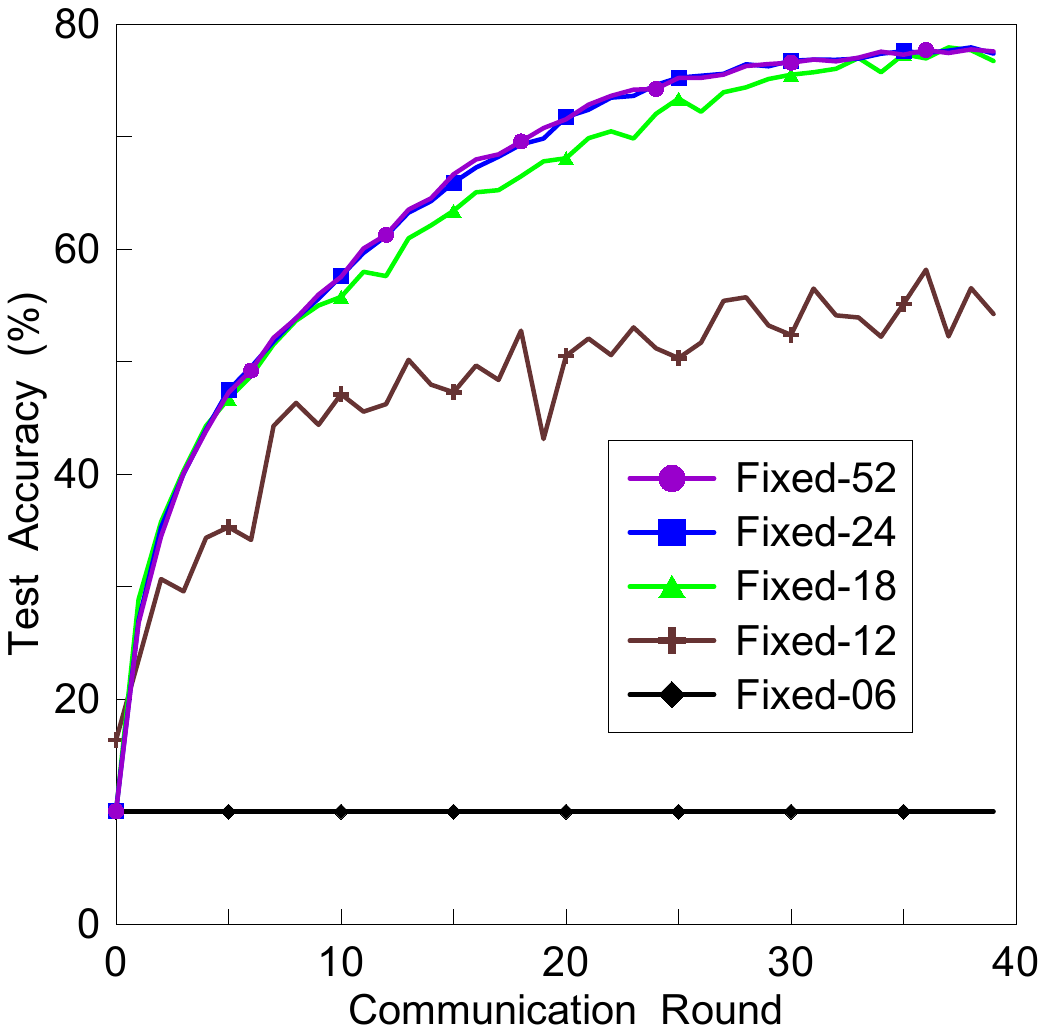}}
  \begin{center}
  (c) 7-layers CNN on CIFAR-10 (IID).
  \end{center}
\end{minipage}
\hfill
\begin{minipage}[b]{0.48\linewidth}
  \centering
  \centerline{\includegraphics[width=4.0cm]{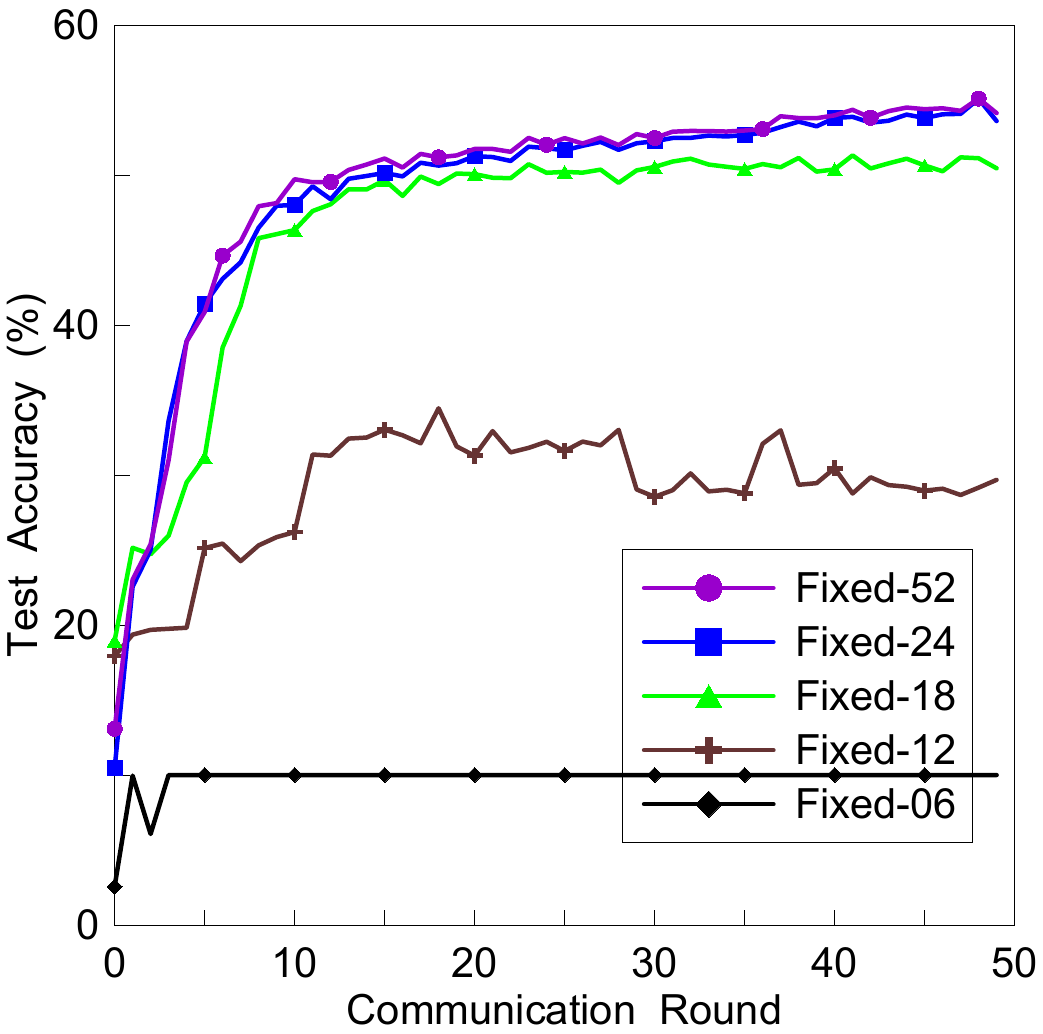}}
  \begin{center}
  (d) Vanilla-CNN on Fashion-MNIST (non-IID).
  \end{center}
\end{minipage}
\caption{Test accuracy on different datasets under the baseline scheme ``Fixed-$Q$'' with a fixed maximum number of LDPC decoding iterations.}
\label{fig:BER impacts on FL}
\end{figure}

\begin{figure}[t]
\small
\begin{minipage}[b]{.48\linewidth}
  \centering
  \centerline{\includegraphics[width=4.0cm]{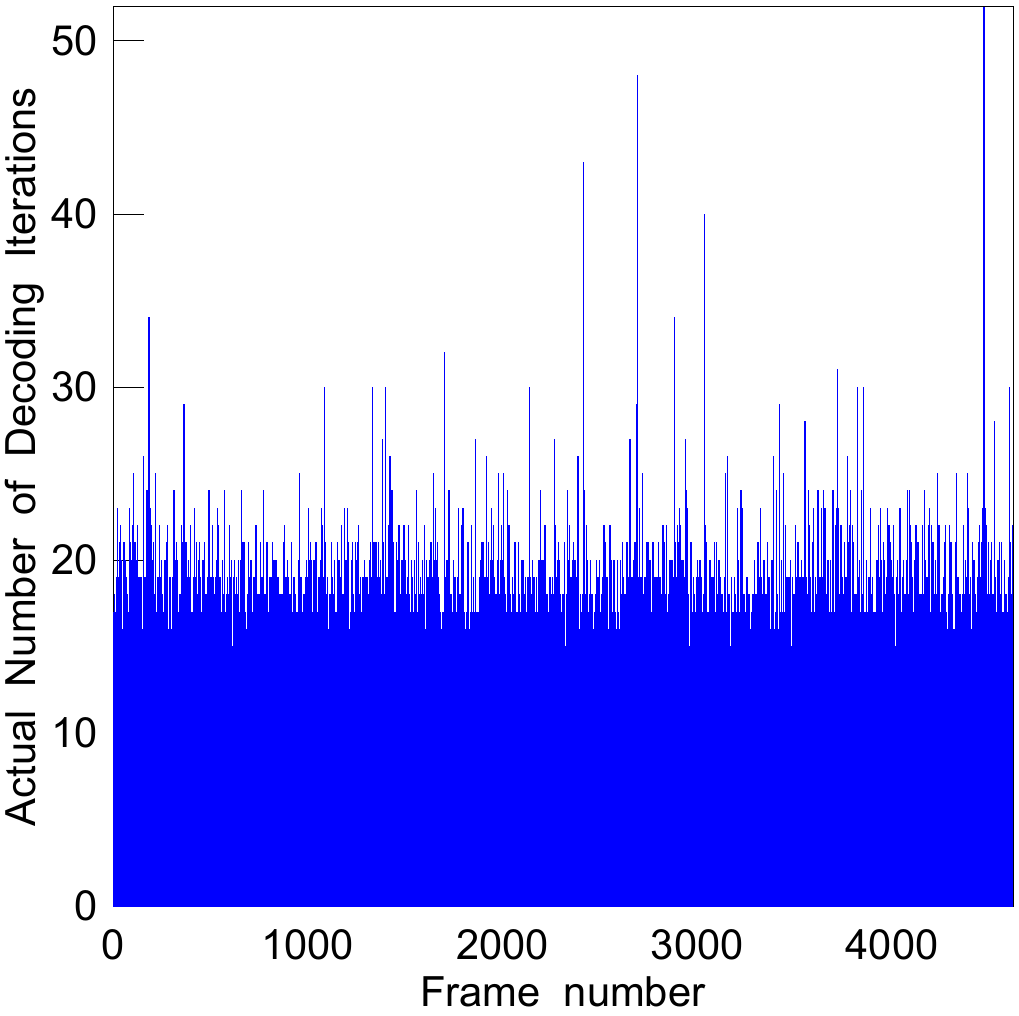}}
  \begin{center}
  (a) Actual number of LDPC decoding iterations.
  \end{center}
\end{minipage}
\hfill
\begin{minipage}[b]{0.48\linewidth}
  \centering
  \centerline{\includegraphics[width=4.0cm]{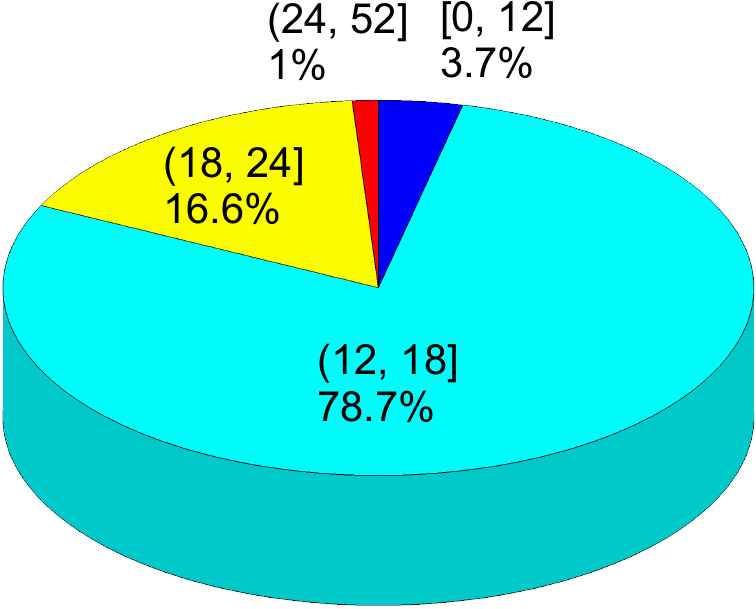}}
  \begin{center}
  (b) Distribution of the actual number of LDPC decoding iterations.
  \end{center}
\end{minipage}
\caption{Actual numbers of LDPC decoding iterations and its distribution at a sample client in a selected communication round on Fashion-MNIST (IID). A 8-bit vanilla CNN model is encoded by the rate 1/2 min-sum LDPC code of block-length of 1008 bit, resulting in 4620 frames in a communication round.}
\label{fig:Consumed LDPC decoding iteration}
\end{figure}
\textit{System Setup:} We consider a wireless FL system with 10 clients. The clients communicate with the server via the rate 1/2 min-sum LDPC code using a code length of 1008 bits\cite{sha2008multi}. BPSK is adopted as the modulation scheme. In each communication round, $N=8$ bits are used to digitalize each model parameter. The number of local training steps in each communication round is set to $E=5$ with the learning rate and batch size as 0.01 and 64, respectively. The receive SNR of all the downlink wireless channels is assumed to be 2.5 dB.

\textit{Datasets}: We conduct experiments on three classic datasets, including MNIST\cite{lecun1998mnist}, Fashion-MNIST\cite{xiao2017fashion}, and CIFAR-10\cite{krizhevsky2010cifar}. The training datasets are distributed among clients in an IID manner. We further consider one experiment where the Fashion-MNIST data are non-IID among clients. In particular, each client is assigned with two categories of training data similar to the setup in \cite{li2020federated}. The test datasets are used for model accuracy evaluation at the server.

\textit{DNN Models}: For the MNIST dataset, the LeNet-300-100 model\cite{han2015deep} is employed, which is a fully connected neural network consisting of two hidden layers and one softmax layer, with approximately 0.27 million parameters.
For the Fashion-MNIST dataset, a vanilla convolutional neural network (CNN) \cite{mcmahan2017communication} with two 5×5 convolution layers, one fully connected layer, and one softmax layer is employed, comprising around 0.58 million parameters. 
For the CIFAR-10 dataset, a 7-layer CNN model with four 3×3 convolution layers, two fully connected layers, and one softmax layer is employed, which contains approximately 9.07 million parameters.

\subsection{Learning Performance of the Baseline Scheme}
\begin{figure}[t]
\small
\begin{minipage}[b]{.48\linewidth}
  \centering
  \centerline{\includegraphics[width=4.0cm]{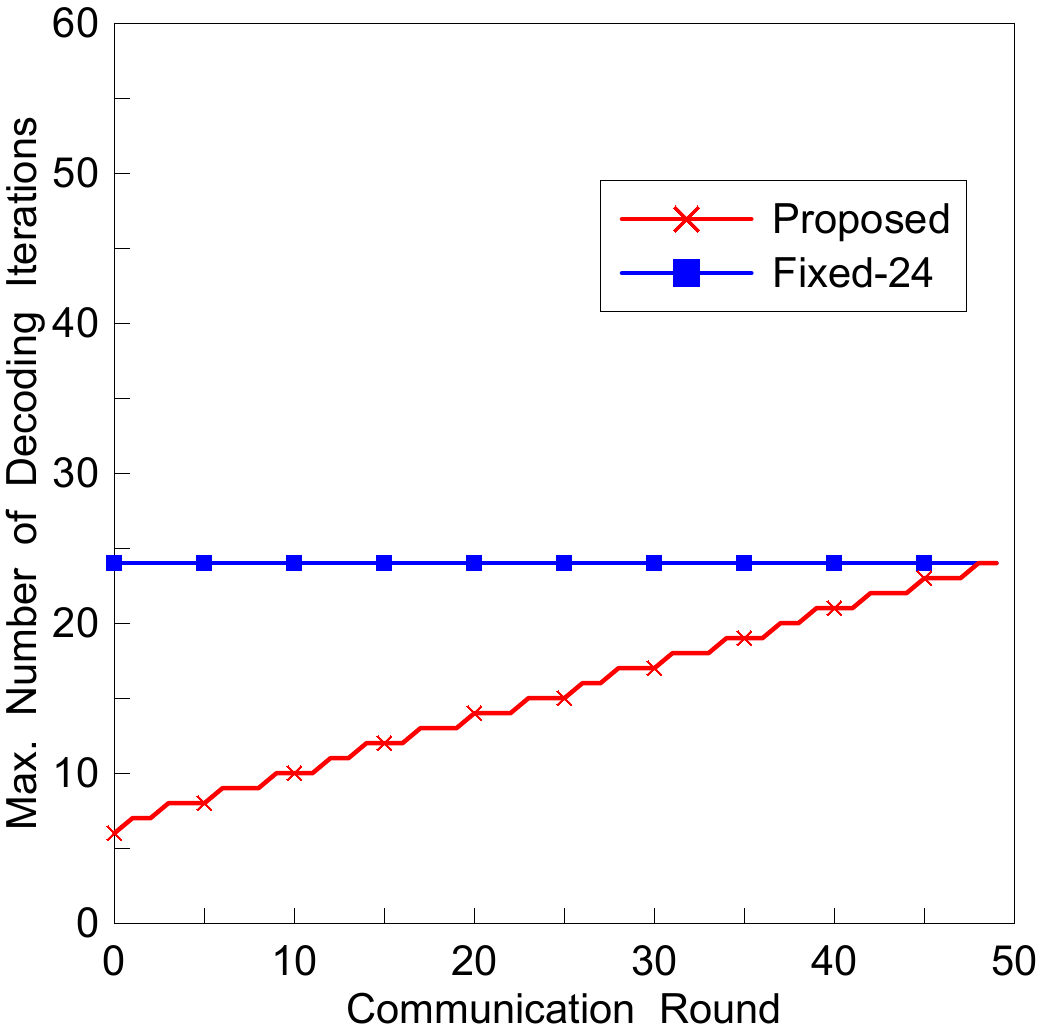}}
  \begin{center}
  (a) $Q_r$.
  \end{center}
\end{minipage}
\hfill
\begin{minipage}[b]{0.48\linewidth}
  \centering
  \centerline{\includegraphics[width=4.0cm]{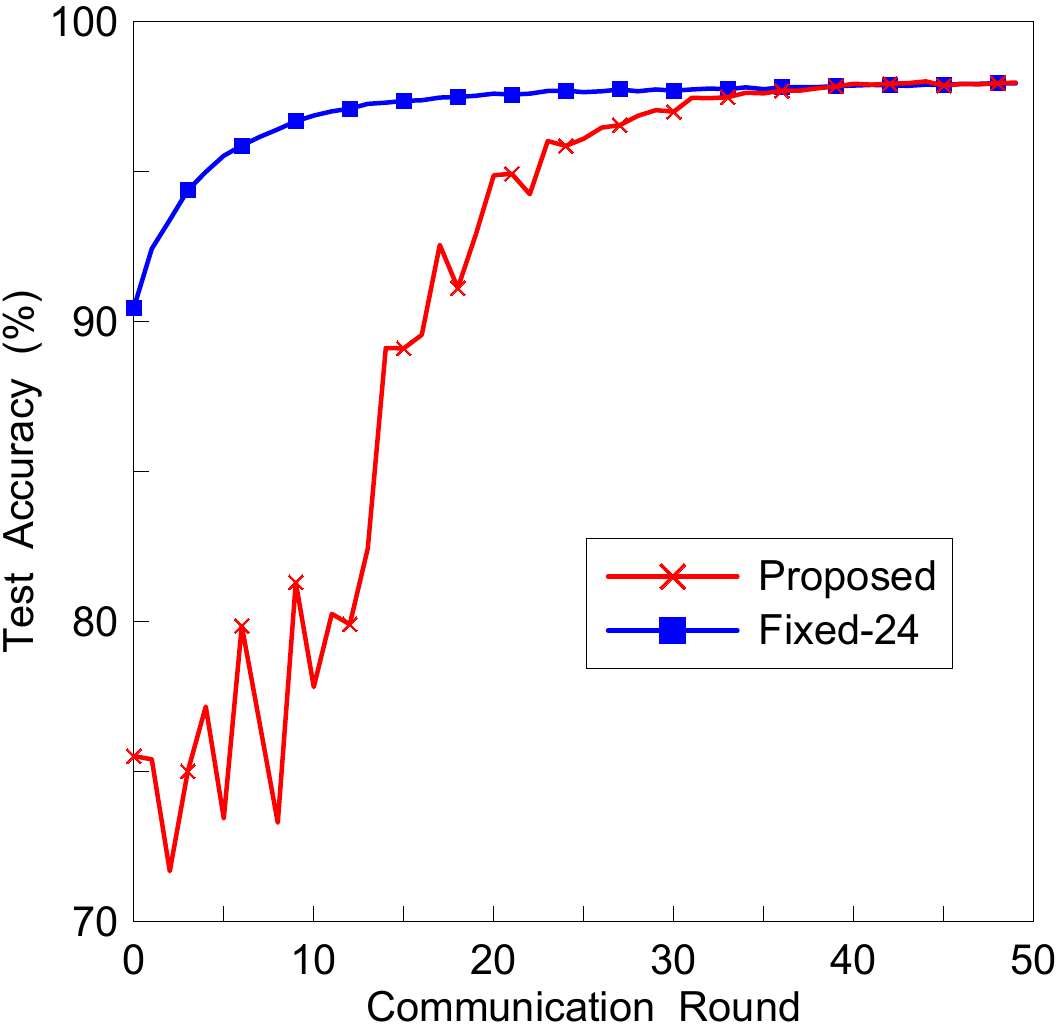}}
  \begin{center}
  (b) Test accuracy.
  \end{center}
\end{minipage}
\hfill
\begin{minipage}[b]{0.48\linewidth}
  \centering
  \centerline{\includegraphics[width=4.0cm]{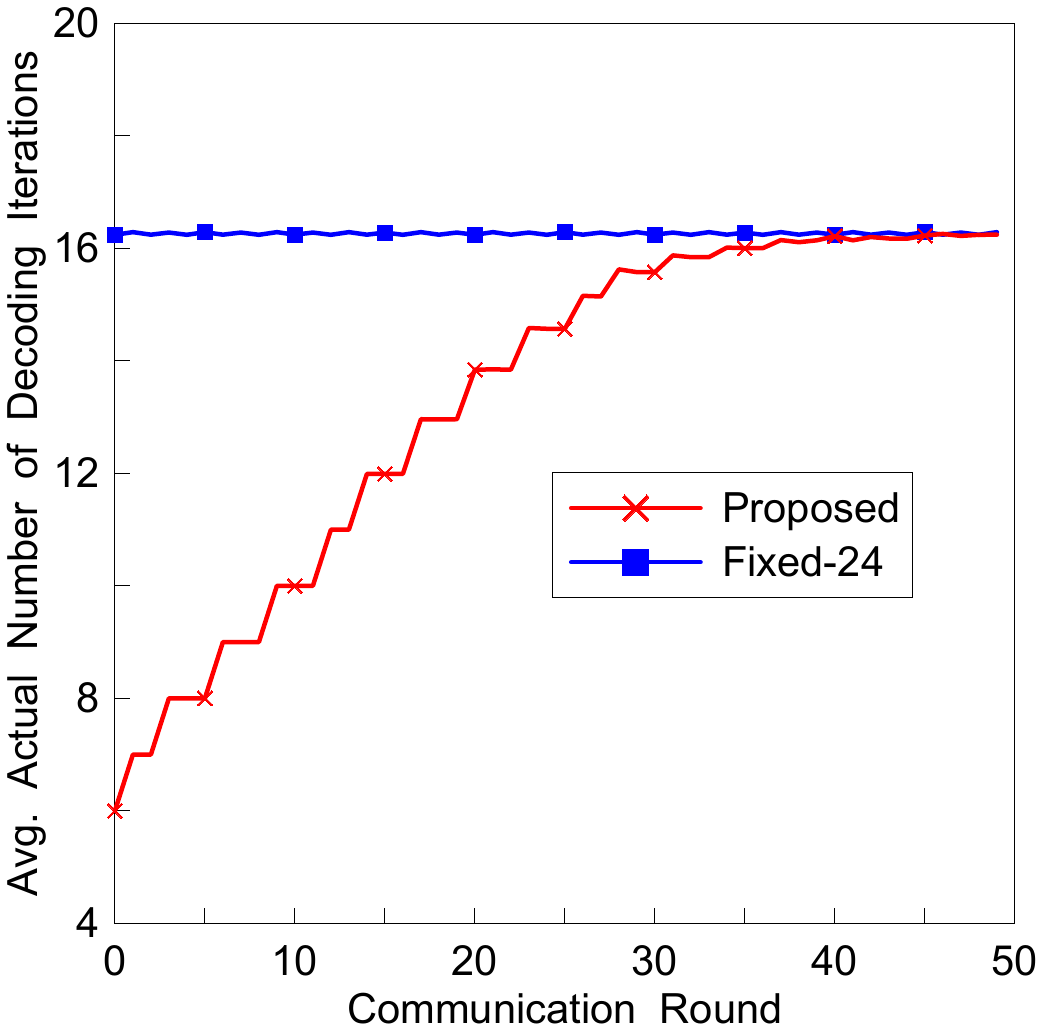}}
  \begin{center}
  (c) Avg. actual number of decoding iterations among clients.
  \end{center}
\end{minipage}
\hfill
\begin{minipage}[b]{0.48\linewidth}
  \centering
  \centerline{\includegraphics[width=4.0cm]{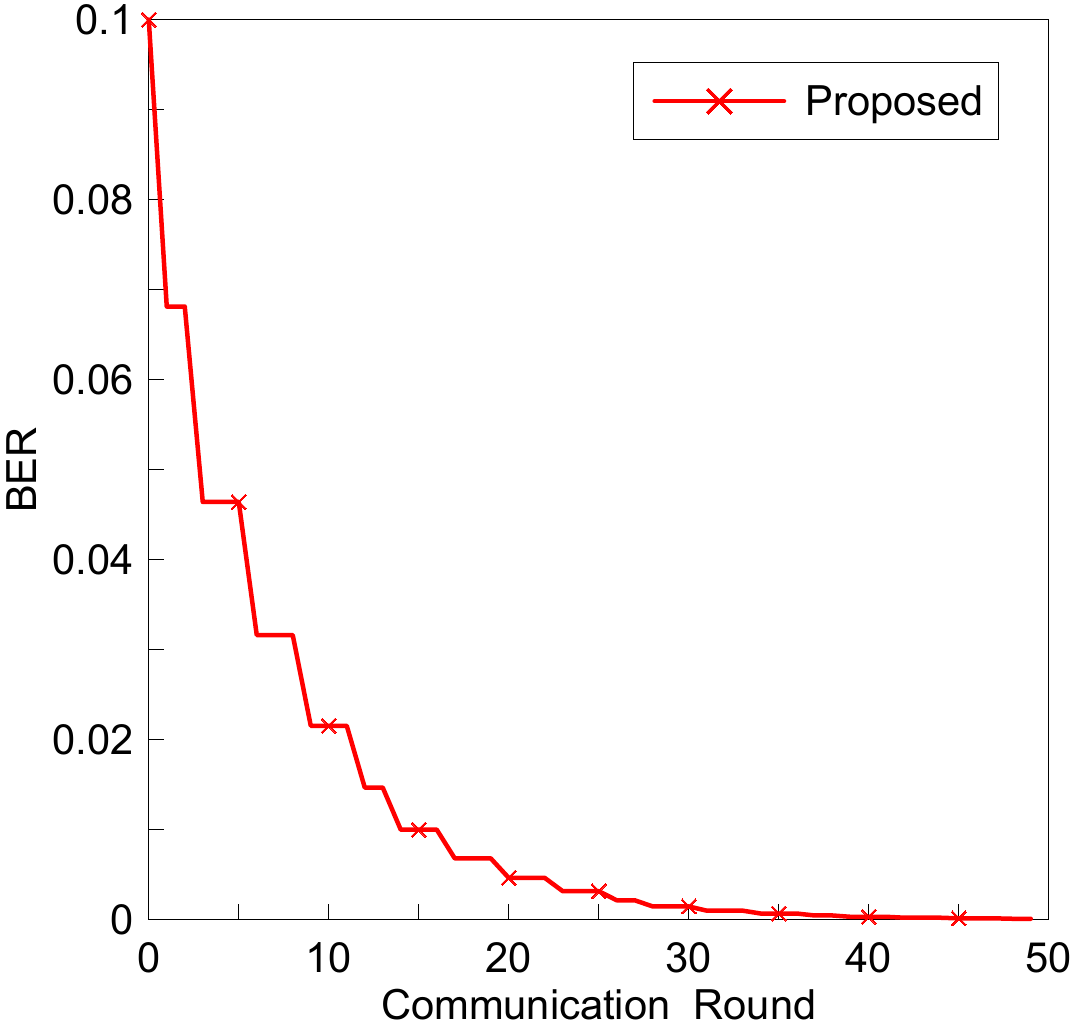}}
  \begin{center}
  (d) BER at every communication round, \\i.e., ($b_{r}$).
  \end{center}
\end{minipage}
\hfill
\begin{minipage}[b]{0.48\linewidth}
  \centering
  \centerline{\includegraphics[width=4.0cm]{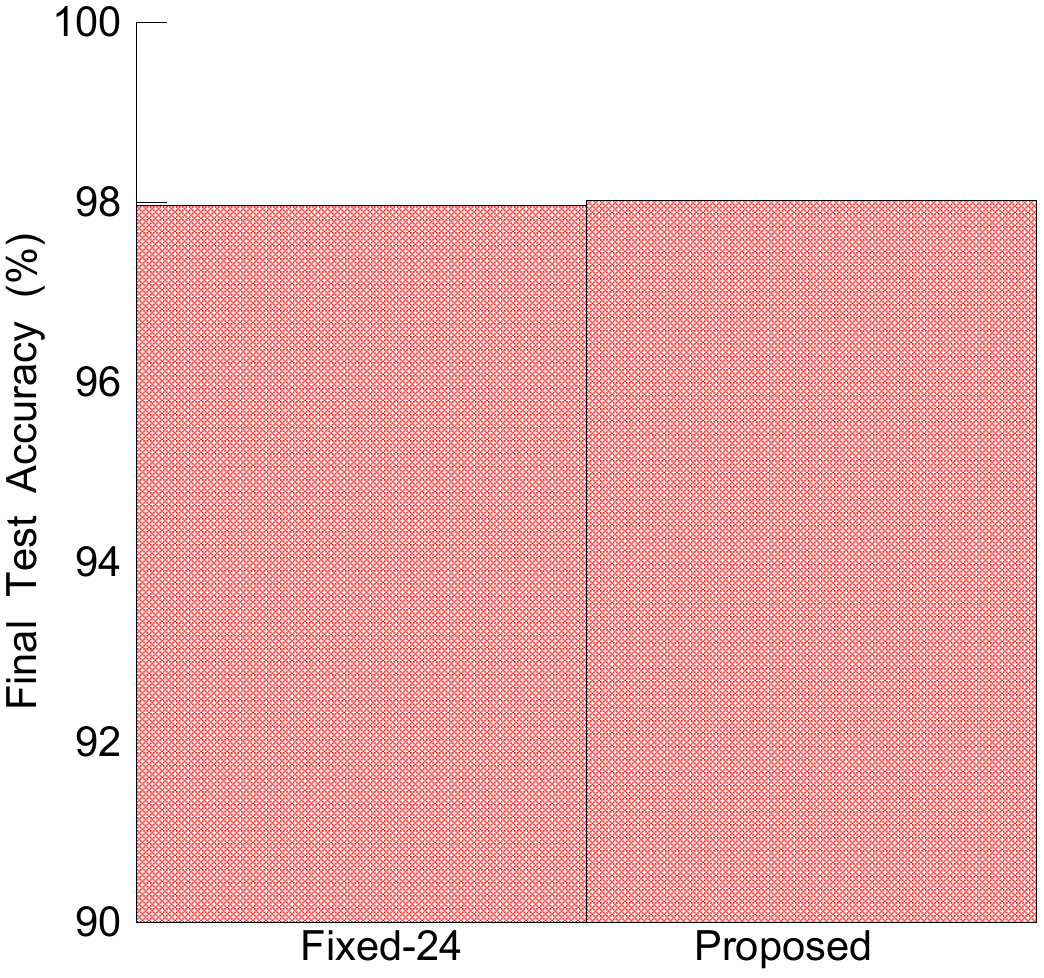}}
  \begin{center}
  (e) Test accuracy at the last communication round.
  \end{center}
\end{minipage}
\hfill
\begin{minipage}[b]{0.48\linewidth}
  \centering
  \centerline{\includegraphics[width=4.0cm]{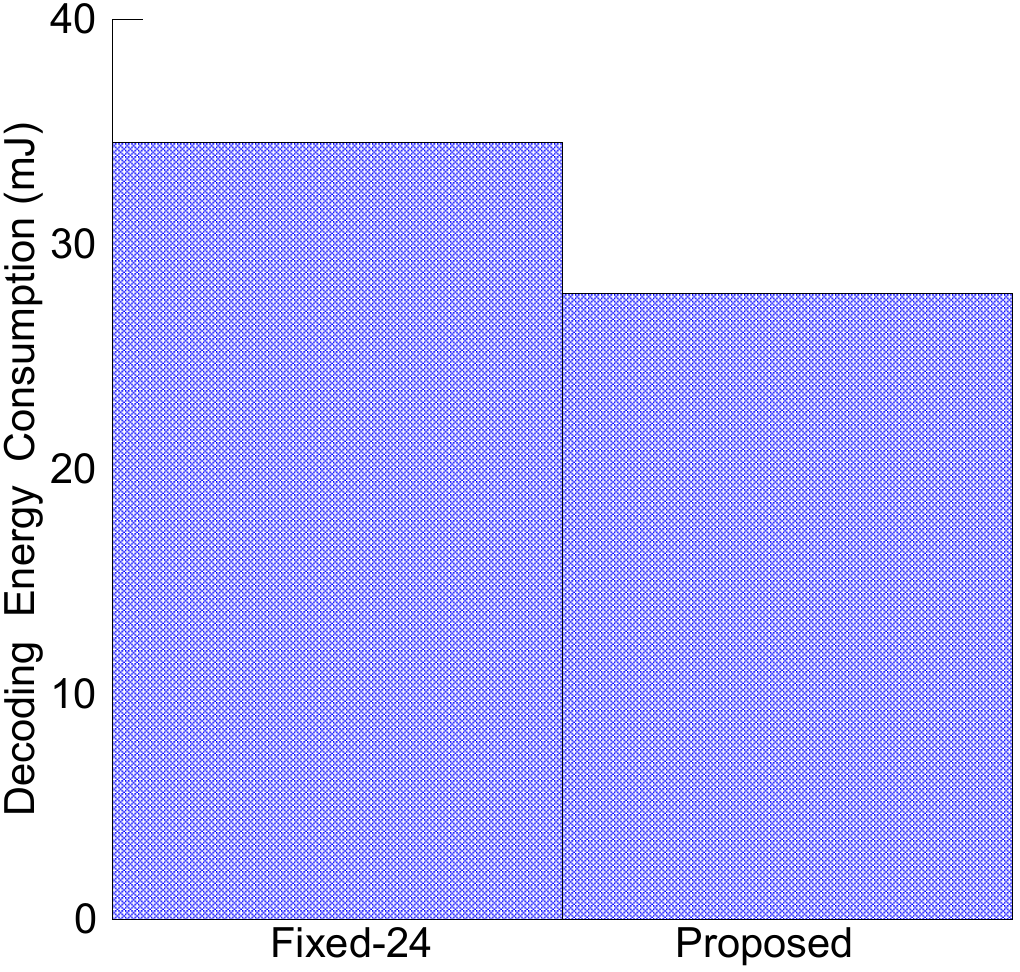}}
  \begin{center}
  (f) Total decoding energy consumption.
  \end{center}
\end{minipage}
\caption{Experimental results on MNIST (IID). Compared with the baseline scheme, our proposed algorithm that adaptively increases the maximum number of LDPC decoding iterations to 24 after 50 communication rounds can achieve the same final test accuracy of 98\% and reduce the decoding energy by 6.7~mJ.}
\label{fig:MNIST}
\end{figure}

\begin{figure}[t]
\small
\begin{minipage}[b]{.48\linewidth}
  \centering
  \centerline{\includegraphics[width=4.0cm]{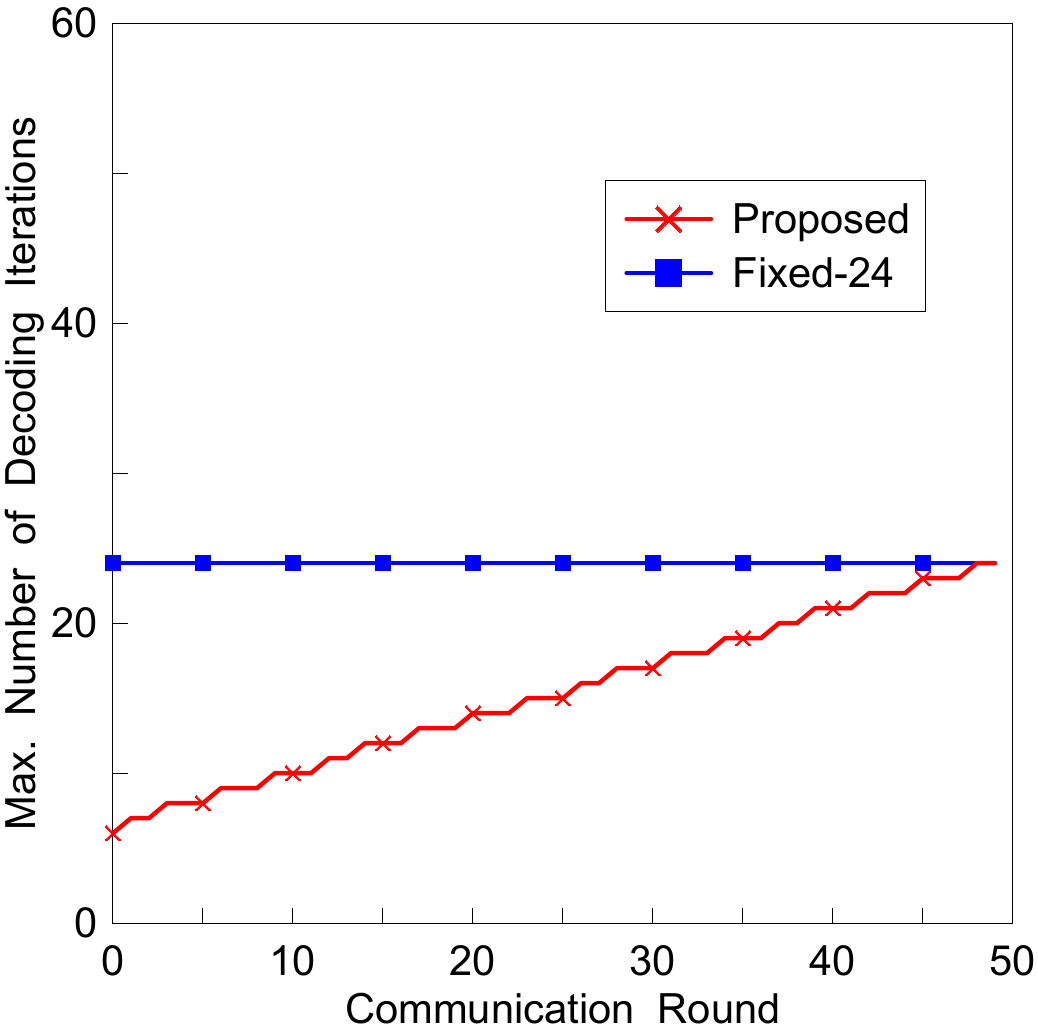}}
  \begin{center}
  (a) $Q_r$.
  \end{center}
\end{minipage}
\hfill
\begin{minipage}[b]{0.48\linewidth}
  \centering
  \centerline{\includegraphics[width=4.0cm]{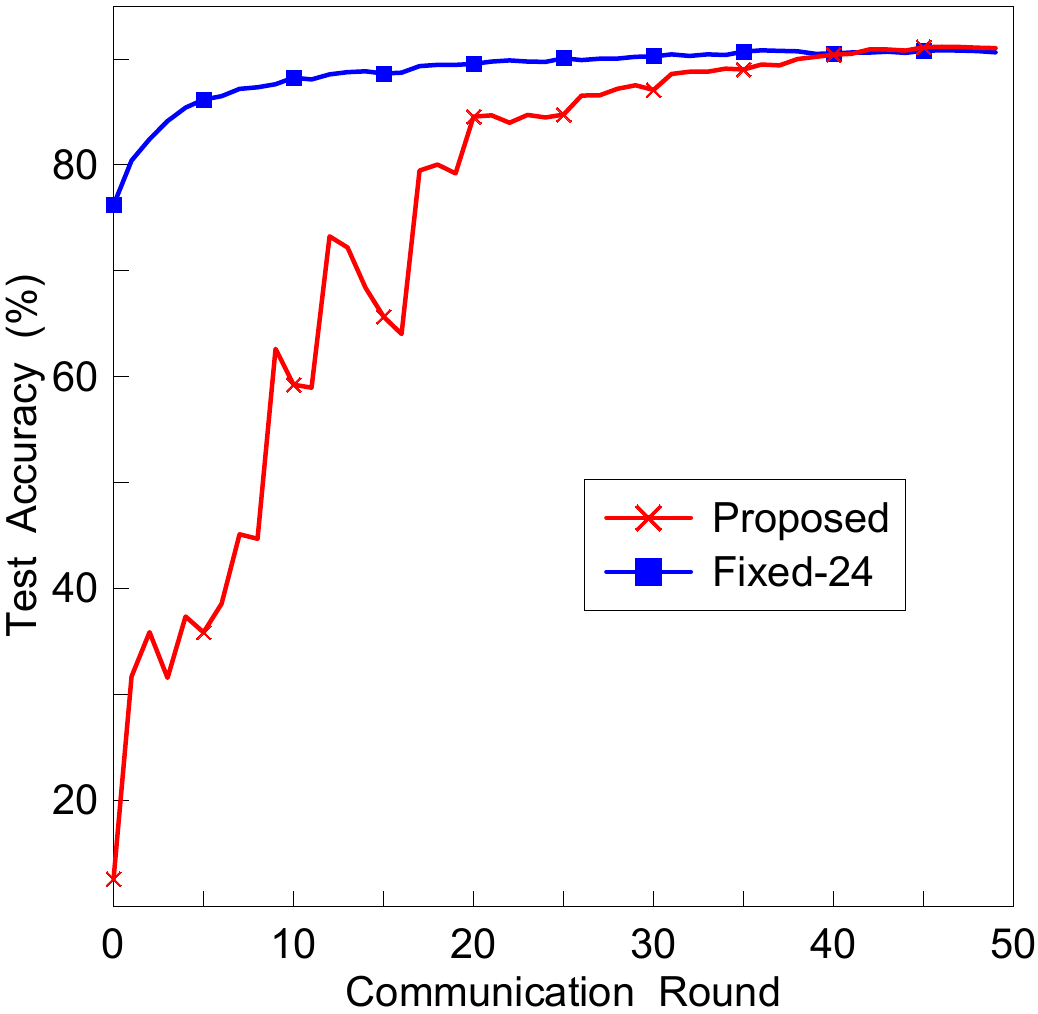}}
  \begin{center}
  (b) Test accuracy.
  \end{center}
\end{minipage}
\hfill
\begin{minipage}[b]{0.48\linewidth}
  \centering
  \centerline{\includegraphics[width=4.0cm]{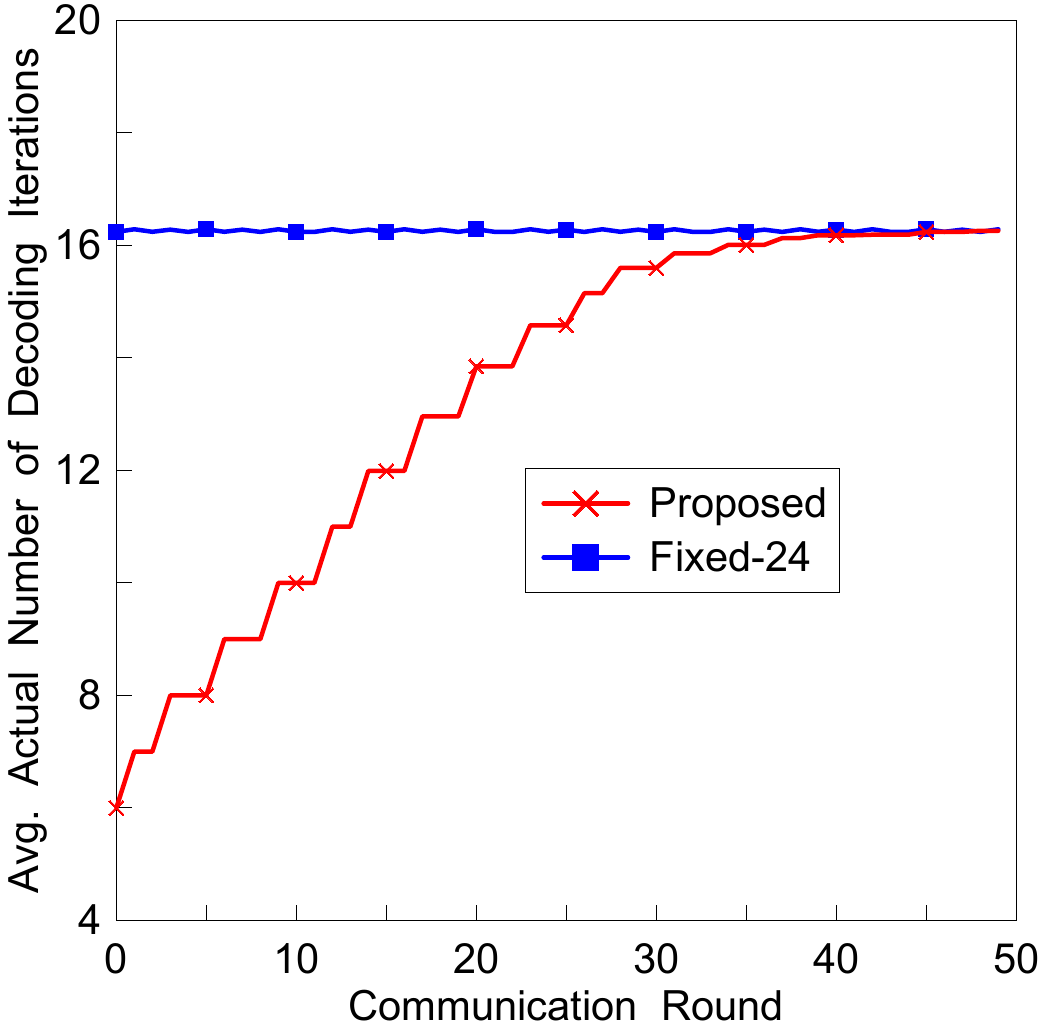}}
  \begin{center}
  (c) Avg. actual number of decoding iterations among clients.
  \end{center}
\end{minipage}
\hfill
\begin{minipage}[b]{0.48\linewidth}
  \centering
  \centerline{\includegraphics[width=4.0cm]{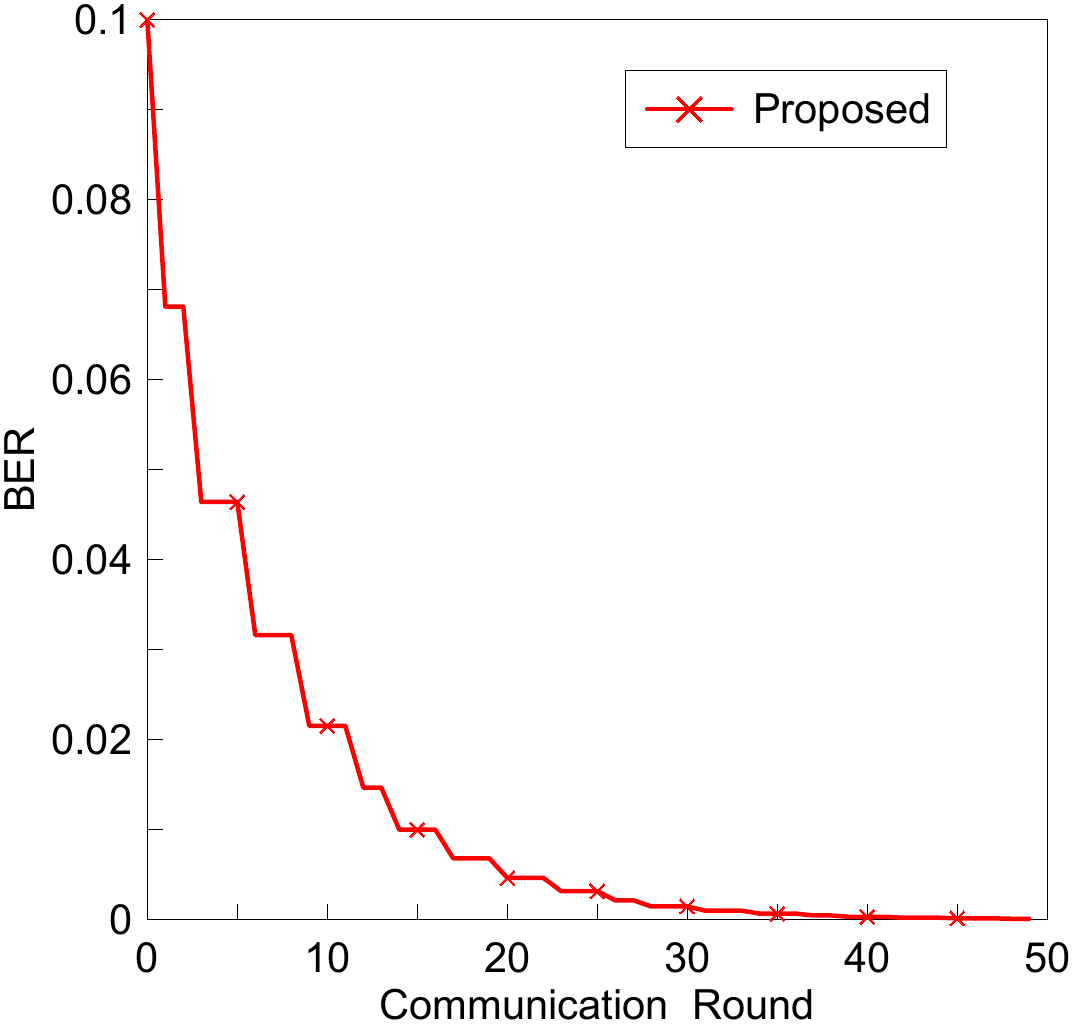}}
  \begin{center}
  (d) BER at every communication round, \\i.e., ($b_{r}$).
  \end{center}
\end{minipage}
\hfill
\begin{minipage}[b]{0.48\linewidth}
  \centering
  \centerline{\includegraphics[width=4.0cm]{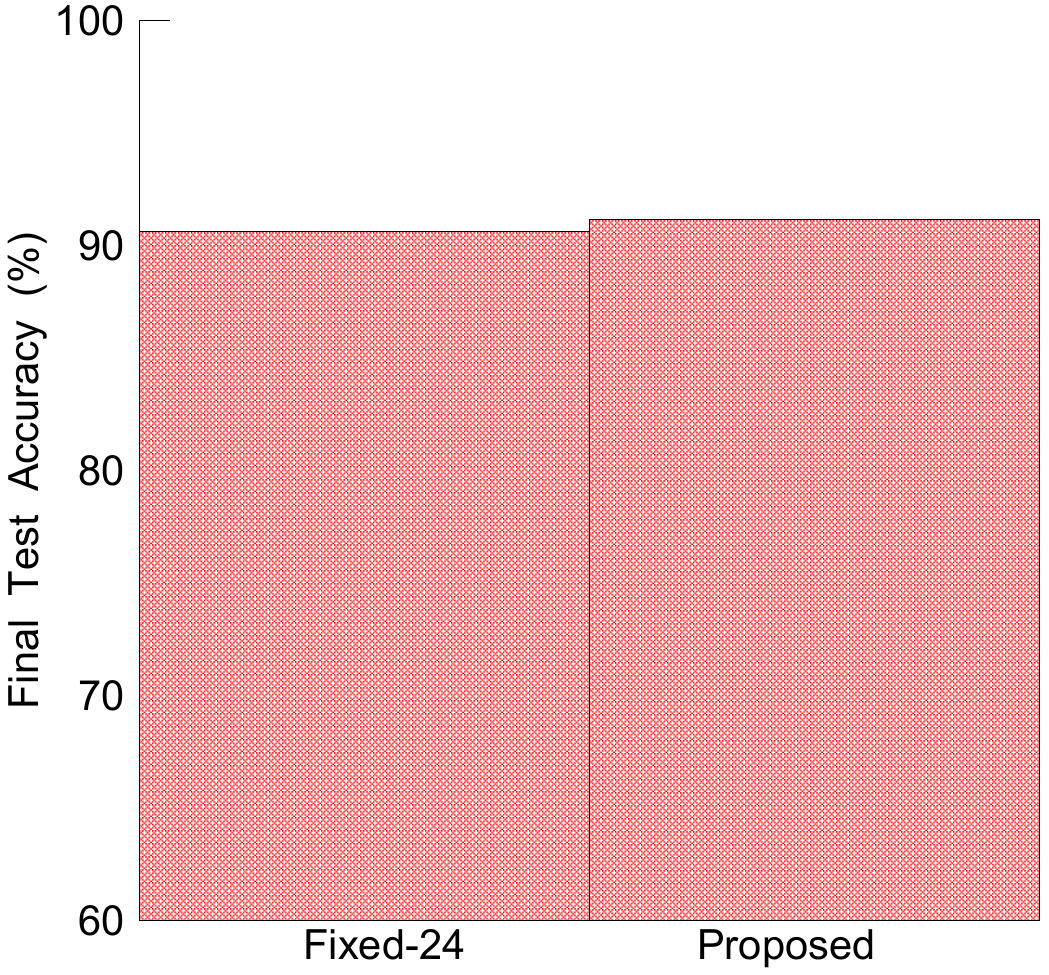}}
  \begin{center}
  (e) Test accuracy at the last communication round.
  \end{center}
\end{minipage}
\hfill
\begin{minipage}[b]{0.48\linewidth}
  \centering
  \centerline{\includegraphics[width=4.0cm]{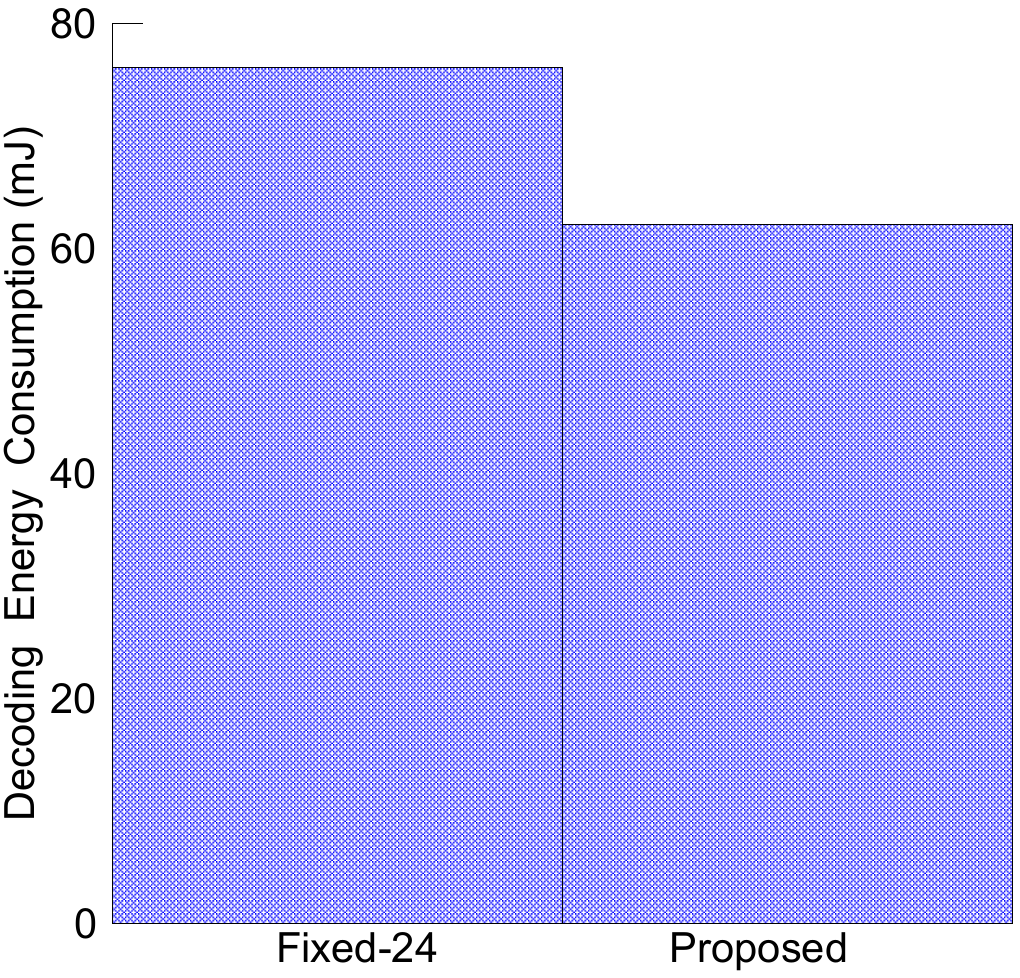}}
  \begin{center}
  (f) Total decoding energy consumption.
  \end{center}
\end{minipage}
\caption{Experimental results on Fashion-MNIST (IID). Compared with the baseline scheme, our proposed algorithm that adaptively increases the maximum number of LDPC decoding iterations to 24 after 50 communication rounds can achieve a similar final test accuracy of 91\% and reduce the decoding energy by 14~mJ.}
\label{fig:Fashion-MNIST}
\end{figure}
In this subsection, experimental results for the baseline scheme ``Fixed-$Q$'' are presented. First, as shown in Fig.~\ref{fig:Achieved BER}, a fixed maximum number of LDPC decoding iterations leads to approximately constant BERs over the training process.

Next, we investigate the impact of a roughly constant BER on the learning performance of wireless FL. Fig.~\ref{fig:BER impacts on FL} summarizes the test accuracy with different maximum numbers of LDPC decoding iterations. It can be observed that for all the four FL tasks, setting the maximum number of decoding iterations to 6 (which corresponds to a BER of $\sim 10^{-1}$) severely affects the learning accuracy. Similarly, when the maximum number of decoding iterations is set to 12 (which corresponds to a BER of $\sim 10^{-2}$), the training process fails to converge to satisfactory accuracy. By setting the maximum number of decoding iterations as 18 (which corresponds to a BER of $\sim 10^{-3}$), the training process converges, but with a noticeable loss in test accuracy compared to the case when the maximum number of decoding iterations is 24 (which corresponds to a BER of $\sim 10^{-4}$). However, further increasing the maximum number of decoding iterations to 52 does not lead to much accuracy improvement.
Thus, we next focus on the most efficient configuration of the baseline scheme with $Q=24$, i.e., ``Fixed-24''.

\subsection{Learning Performance of the Proposed Scheme}

\begin{figure}[t]
\small
\begin{minipage}[b]{.48\linewidth}
  \centering
  \centerline{\includegraphics[width=4.0cm]{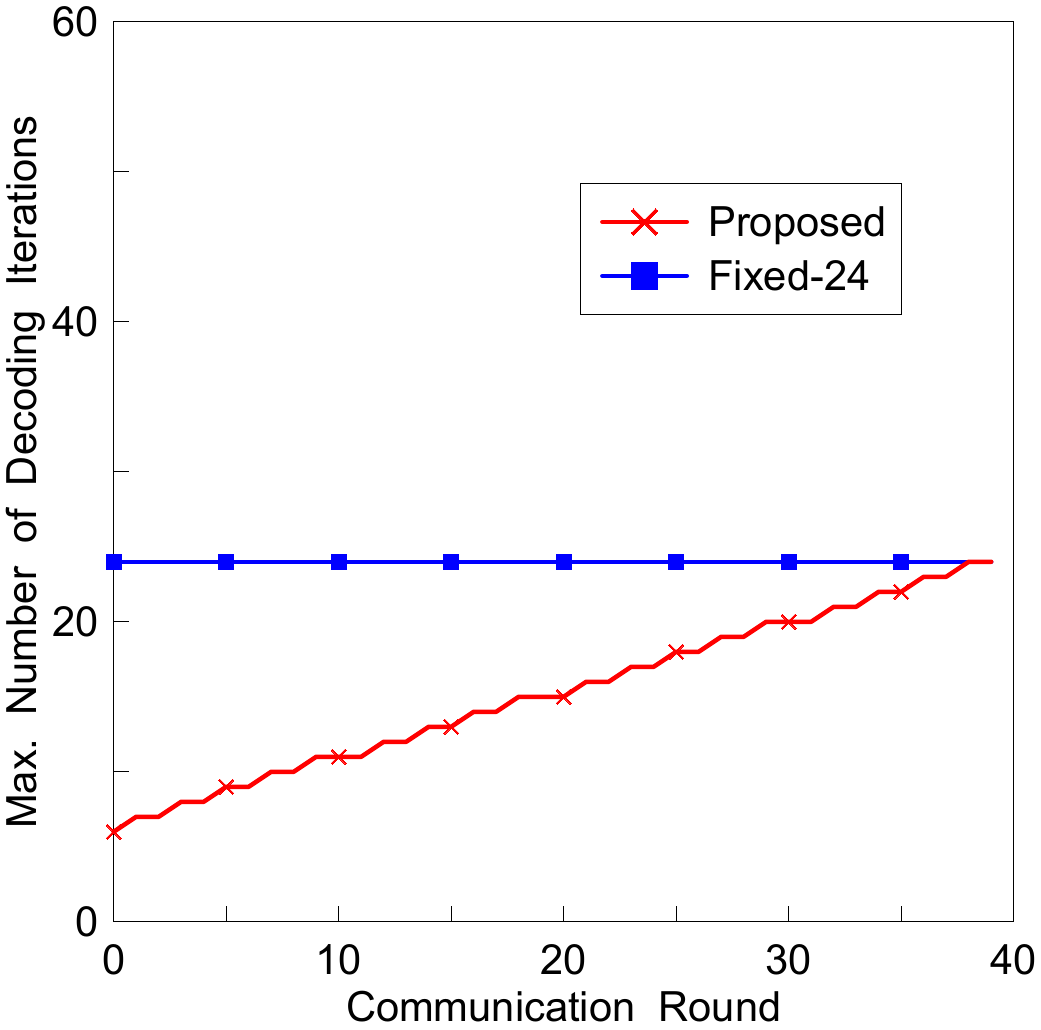}}
  \begin{center}
  (a) $Q_r$.
  \end{center}
\end{minipage}
\hfill
\begin{minipage}[b]{0.48\linewidth}
  \centering
  \centerline{\includegraphics[width=4.0cm]{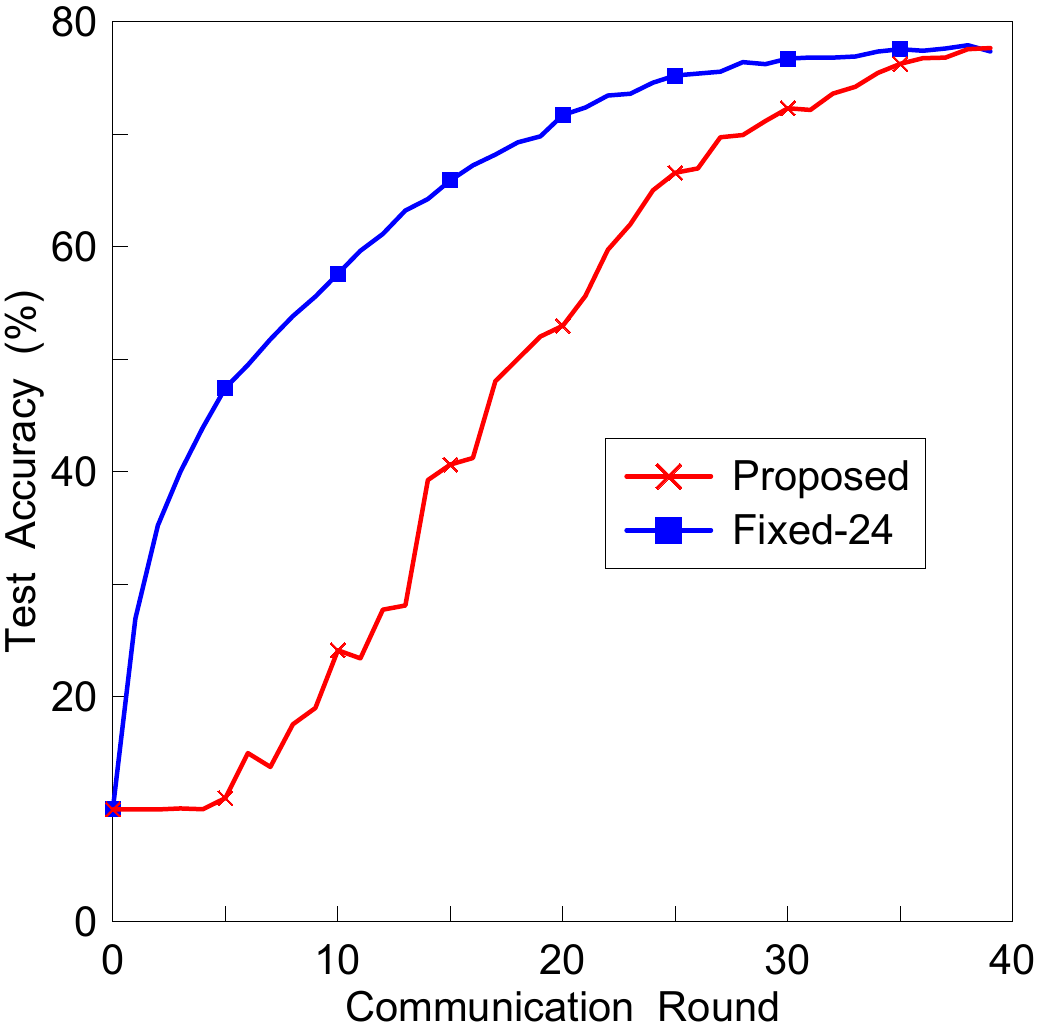}}
  \begin{center}
  (b) Test accuracy.
  \end{center}
\end{minipage}
\hfill
\begin{minipage}[b]{0.48\linewidth}
  \centering
  \centerline{\includegraphics[width=4.0cm]{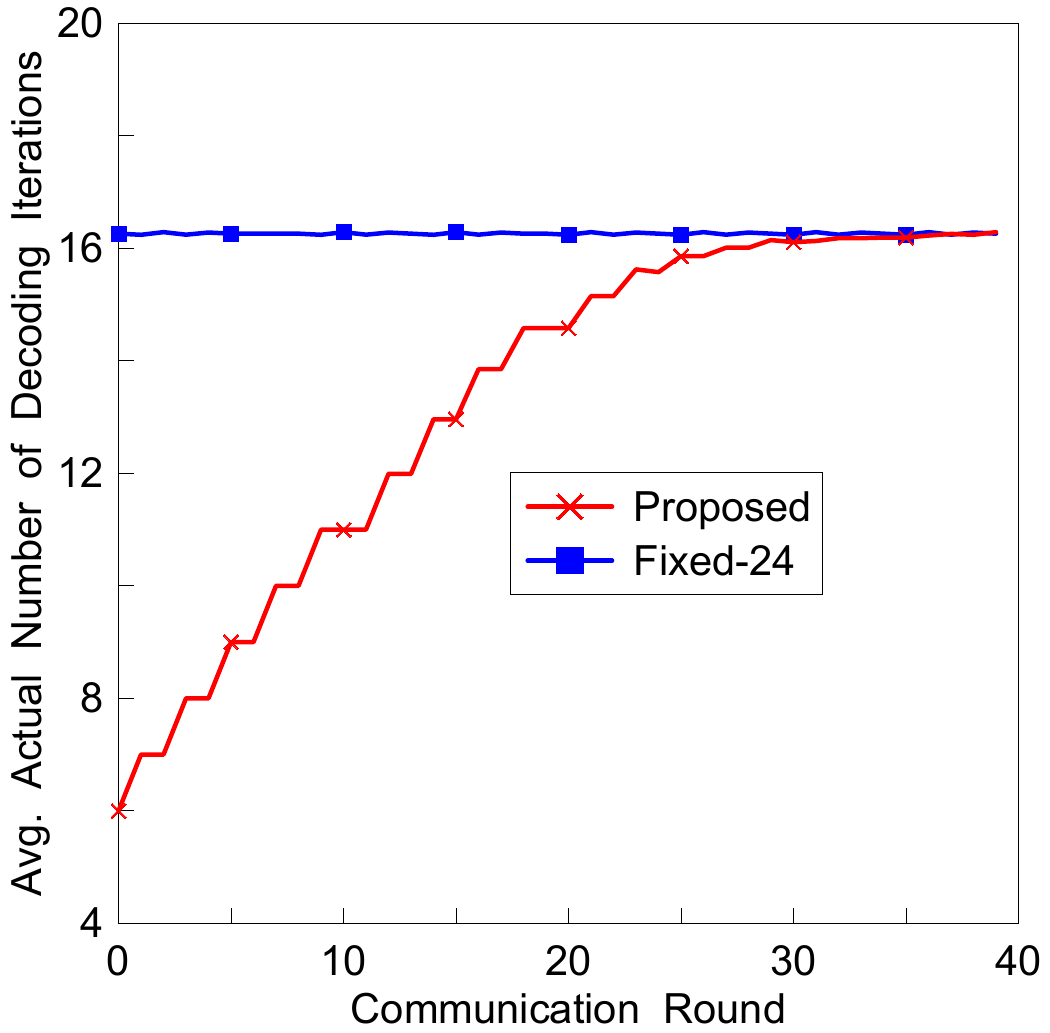}}
  \begin{center}
  (c) Avg. actual number of decoding iterations among clients.
  \end{center}
\end{minipage}
\hfill
\begin{minipage}[b]{0.48\linewidth}
  \centering
  \centerline{\includegraphics[width=4.0cm]{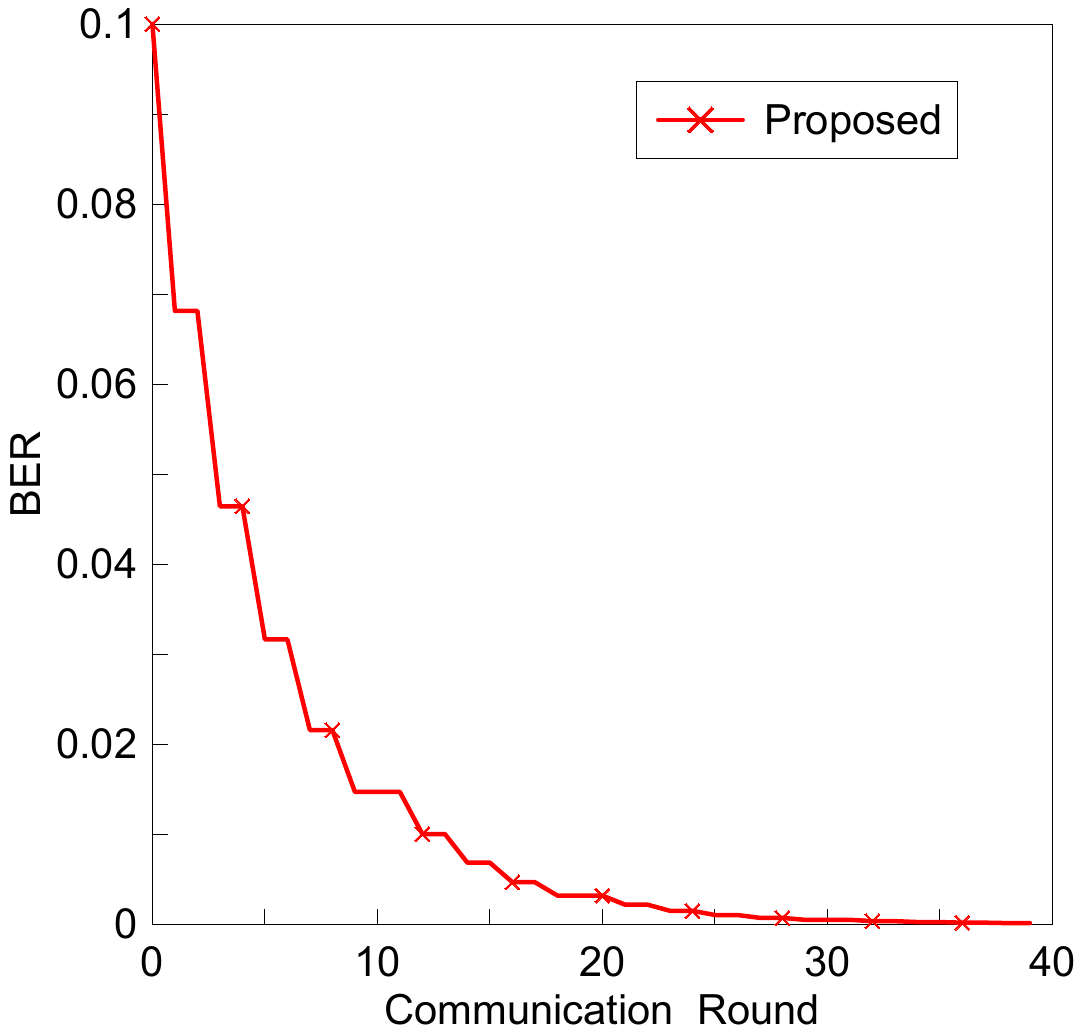}}
  \begin{center}
  (d) BER at every communication round, \\i.e., ($b_{r}$).
  \end{center}
\end{minipage}
\hfill
\begin{minipage}[b]{0.48\linewidth}
  \centering
  \centerline{\includegraphics[width=4.0cm]{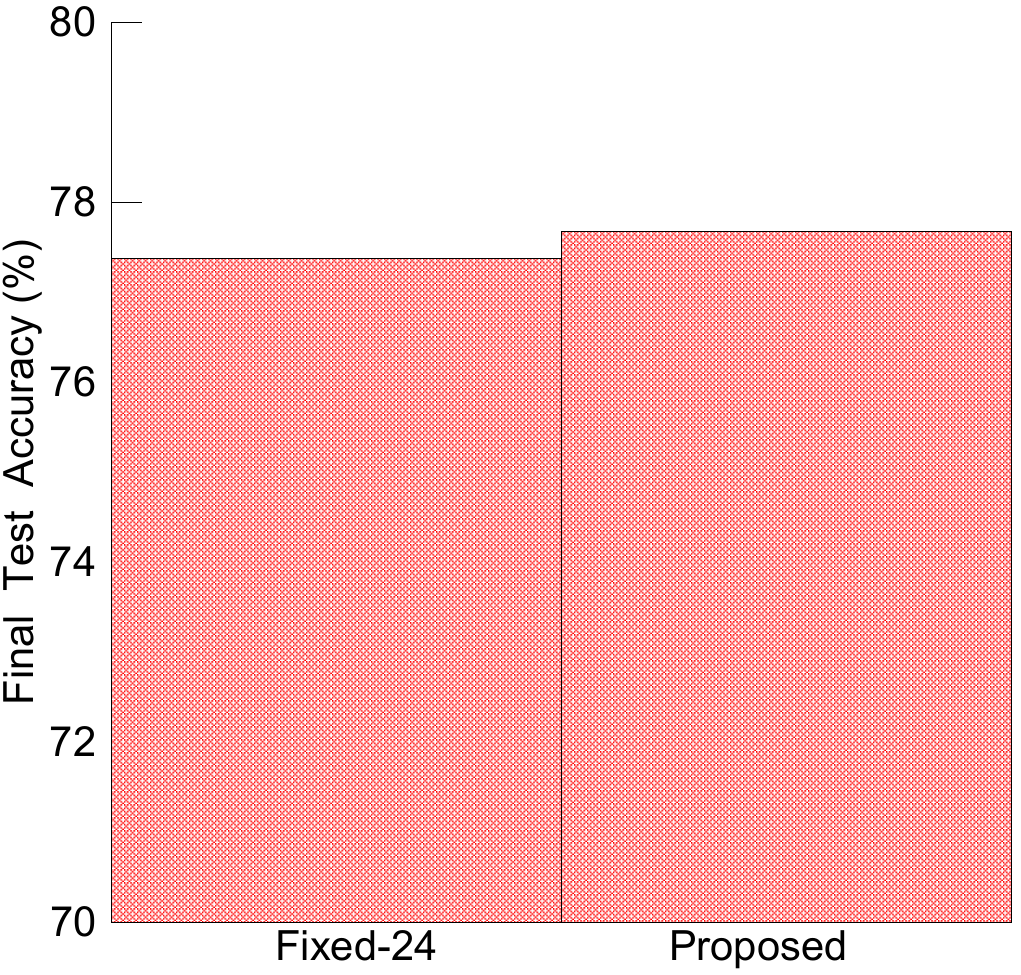}}
  \begin{center}
  (e) Test accuracy at the last communication round.
  \end{center}
\end{minipage}
\hfill
\begin{minipage}[b]{0.48\linewidth}
  \centering
  \centerline{\includegraphics[width=4.0cm]{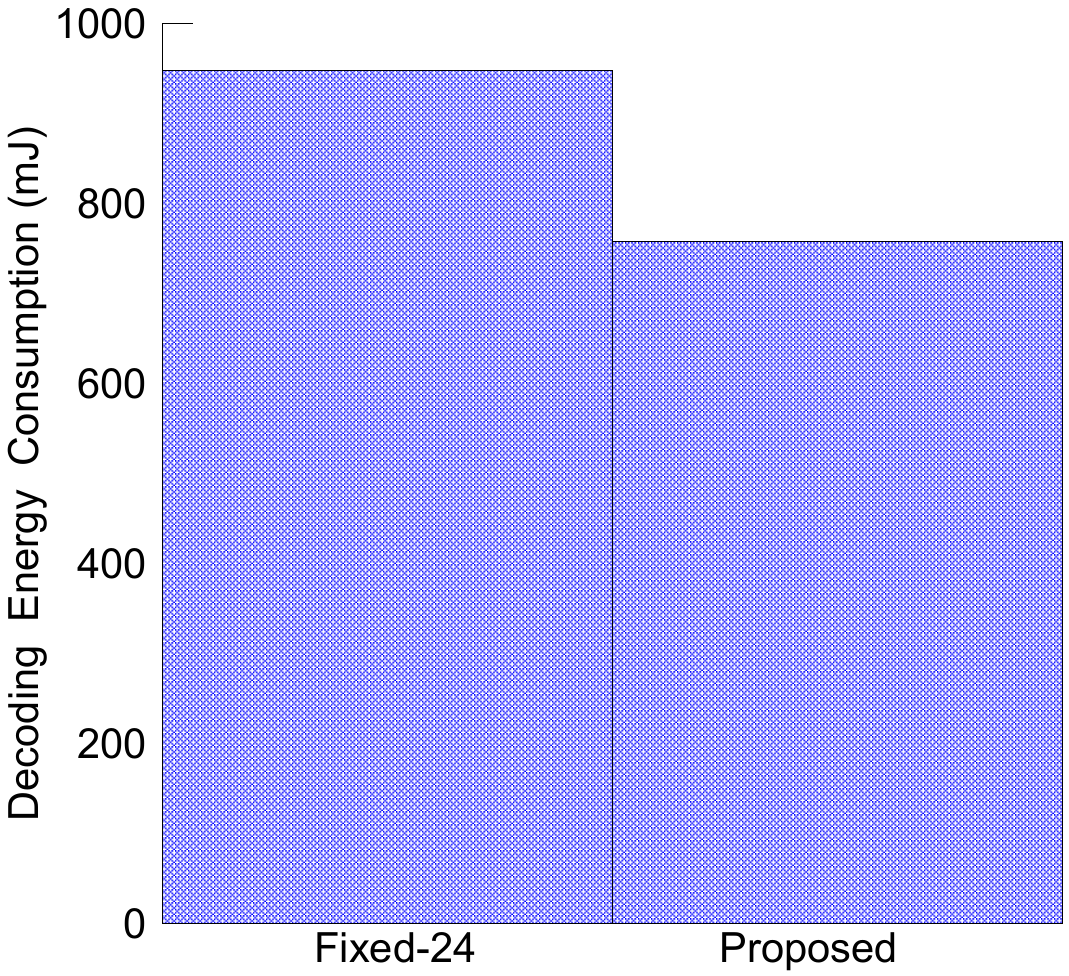}}
  \begin{center}
  (f) Total decoding energy consumption.
  \end{center}
\end{minipage}
\caption{Experimental results on CIFAR-10 (IID). Compared with the baseline scheme, our proposed algorithm that adaptively increases the maximum number of LDPC decoding iterations to 24 after 40 communication rounds can achieve a similar final test accuracy of 77.6\% and reduce the decoding energy by 191~mJ.}
\label{fig:CIFAR-10}
\end{figure}

\begin{figure}[t]
\small
\begin{minipage}[b]{.48\linewidth}
  \centering
  \centerline{\includegraphics[width=4.0cm]{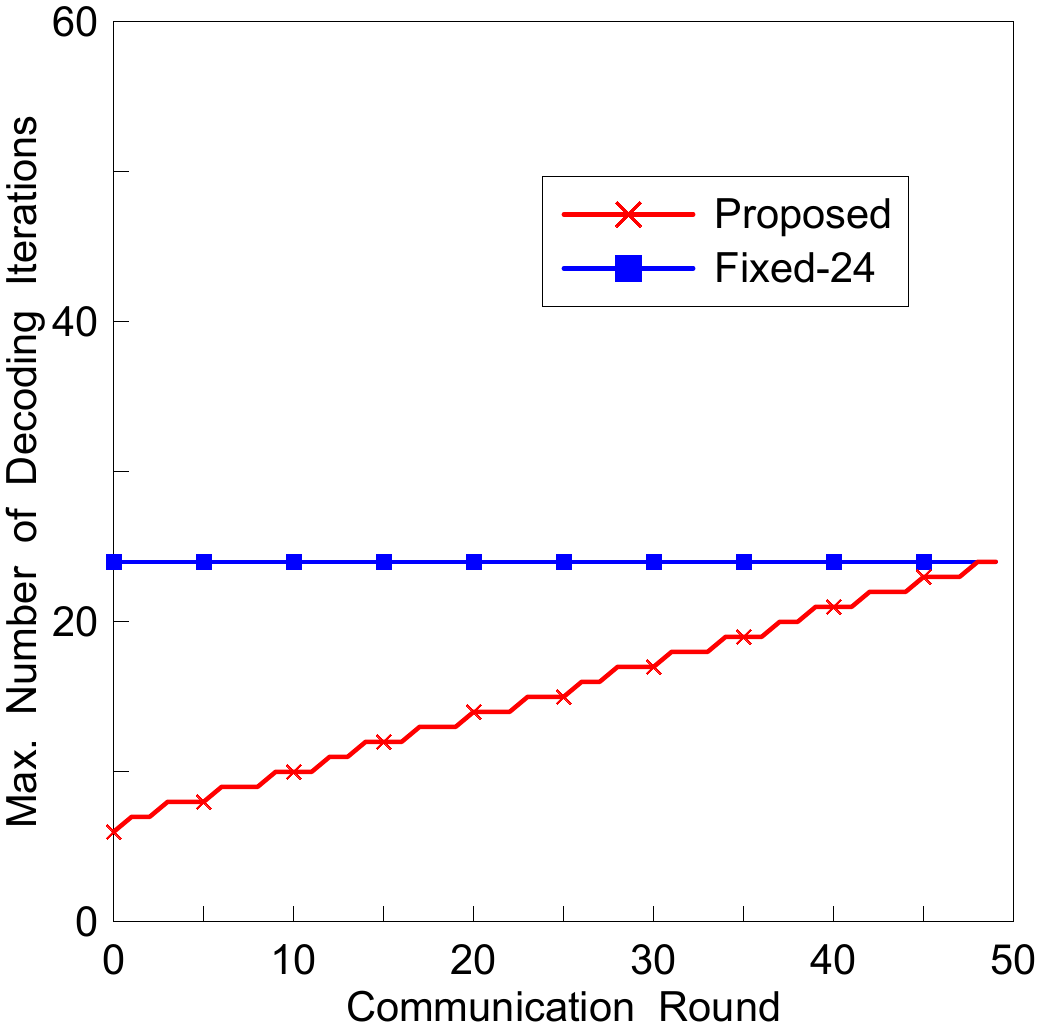}}
  \begin{center}
  (a) $Q_r$.
  \end{center}
\end{minipage}
\hfill
\begin{minipage}[b]{0.48\linewidth}
  \centering
  \centerline{\includegraphics[width=4.0cm]{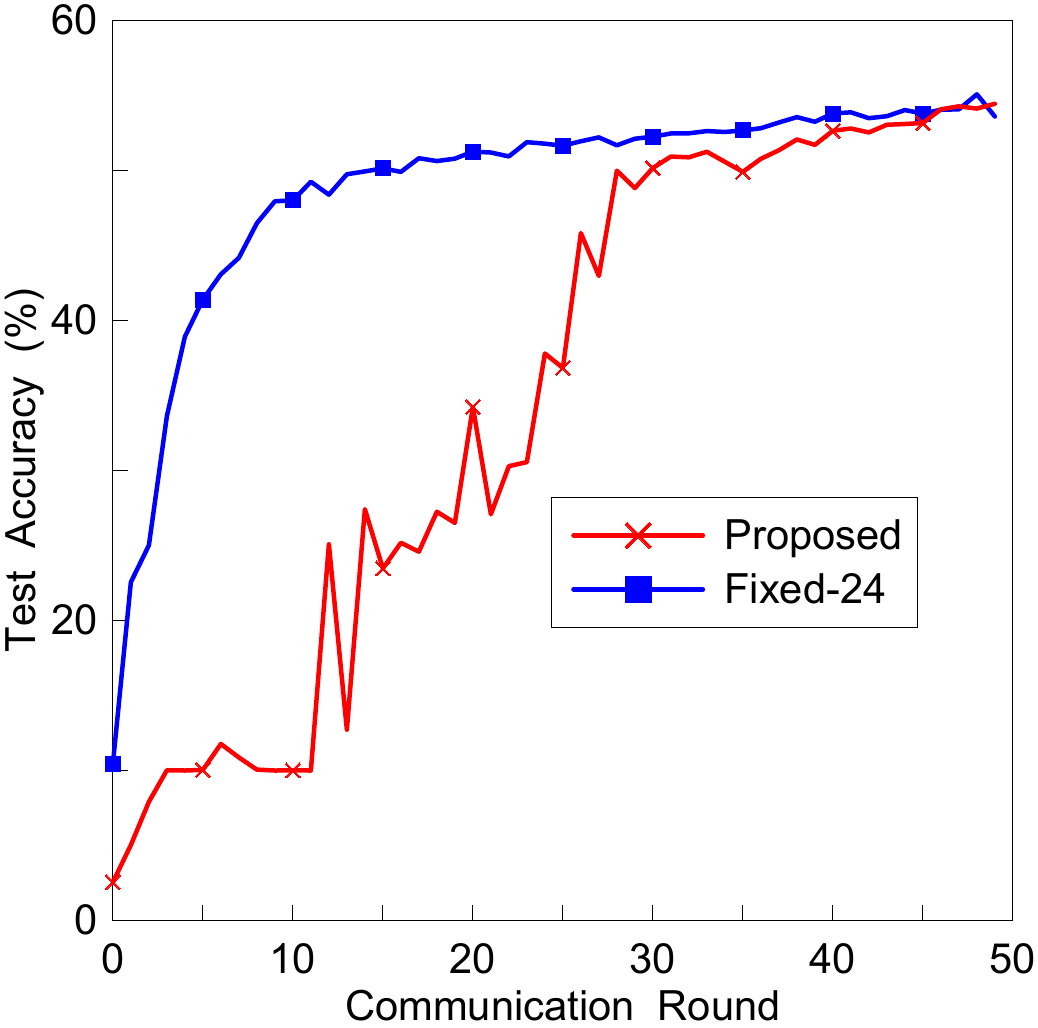}}
  \begin{center}
  (b) Test accuracy.
  \end{center}
\end{minipage}
\hfill
\begin{minipage}[b]{0.48\linewidth}
  \centering
  \centerline{\includegraphics[width=4.0cm]{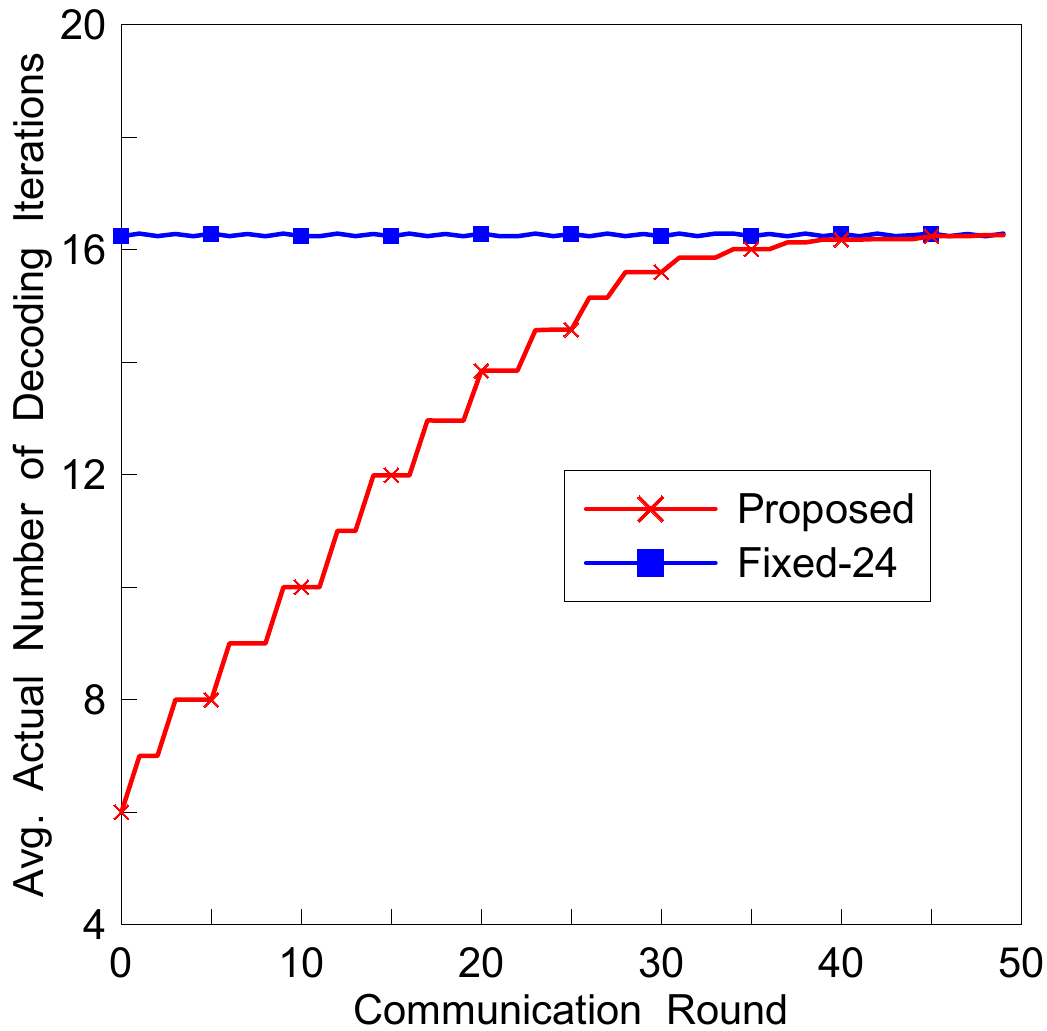}}
  \begin{center}
  (c) Avg. actual number of decoding iterations among clients.
  \end{center}
\end{minipage}
\hfill
\begin{minipage}[b]{0.48\linewidth}
  \centering
  \centerline{\includegraphics[width=4.0cm]{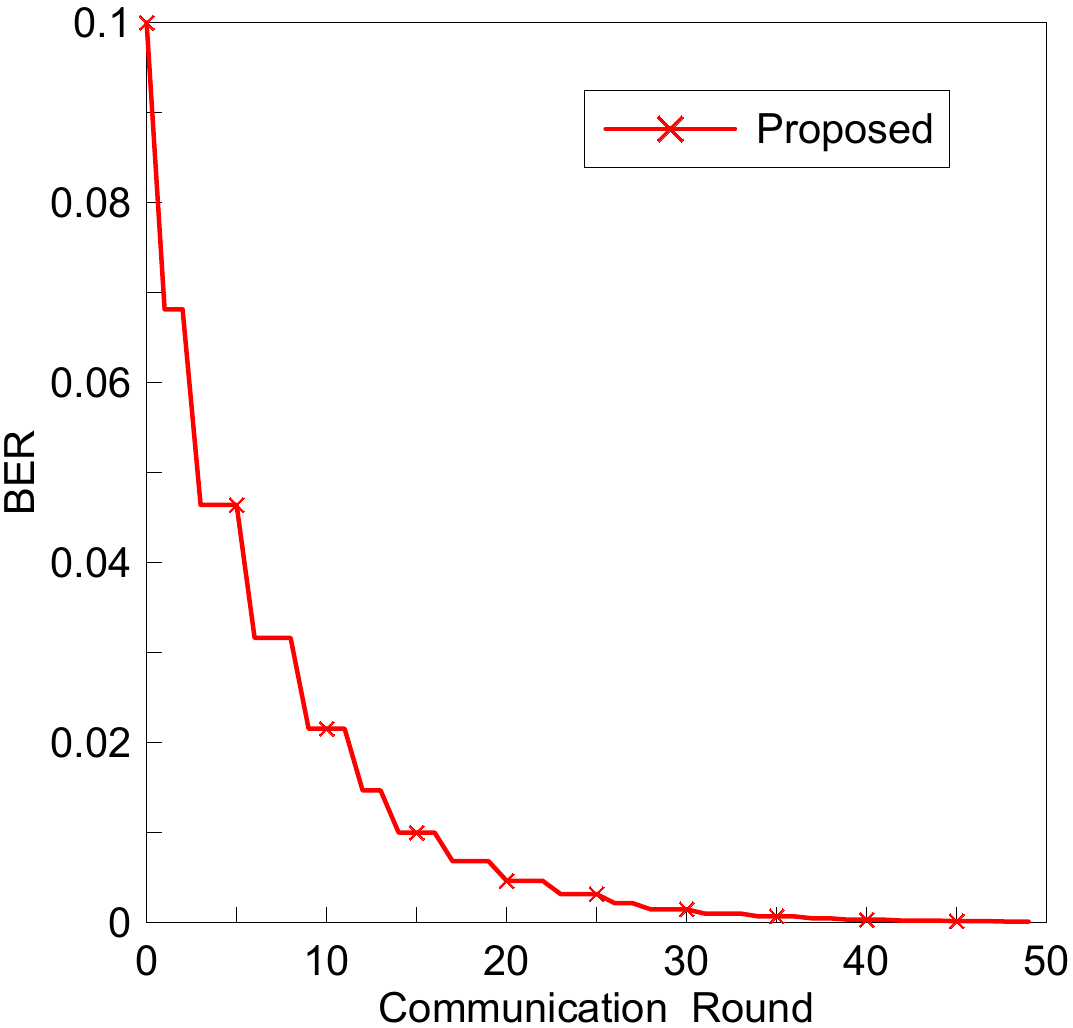}}
  \begin{center}
  (d) BER at every communication round, \\i.e., ($b_{r}$).
  \end{center}
\end{minipage}
\hfill
\begin{minipage}[b]{0.48\linewidth}
  \centering
  \centerline{\includegraphics[width=4.0cm]{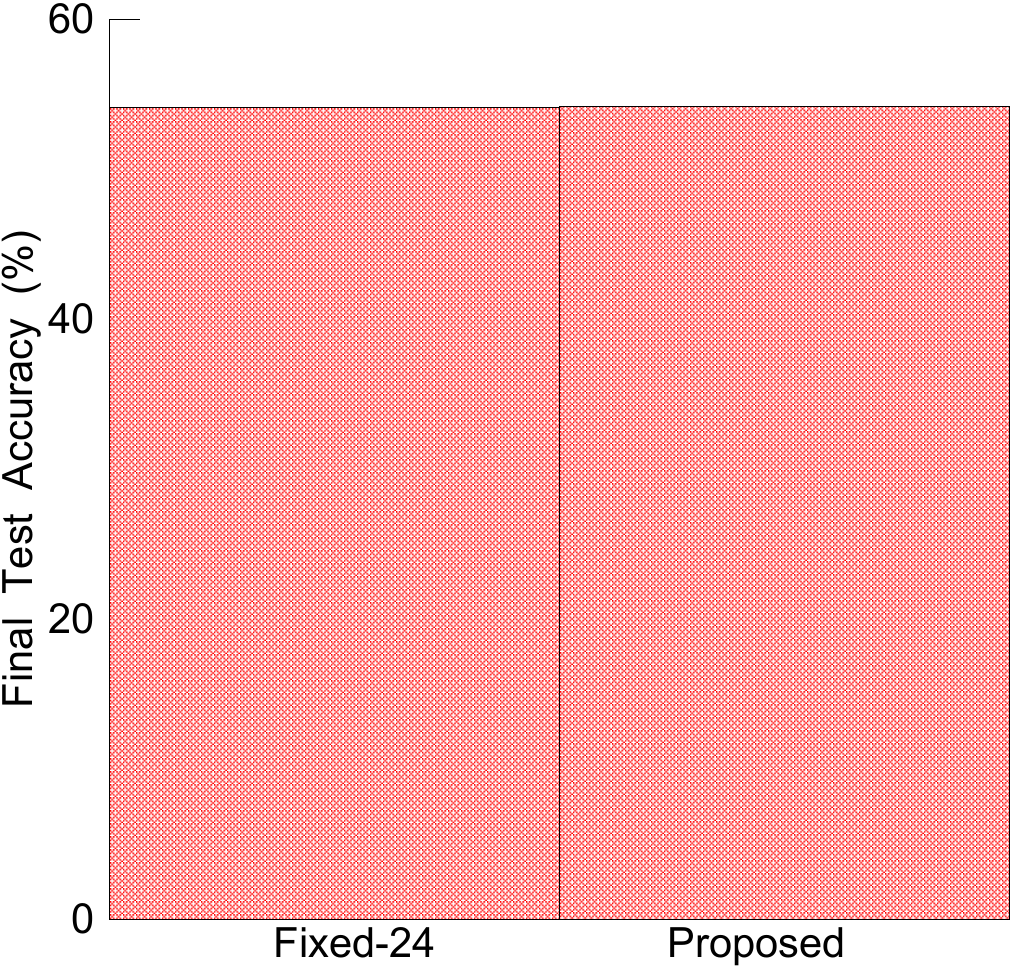}}
  \begin{center}
  (e) Test accuracy at the last communication round.
  \end{center}
\end{minipage}
\hfill
\begin{minipage}[b]{0.48\linewidth}
  \centering
  \centerline{\includegraphics[width=4.0cm]{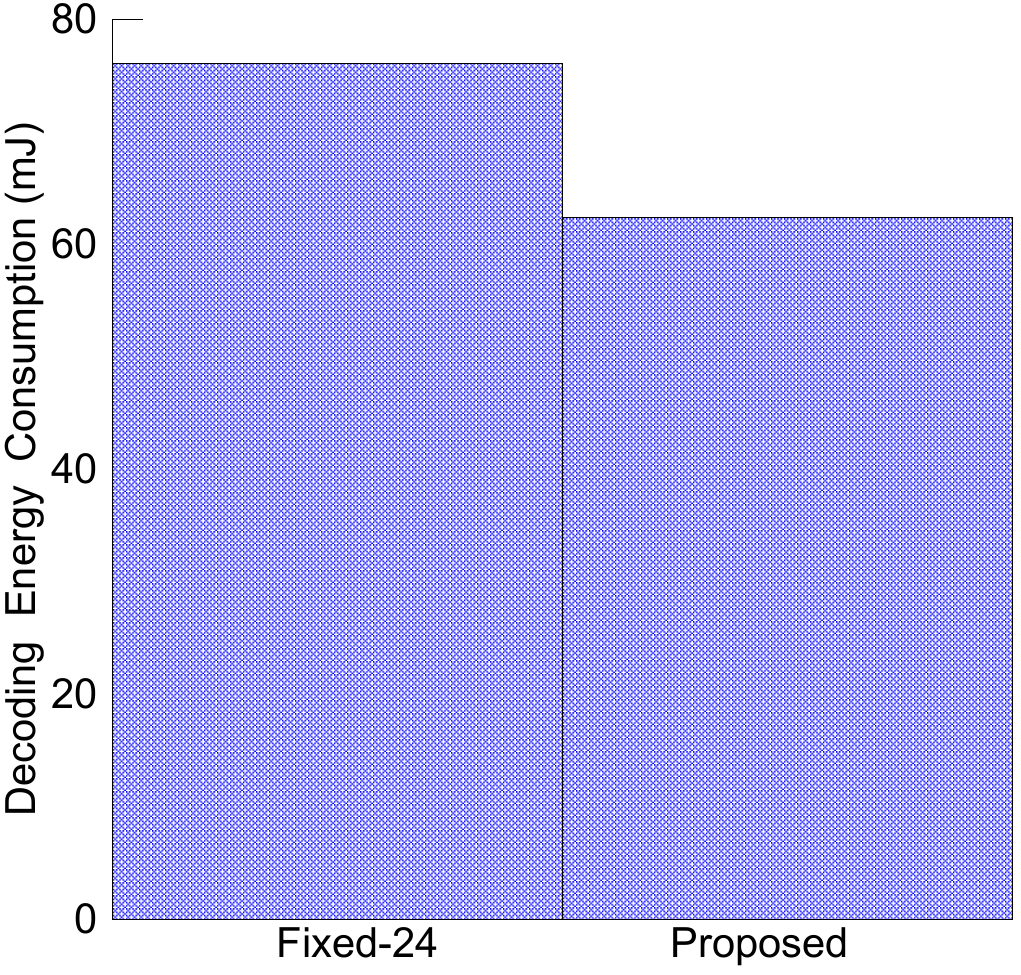}}
  \begin{center}
  (f) Total decoding energy consumption.
  \end{center}
\end{minipage}
\caption{Experimental results on Fashion-MNIST (non-IID). Compared with the baseline scheme, our proposed algorithm that adaptively increases the maximum number of LDPC decoding iterations to 24 after 50 communication rounds can achieve a similar final test accuracy of 54.1\% and reduce the decoding energy by 13.7~mJ.}
\label{fig:noniid-Fashion-MNIST}
\end{figure}
In this subsection, the proposed adaptive LDPC decoding scheme is compared to the baseline scheme.
Specifically, the learning process of the proposed scheme is assigned with an initial BER of $b_0=10^{-1}$ and a final BER of $b_{R-1}=10^{-4}$. The maximum number of LDPC decoding iterations is determined according to its mapping with BER as stated in Section~\ref{Energy Efficient LDPC Decoding Design}. In the experiments, the initial maximum numbers of LDPC decoding iterations are thus set to 6 and gradually increased to 24 at the last communication round. For the baseline ``Fixed-24'' scheme, a maximum of 24 LDPC decoding iterations can be used in each communication round, which corresponds to a BER of $10^{-4}$.

Fig.~\ref{fig:Consumed LDPC decoding iteration}(a) shows an example on the Fashion-MNIST (IID) dataset to illustrate the actual number of decoding iterations of digitalized model parameters at a sample client in a selected communication round. It can observed that only a few frames exhaust all the admissible LDPC decoding iterations, which is set to 52. Note that the decoding energy consumption is proportional to the actual number of decoding iterations. From the distribution of the actual number of LDPC decoding iterations shown in Fig.~\ref{fig:Consumed LDPC decoding iteration}(b), it is observed that 18 decoding iterations are sufficient for 80\% of the frames, implying the significant reduction of decoding energy consumption.

Fig.~\ref{fig:MNIST} summarizes the experimental results on the MNIST (IID) dataset. Specifically, Fig.~\ref{fig:MNIST}(a) shows how the maximum number of LDPC decoding iterations ($Q_{r}$) varies with communication rounds in different schemes. While the baseline scheme ``Fixed-24'' results in a constant BER, the proposed scheme achieves the desired BER by linearly increasing the maximum number of LDPC decoding iterations. We note that this behavior results from the empirical mapping between $Q_r$ and $b_r$ for the used LDPC decoder, which is different from the general theoretical guideline in (\ref{iter_round}).
Fig.~\ref{fig:MNIST}(b) shows the test accuracy achieved by the proposed scheme and the baseline ``Fixed-24'' scheme. 
Compared with the baseline scheme with a fixed number of 24 decoding iterations, the proposed scheme achieves the same test accuracy with 50 communication rounds despite with slower convergence. Fig.~\ref{fig:MNIST}(c) illustrates how the average actual number of decoding iterations at clients over all DNN parameters varies during the training process. The proposed scheme can reduce the average number of decoding iterations and thus the overall decoding energy consumption. Fig.~\ref{fig:MNIST}(d) depicts the BER (i.e., $b_{r}$) of the proposed scheme over communication rounds, which scales as $\mathcal{O}(1/(r+1)^2)$ from $10^{-1}$ to $10^{-4}$. Fig.~\ref{fig:MNIST}(e) shows the final test accuracy of different schemes at the 49-th communication round. It can be observed that the proposed scheme achieves the same accuracy as the baseline schemes (``Fixed-24'') that always maintains a very low BER of $10^{-4}$ throughout the training process. Fig.~\ref{fig:MNIST}(f) illustrates the overall decoding energy consumption at clients under different schemes, which is obtained by multiplying the total number of decoding iterations with the LDPC decoding energy efficiency reported in \cite{ldpc}. Compared to the baseline scheme, the proposed scheme significantly reduces the overall decoding energy consumption by 19.4\%.
\begin{table}[t]
\small
\caption{Summary of Experimental Results. Our scheme can achieve similar accuracy while saving around 20\% of the overall decoding energy consumption compared to the best baseline (``Fixed-$Q$") scheme.}
\centering
\begin{tabular}{|p{2.1cm}|c|c|c|c|}
\hline
&\multicolumn{4}{c|}{\textbf{Final Test Accuracy}} \\
\hline
&{Fixed-18}&{Fixed-24}&{Fixed-52}&{\textbf{Ours}}\\
\hline
{MNIST (IID)}&{96.74\%}&{97.96\%}&{98.06\%}&{98.02\%}\\
\hline
{Fashion-MNIST (IID)}&{86.18\%}&90.63\%&{91.24\%}&{91.16\%}\\
\hline  
{CIFAR-10 (IID)}&{76.72\%}&77.38\%&{77.57\%}&{77.68\%}\\
\hline
Fashion-MNIST (non-IID)&51.3\%&54.1\%&54.3\%&54.2\%\\
\hline
&\multicolumn{4}{c|}{\textbf{Total Decoding Energy Consumption (mJ)}}\\
\hline
&{Fixed-18}&{Fixed-24}&{Fixed-52}&{\textbf{Ours}}\\
\hline
{MNIST (IID)}&{33.87}&{34.51}&{34.68}&{27.81}\\
\hline
{Fashion-MNIST (IID)}&{74.23}&{76.10}&{76.29}&{62.17}\\
\hline
{CIFAR-10 (IID)}&{925.10}&{948.44}&{951.36}&{757.80}\\
\hline
Fashion-MNIST (non-IID)&{74.22}&{76.11}&{76.27}&{62.38}\\
\hline
\end{tabular}
\label{tab:data summary}
\end{table}

Figs.~\ref{fig:Fashion-MNIST}, \ref{fig:CIFAR-10}, and \ref{fig:noniid-Fashion-MNIST} present the experimental results on Fashion-MNIST (IID), CIFAR-10 (IID), and Fashion-MNIST (non-IID) datasets, respectively. The results are similar to those on the MNIST (IID) dataset in Fig.~\ref{fig:MNIST}, showing that adapting the maximum number of LDPC decoding iterations can achieve comparable or even slightly better accuracy compared to the baseline scheme. This observation was also reported in \cite{graves2011practical, hochreiter1994simplifying, neelakantan2015adding}, which suggests that adding annealed noise to the DNN model could enhance its generalization ability. Moreover, the proposed scheme reduces the LDPC decoding energy consumption on the Fashion-MNIST (IID), CIFAR-10 (IID), and Fashion-MNIST (non-IID) datasets by 18.3\%, 20.1\%, and 18.0\%, respectively. Table~\ref{tab:data summary} provides a summary of the experimental results. 

\section{Conclusion}
\label{Conclusion}
In this paper, we considered the energy consumption issue in wireless federated learning (FL) over digital communication networks. Our study recognized that channel decoding contributes to a large portion of the energy consumption at mobile clients in wireless FL systems. To achieve energy-efficient wireless FL, we first show that by properly controlling the bit error rate (BER) over different communication rounds, wireless FL can achieve the same convergence rate as FL with error-free communication. Based on this analysis, we proposed an energy-efficient LDPC decoding scheme to regulate the maximum number of decoding iterations over communication rounds. Experimental results demonstrated that the proposed method can reduce approximately 20\% of the total energy consumed by LDPC decoding, while maintaining the same or even a slightly higher learning accuracy compared to the scheme without adaptive decoding iteration control.

\section*{Appendix}
\label{Appendix}

\subsection{Proof of Lemma \ref{lemma_ber_distortion}}
\label{proof of model error}
We express a decoded model $\tilde{\mathbf{w}}$ in terms of $\mathbf{w}$ as follows:
\begin{equation}
\tilde{\mathbf{w}} = \mathbf{w} + \bm{\delta}\left(\mathbf{w}\right),
\end{equation}
where $\bm{\delta}\left(\mathbf{w}\right)$ is a $D$-dimensional random function of $\mathbf{w}$ denoting the model distortion caused by bit errors. Under Assumption \ref{assumption_ber}, each entry of $\bm{\delta}\left(\mathbf{w}\right)$ depends only on $w_{d}$, i.e., the $d$-th dimension of $\mathbf{w}$. Denote $\delta \left(w_d\right)$ as the $d$-th dimension of $\bm{\delta}\left(\mathbf{w}\right)$. Therefore, the $d$-th dimension of $\tilde{\mathbf{w}}$ is given as
\begin{equation}
\tilde{w}_{d} = w_{d} + \delta \left(w_{d}\right).
\end{equation}
Under Assumptions \ref{assumption_ber} and \ref{onebit_error}, the 1-bit error may occur at different bits, leading to different levels of model distortion. Let $w_{d}^{(i)} \in\{0,1\}$ be the $i$-th digit ($i=0$ and $N-1$ respectively denote the least and most significant bits) in the $N$-bit representation of model parameter $w_{d}$. Hence,
\begin{equation}
\mathbb{E}\left[\tilde{w}_{d} | w_{d} \right] = w_{d} + b \left(1-b\right)^{N-1}\Delta \sum_{i=0}^{N-1}  \left(1 -2 w_{d}^{(i)}\right) 2^{i},
\label{our_mean}
\end{equation}
where $\Delta = \frac{M_{\mathbf{w}}}{2^{N-1}}$. Note that the second term on the RHS of (21) is the non-zero expected bias caused by random bit errors. 
Thus, the mean square error of $w_{d}$ is given as follows:
\begin{equation}
\begin{split}
\mathbb{E}\left[||\tilde{w}_{d}-w_{d}||^2 |w_{d}\right] &= b \left(1-b\right)^{N-1} \sum_{i=0}^{N-1} \left(\Delta 2^{i}\right)^{2}\\
& = b \left(1-b\right)^{N-1} \frac{4^{N} -1}{3(2^{N}-1)^{2}} M_{\mathbf{w}}^{2}.
\end{split}
\label{our_var}
\end{equation} 
Therefore,
\begin{equation}
\begin{split}
\mathbb{E}\left[||\bm{\delta} (\mathbf{w})||^2 | \mathbf{w}\right] &=\mathbb{E}\left[||\mathbf{\tilde{w}}-\mathbf{w}||^2 | \mathbf{w}\right]
\\
&= \sum_{d=1}^D\mathbb{E}\left[||\tilde{w}_d-w_d||^2 | \mathbf{w} \right]
\\&= \frac{D(4^N-1)}{3(2^N-1)^2}\cdot b\left(1-b\right)^{N} M_{\mathbf{w}}^2.
\end{split}
\label{model_bound}
\end{equation}

\subsection{Proof of Lemma~\ref{thm_converge_bound}}
\label{proof_of_theorem_1}
We rewrite (\ref{server_update}) by substituting $\mathbf{w}_{r,0}^{k}=\tilde{\mathbf{w}}_{r}$ as
\begin{align}
\mathbf{w}_{r+1}&=\mathbf{w}_r+\frac{1}{K}\sum_{k=1}^{K}(\mathbf{w}_{r,E}^k-\mathbf{\tilde{w}}_r).
\label{add_notation_1}
\end{align}
To prove the Lemma 2, we first present Lemmas \ref{lemma_equiv_error}$\sim$5.

\begin{lemma}
Based on the definitions of $\mathbf{w}_{r+1}$ and $\mathbf{\bar{w}}_{r,E}$, we have 
\begin{equation}
\mathbf{w}_{r+1}-\mathbf{\bar{w}}_{r,E} =  \mathbf{w}_r-\mathbf{\tilde{w}}_r,\ r = 0,\cdots, R-1.
\label{lemma_equiv_error_formula}
\end{equation}
\label{lemma_equiv_error}
\end{lemma}
\begin{proof}
The proof is straightforward by substituting the definition of $\mathbf{\bar{w}}_{r,E}$ in (\ref{add_notation_1}) to the LHS of (\ref{lemma_equiv_error_formula}).
\end{proof}
\noindent Lemma \ref{lemma_equiv_error} translates the difference between $\mathbf{w}_{r+1}$ and $\mathbf{\bar{w}}_{r,E}$ to the model distortion incurred by  downlink bit errors.
\begin{lemma}
With Assumption 1, for $r=0,...,R-1$,
\begin{equation}
\begin{split}
\mathbb{E}f(\mathbf{w}_{r+1}) 
&\leq \mathbb{E}f(\mathbf{\bar{w}}_{r,E})+\frac{L}{2}\mathbb{E}||\mathbf{w}_{r+1}-\mathbf{\bar{w}}_{r,E}||^2\\
&-\mathbb{E} \left \langle\nabla f(\mathbf{\bar{w}}_{r,E}), \bm{\delta}(\mathbf{w}_{r})\right \rangle.
\end{split}
\label{bound_loss_global_formula}
\end{equation}
\label{bound_loss_global}
\end{lemma}
\begin{proof}
The proof can be obtained using the $L$-smoothness property of $f(\cdot)$, which is omitted for brevity.
\end{proof}
\noindent Lemma \ref{bound_loss_global} provides an upper bound on the expected value of the global loss function at $\mathbf{w}_{r+1}$.
\begin{lemma}
With Assumptions 1 - 3, for $r=0,\cdots,R-1$,
\begin{equation}
\begin{split}
&\mathbb{E}f(\mathbf{\bar{w}}_{r,E}) \leq \mathbb{E}f(\mathbf{w}_r) +\frac{L}{2}\mathbb{E}\Vert \bm{\delta} (\mathbf{w}_{r}) \Vert^2 +\mathbb{E}\left \langle\nabla f(\mathbf{w}_r), \bm{\delta}(\mathbf{w}_{r})\right \rangle \\
&-\frac{\eta}{2K}\left[1-L\eta-2L^2\eta^2E(E-1)\right]\sum_{k=1}^{K}\sum_{e=0}^{E-1}\mathbb{E}\Vert \nabla f(\mathbf{w}_{r,e}^k)\Vert^2 \\
&-\frac{\eta}{2}\sum_{e=0}^{E-1}\mathbb{E}\Vert \nabla f(\mathbf{\bar{w}}_{r,e})\Vert^2+\frac{LE\eta^2(\sigma_L^2+\sigma_G^2)}{2K} \\
& +\frac{L^2\eta^3\sigma_L^2(K+1)E(E-1)}{2K}.
\label{bound_loss_global_bit_error_formula}
\end{split}
\end{equation}
\label{bound_loss_global_bit_error}
\end{lemma}
\begin{proof}
The proof can be obtained following similar lines for Lemma 7 in \cite{reisizadeh2020fedpaq} by further considering the non-IID data distribution among clients (i.e., Assumption 3), which is omitted for brevity.
\end{proof}
\noindent Lemma \ref{bound_loss_global_bit_error} provides an upper bound on the expected value of the global loss function at $\bar{\mathbf{w}}_{r,E}$ in terms of the expected model distortion caused by downlink bit errors, i.e., $\mathbb{E}\Vert \bm{\delta} (\mathbf{w}_{r}) \Vert^2$.

By applying Lemmas \ref{lemma_equiv_error} and \ref{bound_loss_global_bit_error} to Lemma \ref{bound_loss_global}, we obtain the following inequality:
\begin{equation}
\begin{split}
&\mathbb{E}f(\mathbf{w}_{r+1}) \leq \mathbb{E}f(\mathbf{w}_r)+L\mathbb{E}\Vert \bm{\delta} (\mathbf{w}_{r}) \Vert^2 \\
&+\mathbb{E} \left\langle \nabla f(\mathbf{w}_r)-\nabla f(\mathbf{\bar{w}}_{r,E}), \bm{\delta}(\mathbf{w}_{r})\right\rangle -\frac{\eta}{2}\sum_{e=0}^{E-1}\mathbb{E}\Vert \nabla f(\mathbf{\bar{w}}_{r,e})\Vert^2 \\
&-\frac{\eta}{2K} \left[1-L\eta-2L^2\eta^2E(E-1)\right] \sum_{k=1}^{K}\sum_{e=0}^{E-1}\mathbb{E}\Vert \nabla f(\mathbf{w}_{r,e}^k)\Vert^2 \\
&+ \frac{LE\eta^2(\sigma_L^2+\sigma_G^2)}{2K} +\frac{L^2\eta^3\sigma_L^2(K+1)E(E-1)}{2K} .
\end{split}
\label{two_round}
\end{equation}
By leveraging  $2\left\langle \mathbf{a}, \mathbf{b} \right\rangle=\Vert \mathbf{a} \Vert^2+\Vert \mathbf{b} \Vert^2-\Vert \mathbf{a}-\mathbf{b}\Vert^2\leq \Vert \mathbf{a} \Vert^2+\Vert \mathbf{b} \Vert^2$ for any $\mathbf{a}$, $\mathbf{b}$, an upper bound for the third term on the RHS of (\ref{two_round}) is derived as follows:
\begin{equation}
\begin{split}
&\mathbb{E}\left\langle \nabla f(\mathbf{w}_r)-\nabla f(\mathbf{\bar{w}}_{r,E}), \bm{\delta}(\mathbf{w}_{r})\right\rangle \\
& \leq \frac{\mathbb{E} \Vert \nabla f(\mathbf{w}_r)-\nabla f(\mathbf{\bar{w}}_{r,E}) \Vert^2}{2}+\frac{\mathbb{E} \Vert \bm{\delta}(\mathbf{w}_{r}) \Vert ^2}{2}\\
&\overset{(\text{a})}{\leq} \frac{L^2}{2}\mathbb{E}\Vert \mathbf{\bar{w}}_{r,E}-\mathbf{w}_r \Vert^2+ \frac{\mathbb{E} \Vert \bm{\delta}(\mathbf{w}_{r}) \Vert ^2}{2}\\
&=\frac{L^2}{2} \mathbb{E}\bigg \Vert \frac{1}{K}\sum_{k =1}^{K} (\mathbf{w}_{r,E}^k -\mathbf{w}_r) \bigg \Vert^2+\frac{\mathbb{E} \Vert \bm{\delta}(\mathbf{w}_{r}) \Vert ^2}{2} \\
&=\frac{L^2}{2} \mathbb{E}\bigg \Vert \frac{1}{K}\sum_{k =1}^{K} (\mathbf{w}_{r,E}^k -\tilde{\mathbf{w}}_r) + \tilde{\mathbf{w}}_r - \mathbf{w}_r \bigg \Vert^2+\frac{\mathbb{E} \Vert \bm{\delta}(\mathbf{w}_{r}) \Vert ^2}{2} \\
&\overset{(\text{b})}{=}\frac{L^2}{2} \mathbb{E}\bigg \Vert \frac{\eta}{K}\sum_{k =1}^{K} \sum_{e=0}^{E-1} \nabla f_k (\mathbf{w}_{r,e}^{k};\xi_{r,e}^{k}) + \tilde{\mathbf{w}}_r - \mathbf{w}_r \bigg \Vert^2\\
&\ \ \ +\frac{\mathbb{E} \Vert \bm{\delta}(\mathbf{w}_{r}) \Vert ^2}{2}		
\end{split}
\label{inner_a}
\end{equation}
where (a) is due to the $L$-smoothness property of $f(\cdot)$ and (b) applies the local updating rules in (\ref{receive}).
Based on the Assumptions \ref{assumption_unbiased} and \ref{assumption_dissimilarity}, we have
\begin{equation}
\begin{split}
\mathbb{E} \Vert \nabla & f_k(\mathbf{w}_{r,e}^k  ; \xi_{r,e}^{k})  \Vert^2 = \mathbb{E} \Vert\nabla f_k(\mathbf{w}_{r,e}^k ; \xi_{r,e}^{k}) \\
&-\nabla f_k(\mathbf{w})+\nabla f_k(\mathbf{w})-\nabla f(\mathbf{w})+\nabla f(\mathbf{w})\Vert^2\\
&\leq \sigma_L^2+\sigma_G^2+ \mathbb{E} \Vert \nabla f(\mathbf{w})\Vert^2.
\end{split}
\label{sigma_L_G}
\end{equation}
By combing (\ref{sigma_L_G}) and (\ref{inner_a}), the bound in (\ref{inner_a}) can be further updated as
\begin{equation}
\begin{split}
&\mathbb{E}\left\langle \nabla f(\mathbf{w}_r)-\nabla f(\mathbf{\bar{w}}_{r,E}), \bm{\delta}(\mathbf{w}_{r})\right\rangle \\
&\leq \frac{L^2\eta^{2}}{2K^{2}} \sum_{k =1}^{K} \sum_{e=0}^{E-1} \mathbb{E} \Vert \nabla f_k (\mathbf{w}_{r,e}^{k};\xi_{r,e}^{k}) \Vert ^2 + \frac{L^{2}}{2}\mathbb{E}  \Vert \tilde{\mathbf{w}}_r - \mathbf{w}_r  \Vert^2\\
&\ \ \ +\frac{\mathbb{E} \Vert \bm{\delta}(\mathbf{w}_{r}) \Vert ^2}{2} \\
&\leq \frac{L^2\eta^{2}}{2K^{2}} \sum_{k =1}^{K} \sum_{e=0}^{E-1} \mathbb{E} \Vert \nabla f (\mathbf{w}_{r,e}^{k}) \Vert ^2 + \frac{L^{2}\eta^{2}E(\sigma_{G}^{2}+\sigma_{L}^{2})}{2K}\\
&\ \ \ +\frac{(L^2+1)\mathbb{E} \Vert \bm{\delta}(\mathbf{w}_{r}) \Vert ^2}{2}.
\label{inner}
\end{split}
\end{equation}
 
\noindent Substituting (\ref{inner}) in (\ref{two_round}), we have
\begin{equation}
\begin{split}
&\mathbb{E}f(\mathbf{w}_{r+1}) \leq \mathbb{E}f(\mathbf{w}_r) + \frac{(L+1)^2\mathbb{E} \Vert \bm{\delta}(\mathbf{w}_{r}) \Vert ^2}{2} \\
&-\frac{\eta}{2K} \left[1-\frac{K+L}{K}L\eta-2L^2\eta^2E(E-1)\right] \\
&\cdot \sum_{k=1}^{K}\sum_{e=0}^{E-1}\mathbb{E}\Vert \nabla f(\mathbf{w}_{r,e}^k)\Vert^2 -\frac{\eta}{2}\sum_{e=0}^{E-1}\mathbb{E}\Vert \nabla f(\mathbf{\bar{w}}_{r,e})\Vert^2 \\
&+\frac{L(L+1)E\eta^2(\sigma_L^2+\sigma_G^2)}{2K} +\frac{L^2\eta^3\sigma_L^2(K+1)E(E-1)}{2K} .
\end{split}
\label{two_round_full}
\end{equation}
For any $\eta$ that satisfies (\ref{smalleta}), we have
\begin{equation}
\begin{split}
&\mathbb{E}f(\mathbf{w}_{r+1}) \leq \mathbb{E}f(\mathbf{w}_r)+ \frac{(L+1)^2\mathbb{E} \Vert \bm{\delta}(\mathbf{w}_{r}) \Vert ^2}{2}  \\
&-\frac{\eta}{2}\sum_{e=0}^{E-1}\mathbb{E}\Vert \nabla f(\mathbf{\bar{w}}_{r,e})\Vert^2 \\
&+\frac{L(L+1)E\eta^2(\sigma_L^2+\sigma_G^2)}{2K} +\frac{L^2\eta^3\sigma_L^2(K+1)E(E-1)}{2K} .
\end{split}
\label{tosum}
\end{equation}
Summing (\ref{tosum}) over $r=0,\cdots,R-1$, rearranging the terms, and multiplying $\frac{2}{\eta RE}$ on both sides, we have
\begin{equation}
\begin{split}
&\frac{1}{RE}\sum_{r=0}^{R-1}\sum_{e=0}^{E-1}\mathbb{E}\Vert \nabla f(\mathbf{\bar{w}}_{r,e})\Vert^2 \leq \frac{2(\mathbb{E} f(\mathbf{w}_0)-\mathbb{E} f(\mathbf{w}_{R}))}{\eta RE} \\
&+\frac{(L+1)^2}{\eta RE}\sum_{r=0}^{R-1}\mathbb{E} \Vert\bm{\delta}(\mathbf{w}_{r})\Vert^{2} \\
&+\frac{L(L+1)\eta(\sigma_L^2+\sigma_G^2)}{K}+\frac{L^2\eta^2\sigma_L^2(K+1)(E-1)}{K} .
\end{split}
\label{our_full}
\end{equation}
By substituting (\ref{model_bound}) into (\ref{our_full}), and denote $f^{*}$ as the minimum global training loss, we have
\begin{align}
&\frac{1}{RE}\sum_{r=0}^{R-1}\sum_{e=0}^{E-1}\mathbb{E}\Vert \nabla f(\mathbf{\bar{w}}_{r,e})\Vert^2 \leq \frac{2(\mathbb{E} f(\mathbf{w}_0)-f^*)}{\eta RE} \nonumber\\
&+\frac{(L+1)^2D(4^N-1)}{3(2^N-1)^2\eta RE}\sum_{r=0}^{R-1}b_r(1-b_r)^{N-1}M_r^2 \nonumber\\
&+\frac{L(L+1)\eta(\sigma_L^2+\sigma_G^2)}{K}+\frac{L^2\eta^2\sigma_L^2(K+1)(E-1)}{K}.
\label{convergence_ber}
\end{align}
Finally, based on Assumption~{\ref{range_bound}}, we are able to show (\ref{convergence_ber_bound}) in Lemma~\ref{thm_converge_bound}.

\bibliography{IEEEabrv,refs}

\begin{thebibliography}{10}
\providecommand{\url}[1]{#1}
\csname url@samestyle\endcsname
\providecommand{\newblock}{\relax}
\providecommand{\bibinfo}[2]{#2}
\providecommand{\BIBentrySTDinterwordspacing}{\spaceskip=0pt\relax}
\providecommand{\BIBentryALTinterwordstretchfactor}{4}
\providecommand{\BIBentryALTinterwordspacing}{\spaceskip=\fontdimen2\font plus
\BIBentryALTinterwordstretchfactor\fontdimen3\font minus \fontdimen4\font\relax}
\providecommand{\BIBforeignlanguage}[2]{{%
\expandafter\ifx\csname l@#1\endcsname\relax
\typeout{** WARNING: IEEEtran.bst: No hyphenation pattern has been}%
\typeout{** loaded for the language `#1'. Using the pattern for}%
\typeout{** the default language instead.}%
\else
\language=\csname l@#1\endcsname
\fi
#2}}
\providecommand{\BIBdecl}{\relax}
\BIBdecl

\bibitem{qu2024robust}
L.~Qu, S.~Song, C.-Y. Tsui, and Y.~Mao, ``How robust is federated learning to communication error? a comparison study between uplink and downlink channels,'' in \emph{Proc. IEEE Wireless Commun. Netw. Conf. (WCNC)}, Dubai, UAE, April. 2024.

\bibitem{lecun2015deep}
Y.~LeCun, Y.~Bengio, and G.~Hinton, ``Deep learning,'' \emph{Nat.}, vol. 521, no. 7553, pp. 436--444, May. 2015.

\bibitem{yang2019federated}
Q.~Yang, Y.~Liu, T.~Chen, and Y.~Tong, ``Federated machine learning: Concept and applications,'' \emph{ACM Trans. Intell. Syst. Technol.}, vol.~10, no.~2, pp. 1--19, Jan. 2019.

\bibitem{mcmahan2017communication}
B.~McMahan, E.~Moore, D.~Ramage, S.~Hampson, and B.~A. y~Arcas, ``Communication-efficient learning of deep networks from decentralized data,'' in \emph{Proc. Int. Conf. Artif. Intell. Stat. (AISTATS)}, Fort Lauderdale, FL, USA, Apr. 2017.

\bibitem{chai2019towards}
Z.~Chai, H.~Fayyaz, Z.~Fayyaz, A.~Anwar, Y.~Zhou, N.~Baracaldo, H.~Ludwig, and Y.~Cheng, ``Towards taming the resource and data heterogeneity in federated learning.'' in \emph{Proc. USENIX Conf. Oper. Mach. Learn. (OpML)}, Santa Clara, CA, USA, May. 2019.

\bibitem{xie2023fedkl}
Z.~Xie and S.~Song, ``{FedKL: Tackling data heterogeneity in federated reinforcement learning by penalizing KL divergence},'' \emph{IEEE J. Sel. Areas Commun.}, vol.~41, no.~4, pp. 1227--1242, Apr. 2023.

\bibitem{chen2021communication}
M.~Chen, N.~Shlezinger, H.~V. Poor, Y.~C. Eldar, and S.~Cui, ``Communication-efficient federated learning,'' \emph{Proc. Natl. Acad. Sci.}, vol. 118, no.~17, p. e2024789118, Mar. 2021.

\bibitem{mao2023green}
Y.~Mao, X.~Yu, K.~Huang, Y.-J.~A. Zhang, and J.~Zhang, ``{Green Edge AI: A Contemporary Survey},'' \emph{Proc. IEEE}, to appear.

\bibitem{kim2021autofl}
Y.~G. Kim and C.-J. Wu, ``{AutoFL: Enabling heterogeneity-aware energy efficient federated learning},'' in \emph{Proc. IEEE/ACM Int. Symp. Micro. (MICRO)}, New York, NY, USA, Oct. 2021.

\bibitem{tran2019federated}
N.~H. Tran, W.~Bao, A.~Zomaya, M.~N. Nguyen, and C.~S. Hong, ``Federated learning over wireless networks: Optimization model design and analysis,'' in \emph{Proc. IEEE Conf. Comput. Commun. (INFOCOM)}, Paris, France, Apr. 2019.

\bibitem{yang2020energy}
Z.~Yang, M.~Chen, W.~Saad, C.~S. Hong, and M.~Shikh-Bahaei, ``Energy efficient federated learning over wireless communication networks,'' \emph{IEEE Trans. Wireless Commun.}, vol.~20, no.~3, pp. 1935--1949, Mar. 2021.

\bibitem{chen2020joint}
M.~Chen, Z.~Yang, W.~Saad, C.~Yin, H.~V. Poor, and S.~Cui, ``A joint learning and communications framework for federated learning over wireless networks,'' \emph{IEEE Trans. Wireless Commun.}, vol.~20, no.~1, pp. 269--283, Jan. 2021.

\bibitem{han2020adaptive}
P.~Han, S.~Wang, and K.~K. Leung, ``Adaptive gradient sparsification for efficient federated learning: An online learning approach,'' in \emph{Proc. Int. Conf. Distrib. Comput. Syst. (ICDCS)}, Singapore, Nov. 2020.

\bibitem{zheng2020design}
S.~Zheng, C.~Shen, and X.~Chen, ``Design and analysis of uplink and downlink communications for federated learning,'' \emph{IEEE J. Sel. Areas Commun.}, vol.~39, no.~7, pp. 2150--2167, Jul. 2020.

\bibitem{qu2022feddq}
L.~Qu, S.~Song, and C.-Y. Tsui, ``{FedDQ: Communication-efficient federated learning with descending quantization},'' in \emph{Proc. IEEE Global Commun. Conf. (Globecom)}, Rio de Janeiro, Brazil, Dec. 2022.

\bibitem{reisizadeh2020fedpaq}
A.~Reisizadeh, A.~Mokhtari, H.~Hassani, A.~Jadbabaie, and R.~Pedarsani, ``{FedPAQ: A communication-efficient federated learning method with periodic averaging and quantization},'' in \emph{Proc. Int. Conf. Artif. Intell. Stat. (AISTATS)}, Palermo, Sicily, Italy, Jun. 2020.

\bibitem{qiu2022zerofl}
X.~Qiu, J.~Fernandez-Marques, P.~P. Gusmao, Y.~Gao, T.~Parcollet, and N.~D. Lane, ``Zerofl: Efficient on-device training for federated learning with local sparsity,'' \emph{arXiv:2208.02507}, 2022.

\bibitem{kim2023green}
M.~Kim, W.~Saad, M.~Mozaffari, and M.~Debbah, ``Green, quantized federated learning over wireless networks: An energy-efficient design,'' \emph{IEEE Trans. Wireless Commun.}, vol.~23, no.~2, pp. 1386--1402, Feb. 2024.

\bibitem{tong2022nine}
W.~Tong and G.~Y. Li, ``{Nine challenges in artificial intelligence and wireless communications for 6G},'' \emph{IEEE Wireless Commun.}, vol.~29, no.~4, pp. 140--145, Aug. 2022.

\bibitem{krizhevsky2012imagenet}
A.~Krizhevsky, I.~Sutskever, and G.~E. Hinton, ``Imagenet classification with deep convolutional neural networks,'' in \emph{Proc. Adv. Neural Inf. Process. Syst. (NeurIPS)}, Lake Tahoe, Nevada, USA, Dec. 2012.

\bibitem{aichip}
S.~Choi, J.~Sim, M.~Kang, Y.~Choi, H.~Kim, and L.-S. Kim, ``An energy-efficient deep convolutional neural network training accelerator for in situ personalization on smart devices,'' \emph{IEEE J. Solid-State Circuits}, vol.~55, no.~10, pp. 2691--2702, Oct. 2020.

\bibitem{trx}
X.~Huang, H.~Jia, S.~Dong, W.~Deng, Z.~Wang, and B.~Chi, ``{A 24-30GHz 4-element phased array transceiver with low insertion loss compact T/R switch and bidirectional phase shifter in 65 nm CMOS technology},'' in \emph{Proc. IEEE Asian Solid-State Circuits Conf. (ASSCC)}, Busan, Korea, Nov. 2021.

\bibitem{ldpc}
W.~Tang, C.-H. Chen, and Z.~Zhang, ``{A 2.4-mm$^2$ 130-mW MMSE-nonbinary LDPC iterative detector decoder for 4$\times$4 256-QAM MIMO in 65-nm CMOS},'' \emph{IEEE J. Solid-State Circuits}, vol.~54, no.~7, pp. 2070--2080, Jul. 2019.

\bibitem{koike2015iteration}
T.~Koike-Akino, D.~S. Millar, K.~Kojima, K.~Parsons, Y.~Miyata, K.~Sugihara, and W.~Matsumoto, ``{Iteration-aware LDPC code design for low-power optical communications},'' \emph{IEEE J. Light. Techn.}, vol.~34, no.~2, pp. 573--581, Jan. 2016.

\bibitem{nam2021early}
K.~Nam-Il and K.~Jin-Up, ``{Early termination scheme for 5G NR LDPC code},'' in \emph{Proc. Int. Conf. Inf. Commun. Techn. Converg. (ICTC)}, Jeju Island, Korea, Oct. 2021.

\bibitem{ang2020robust}
F.~Ang, L.~Chen, N.~Zhao, Y.~Chen, W.~Wang, and F.~R. Yu, ``Robust federated learning with noisy communication,'' \emph{IEEE Trans. Commun.}, vol.~68, no.~6, pp. 3452--3464, Jun. 2020.

\bibitem{amiri2021convergence}
M.~M. Amiri, D.~G{\"u}nd{\"u}z, S.~R. Kulkarni, and H.~V. Poor, ``Convergence of federated learning over a noisy downlink,'' \emph{IEEE Trans. Wireless Commun.}, vol.~21, no.~3, pp. 1422--1437, Mar. 2022.

\bibitem{amiri2020federated}
M.~M. Amiri and D.~G{\"u}nd{\"u}z, ``Federated learning over wireless fading channels,'' \emph{IEEE Trans. Wireless Commun.}, vol.~19, no.~5, pp. 3546--3557, May. 2020.

\bibitem{wei2022federated}
X.~Wei and C.~Shen, ``Federated learning over noisy channels: Convergence analysis and design examples,'' \emph{IEEE Trans. Cogn. Commun. Netw.}, vol.~8, no.~2, pp. 1253--1268, Jun. 2022.

\bibitem{sery2021over}
T.~Sery, N.~Shlezinger, K.~Cohen, and Y.~C. Eldar, ``Over-the-air federated learning from heterogeneous data,'' \emph{IEEE Trans. Signal Process.}, vol.~69, pp. 3796--3811, Jun. 2021.

\bibitem{sun2023channel}
Y.~Sun, Z.~Lin, Y.~Mao, S.~Jin, and J.~Zhang, ``Channel and gradient-importance aware device scheduling for over-the-air federated learning,'' \emph{IEEE Trans. Wireless Commun.}, to appear.

\bibitem{yang2021achieving}
H.~Yang, M.~Fang, and J.~Liu, ``Achieving linear speedup with partial worker participation in non-iid federated learning,'' in \emph{Proc. Int. Conf. Learn. Represent. (ICLR)}, Virtual Event, May. 2021.

\bibitem{morrow1989bit}
R.~K. Morrow and J.~S. Lehnert, ``Bit-to-bit error dependence in slotted ds/ssma packet systems with random signature sequences,'' \emph{IEEE Trans. Commun.}, vol.~37, no.~10, pp. 1052--1061, Oct. 1989.

\bibitem{you2023broadband}
L.~You, X.~Zhao, R.~Cao, Y.~Shao, and L.~Fu, ``Broadband digital over-the-air computation for wireless federated edge learning,'' \emph{IEEE Trans. Mobile Comput.}, pp. 1--16, 2023.

\bibitem{khandekar2001complexity}
A.~Khandekar and R.~McEliece, ``On the complexity of reliable communication on the erasure channel,'' in \emph{Proc. IEEE Int. Symp. Info. Theory (ISIT)}, Washington, D.C., USA, Jun. 2001.

\bibitem{mackay1999good}
D.~J. MacKay, ``Good error-correcting codes based on very sparse matrices,'' \emph{IEEE Trans. Inf. Theory}, vol.~45, no.~2, pp. 399--431, Mar. 1999.

\bibitem{lu2004performance}
B.~Lu, G.~Yue, and X.~Wang, ``{Performance analysis and design optimization of LDPC-coded MIMO OFDM systems},'' \emph{IEEE Trans. Signal Process.}, vol.~52, no.~2, pp. 348--361, Feb. 2004.

\bibitem{sha2008multi}
J.~Sha, Z.~Wang, M.~Gao, and L.~Li, ``{Multi-Gb/s LDPC code design and implementation},'' \emph{IEEE Trans. Very Large Scale Integr. (VLSI) Syst.}, vol.~17, no.~2, pp. 262--268, Feb. 2009.

\bibitem{lecun1998mnist}
Y.~LeCun, ``{The MNIST database of handwritten digits},'' \emph{[Online]. Available: http://yann. lecun. com/exdb/mnist/}, 1998.

\bibitem{xiao2017fashion}
H.~Xiao, K.~Rasul, and R.~Vollgraf, ``{Fashion-MNIST: A novel image dataset for benchmarking machine learning algorithms},'' \emph{arXiv:1708.07747}, 2017.

\bibitem{krizhevsky2010cifar}
A.~Krizhevsky, V.~Nair, and G.~Hinton, ``{Cifar-10 (Canadian institute for advanced research)},'' \emph{[Online]. Available: http://www. cs. toronto. edu/kriz/cifar. html}, vol.~5, no.~4, p.~1, 2010.

\bibitem{li2020federated}
T.~Li, A.~K. Sahu, M.~Zaheer, M.~Sanjabi, A.~Talwalkar, and V.~Smith, ``Federated optimization in heterogeneous networks,'' \emph{Proc. Mach. Learn. Res.}, vol.~2, pp. 429--450, 2020.

\bibitem{han2015deep}
S.~Han, H.~Mao, and W.~J. Dally, ``Deep compression: Compressing deep neural networks with pruning, trained quantization and huffman coding,'' in \emph{Proc. Int. Conf. Learn. Represent. (ICLR)}, San Juan, Puerto Rico, May. 2016.

\bibitem{graves2011practical}
A.~Graves, ``Practical variational inference for neural networks,'' in \emph{Proc. Adv. Neural Inf. Process. Syst. (NeurIPS)}, Granada, Spain, Dec. 2011.

\bibitem{hochreiter1994simplifying}
S.~Hochreiter and J.~Schmidhuber, ``Simplifying neural nets by discovering flat minima,'' in \emph{Proc. Adv. Neural Inf. Process. Syst. (NeurIPS)}, Denver, CO, USA, Nov. 1994.

\bibitem{neelakantan2015adding}
A.~Neelakantan, L.~Vilnis, Q.~V. Le, I.~Sutskever, L.~Kaiser, K.~Kurach, and J.~Martens, ``Adding gradient noise improves learning for very deep networks,'' \emph{arXiv:1511.06807}, 2015.

\end{thebibliography}
\bibliographystyle{IEEEtran}

\end{document}